\documentclass[11pt,letterpaper]{article} %

\usepackage{amsmath,amsthm,amssymb} %
\usepackage[american]{babel} %
\usepackage{booktabs} %
\usepackage{dsfont} %
\usepackage{enumitem} %
\usepackage{float} %
\usepackage[margin=1in]{geometry} %
\usepackage{graphicx} %
\usepackage{hyperref} %
\usepackage[utf8]{inputenc} %
\usepackage{mdframed} %
\usepackage{microtype} %
\usepackage{pxfonts} %
\usepackage{soul} %
\usepackage{subdepth} %
\usepackage{suffix} %
\usepackage{tikz} %
\usepackage{xcolor} %
\usepackage{xspace} %
\usepackage[scr=rsfso]{mathalfa} %

\hypersetup{pagebackref, colorlinks=true, urlcolor=blue,
  linkcolor=blue, citecolor=magenta}

\usetikzlibrary{calc,positioning}

\mdfdefinestyle{figstyle}{ %
  linecolor=black!7, %
  backgroundcolor=black!7, %
  innertopmargin=10pt, %
  innerleftmargin=25pt, %
  innerrightmargin=25pt, %
  innerbottommargin=10pt %
}

\definecolor{White}{rgb}{1,1,1} %
\definecolor{Black}{rgb}{0,0,0} %
\definecolor{LightGray}{rgb}{.8,.8,.8} %
\colorlet{ChannelColor}{LightGray} %
\colorlet{ChannelTextColor}{Black} %
\colorlet{ReadoutColor}{White} %

\newcommand{\tnote}[1]{}
\newcommand{\hnote}[1]{}
\newcommand{\znote}[1]{}

\renewcommand{\Game}{\mathcal{G}}

\newcommand{\cV}{\mathcal V}
\newcommand{\cQ}{\mathcal Q}
\newcommand{\cA}{\mathcal A}
\newcommand{\strat}{\mathcal S}

\newcommand{\cL}{\mathcal L} %
\newcommand{\cG}{\mathcal G} %
\newcommand{\cH}{\mathcal H} %
\newcommand{\Fp}{\mathbb{F}}
\newcommand{\N}{\mathbb{N}}

\newcommand{\sA}{{\mathsf{A}}}
\newcommand{\sB}{{\mathsf{B}}}
\newcommand{\sC}{{\mathsf{C}}}
\newcommand{\sE}{{\mathsf{E}}}
\newcommand{\sF}{{\mathsf{F}}}

\newcommand{\sM}{{\mathsf{M}}}

\newcommand{\sP}{{\mathsf{P}}}
\newcommand{\sR}{{\mathsf{R}}}
\newcommand{\sS}{{\mathsf{S}}}
\newcommand{\sV}{{\mathsf{V}}}
\newcommand{\sX}{{\mathsf{X}}}

\renewcommand{\hat}[1]{\widehat{#1}} %
\renewcommand{\tilde}[1]{\widetilde{#1}} %

\newcommand{\Id}{\mathds{1}}
\newcommand{\comp}[1]{\overline{#1}}
\newcommand{\Protocol}{\mathcal{P}}
\newcommand{\reg}[1]{{\textsf{#1}}}

\newcommand{\ctl}{\text{CTL-}}
\newcommand{\tgt}{\text{TGT-}}
\newcommand{\CKT}{\textsc{CKT}}

\newcommand{\compr}{\sharp}

\newcommand{\TMSIM}{\textsc{TMSIM}}
\newcommand{\DESC}{\textsc{DESC}}
\newcommand{\CKTSIM}{\textsc{CKTSIM}}

\newcommand{\type}{\mathsf{type}}
\newcommand{\wire}{\mathsf{wire}}

\newcommand{\complex}{\mathbb{C}} %
\let\epsilon=\varepsilon %
\newcommand{\eps}{\epsilon} %

\newcommand{\ip}[2]{ \langle #1 \microspace | \microspace #2 \rangle}

\newcommand{\microspace}{\mspace{.5mu}} %
\newcommand{\ket}[1]{\ensuremath{\lvert\microspace #1
    \microspace\rangle}} %
\newcommand{\bra}[1]{\ensuremath{\langle\microspace #1
    \microspace\rvert}} %
\newcommand{\ketbra}[2]{\lvert #1 \rangle \! \langle #2 \rvert} %
\newcommand{\paren}[1]{(#1)}
\newcommand{\Paren}[1]{\left(#1\right)}

\newcommand{\Brac}[1]{\left[#1\right]}

\newcommand{\Abs}[1]{\left\lvert#1\right\rvert}

\newcommand{\set}[1]{\{#1\}}
\newcommand{\Set}[1]{\left\{#1\right\}}

\newcommand{\norm}[1]{\lVert#1\rVert}
\newcommand{\Norm}[1]{\left\lVert#1\right\rVert}

\newcommand{\class}[1]{\textup{#1}\xspace} %

\newcommand{\EXP}{\class{EXP}} %
\newcommand{\NEXP}{\class{NEXP}} %
\newcommand{\QMA}{\class{QMA}} %
\newcommand{\QMIP}{\class{QMIP}} %
\WithSuffix\newcommand\QMIP*{\ensuremath{\class{QMIP}^*}} %
\newcommand{\PSPACE}{\class{PSPACE}} %
\newcommand{\MIP}{\class{MIP}} %
\WithSuffix\newcommand\MIP*{\ensuremath{\class{MIP}^*}} %
\newcommand{\QIP}{\class{QIP}} %

\newcommand{\NTIME}{\class{NTIME}}

\newtheorem{theorem}{Theorem}[section] %
\newtheorem{lemma}[theorem]{Lemma} %
\newtheorem{proposition}[theorem]{Proposition} %
\newtheorem{corollary}[theorem]{Corollary} %
\newtheorem{definition}[theorem]{Definition} %
\newtheorem{conjecture}[theorem]{Conjecture} %
\newtheorem{claim}[theorem]{Claim} %

\newcommand{\E}{\mathop{\mathbb{E}}\displaylimits} %
\newcommand{\ind}{\mathrm{ind}}

\newcommand{\setft}[1]{\mathrm{#1}} %
\newcommand{\Density}{\setft{D}} %
\newcommand{\Lin}{\setft{L}} %

\DeclareMathOperator{\Tr}{Tr}

\DeclareMathOperator{\poly}{poly}

\DeclareMathOperator{\supp}{supp}

\title{Quantum proof systems for iterated exponential time, \\ and beyond}

\date{}
\author{Joseph Fitzsimons\thanks{Singapore University of Technology and Design, 8 Somapah Road, Singapore 487372 and Centre for Quantum Technologies, National University of Singapore, 3 Science Drive 2, Singapore 117543.} \and Zhengfeng Ji\thanks{Centre for Quantum Software and Information, School of Software, Faculty of Engineering and Information Technology, University of Technology Sydney, NSW, Australia} \and Thomas Vidick\thanks{California Institute of Technology, USA.} \and Henry Yuen\thanks{UC Berkeley, USA and University of Toronto, Canada.}}

\begin{document}

\maketitle

\begin{abstract}
We show that any language in nondeterministic time $\exp(\exp(\cdots\exp(n)))$, where the number of iterated exponentials is an arbitrary function $R(n)$, can  be decided by a multiprover interactive proof system with a classical polynomial-time verifier and a constant number of quantum entangled provers, with completeness $1$ and soundness $1 - \exp(-C\exp(\cdots\exp(n)))$, where the number of iterated exponentials is $R(n)-1$ and $C>0$ is a universal constant. The result was previously known for $R=1$ and $R=2$; we obtain it for any time-constructible function $R$.

The result is based on a compression technique for interactive proof systems with entangled provers that significantly simplifies and strengthens a protocol compression result of Ji (STOC'17). As a separate consequence of this technique we obtain a different proof of Slofstra's recent result (unpublished) on the uncomputability of the entangled value of multiprover games.

Finally, we show that even minor improvements to our compression result would yield remarkable consequences in computational complexity theory and the foundations of quantum mechanics: first, it would imply that the class $\MIP^*$ contains all computable languages; second, it would provide a negative resolution to a multipartite version of Tsirelson's problem on the relation between the commuting operator and tensor product models for quantum correlations.
\end{abstract}

\newpage

\section{Introduction}

The combined study of interactive proof systems and quantum
entanglement has led to multiple discoveries at the intersection of
theoretical computer science and quantum physics.
On the one hand, the study has revealed that quantum entanglement, a
fundamental physical phenomenon, can be harnessed in interactive
protocols to accomplish an array of novel computing and cryptographic
tasks, ranging from the certified generation of random numbers to
improved protocols for multi-party cryptography and
classically-verifiable quantum computation.
On the other hand, interactive proof systems, a cornerstone of modern
complexity theory and cryptography, have provided a powerful lens
through which to examine the counter-intuitive properties of quantum
entanglement.
This lens has enabled researchers to develop sophisticated ways of
exploring phenomena such as the monogamy of entanglement, embezzlement
of quantum states, and more.

We investigate a central question in this area: what is the
\emph{computational complexity} of interactive proof systems with
multiple quantum entangled provers?
The starting point for this question dates back to the seminal result
of Babai, Fortnow and Lund, who showed that the set of languages that
can be decided by a (classical) multiprover interactive proof system,
denoted by $\MIP$, equals the set of languages that can be decided
in nondeterministic exponential time (denoted by
$\NEXP$)~\cite{BabForLun91CC}.
It is not difficult to show that $\MIP \subseteq \NEXP$, but the reverse
containment is nontrivial and the work of~\cite{BabForLun91CC} was an
influential stepping stone towards the PCP
Theorem~\cite{AroSaf98JACM,AroLunMotSudSze98JACM}.

A long line of work, starting with that of Cleve et
al.~\cite{cleve2004consequences}, has explored the setting of
interactive proof systems where a classical polynomial-time verifier
interacts with provers that are \emph{quantum} and may share \emph{entanglement}.
This gives rise to the complexity class $\MIP*$, which is the set of
all languages decidable by such proof systems.\footnote{The $^*$ in
  $\MIP^*$ refers to the entanglement.}
Quantum entanglement is a resource that allows isolated parties to
generate correlations that cannot be reproduced by (classical) shared
randomness alone; however, entanglement does not allow for
instantaneous communication.
A central question raised by~\cite{cleve2004consequences} is whether
$\MIP^* = \MIP$, or equivalently, whether $\MIP^* = \NEXP$.

A richer set of correlations gives additional power to provers in
an interactive proof system, making the relationship between $\MIP*$
and $\MIP$ non-obvious.
On the one hand, a multiprover interactive proof system that is sound
against ``cheating'' classical provers may no longer be sound against
``cheating'' entangled provers; this prevents one from automatically
concluding that $\MIP \subseteq \MIP*$.
On the other hand, a proof system may require ``honest provers'' to
use quantum entanglement in order to satisfy the completeness
property.
Entanglement thus allows one to consider a
broader set of protocols, putting in question the inclusion $\MIP* \subseteq
\MIP$.

The quest to pin down the computational power of proof systems with
entangled provers has led to a number of surprising discoveries.
The best lower bound that is currently known is that $\NEXP = \MIP
\subseteq \MIP*$, a nontrivial result that follows from a more general
technique of ``immunization'' of classical proof systems against
malicious entangled provers~\cite{IV12,NatarajanV17twoprover}.
Surprisingly, there are no meaningful upper bounds known for $\MIP*$.
In a striking result, Slofstra gave evidence that the complexity of
$\MIP*$ might be very different from its classical counterpart: he
proved that it is \emph{undecidable} to determine whether an
interactive proof system with two provers has an entangled strategy
that is accepted with probability $1$ (in other words, whether there
is a \emph{perfect} entangled
strategy)~\cite{slofstra2016tsirelson,slofstra2017set}.
In contrast, the complexity of determining whether such a proof system
has a perfect \emph{classical} strategy is exactly equal to
$\NEXP$.
Another recent result of 
Ji~\cite{ji2017compression} points in the same direction: Ji showed that any language in non-deterministic doubly-exponential time can be decided by a classical polynomial-time verifier interacting with $k=11$ provers, with completeness $1$ and soundness that is exponentially close to $1$.\footnote{Due to the vanishing gaps neither Slofstra's nor Ji's result directly separates $\MIP^*$ from $\MIP$, though they do separate the zero-error and exponentially-small error variants respectively.}

In this work we explore the expanse of complexity-space that
entangled-prover interactive proof systems can reach.
We focus on the ``small gap'' regime: we consider the problem of
distinguishing between the cases when a multiprover proof system has a
perfect entangled strategy, or when all entangled provers are rejected
with probability at least $\eps$, where $\eps$ is a quantity that may
go to $0$ quickly with the size of the verifier in the proof system.
Our results smoothly interpolate between the hardness result
of~\cite{IV12,NatarajanV17twoprover,ji2017compression} and Slofstra's undecidability
result.
For clarity we restrict our attention to \emph{hyper-exponential} time
functions, i.e.
time-constructible functions of the form $t(n) = \Lambda_R(n)$, where
$\Lambda_0(n)=n$ and for any integer-valued function $R=R(n)\geq 0$,
$\Lambda_{R+1}(n)=2^{\Lambda_R(n)}$. For a multiprover game $\Game$, the \emph{entangled value} $\omega^*(\Game)$ is the maximum success probability of quantum provers sharing entanglement in the game. 

\begin{theorem}\label{thm:main}
  Let $k \geq 15$ be an integer.
  Let $t:\N \to \N$ be a hyper-exponential function.
  There are universal constants $C,c>0$ such that given the description
  of polynomial-size circuits for the verifier in a $k$-prover
  game $\Game$, the problem of distinguishing between
  \[
    \omega^*(\Game) = 1 \qquad \text{or} \qquad \omega^*(\Game) \leq 1 -
    \frac{C}{(t(n))^c}\;
  \]
is hard for nondeterministic $2^{t(n)}$
  time.
\end{theorem}

The ``base case'' for Theorem~\ref{thm:main}, corresponding to $R=0$
and $t(n)=n$, is the result that
$\NEXP\subseteq\MIP^*$~\cite{IV12,NatarajanV17twoprover}, where
$\MIP^*$ is the class of languages that can be decided using an
entangled-prover interactive proof system, with completeness
$\frac{2}{3}$ and soundness $\frac{1}{3}$ (the completeness-soundness
gap can be amplified from inverse polynomial to constant using
hardness amplification
techniques~\cite{bavarian2017hardness}).
The first step, $R=1$ and $t(n)=2^n$, follows from Ji's result~\cite{ji2017compression} mentioned earlier, albeit using a game with $k=11$ provers.

A corollary of both our and Ji's earlier result is that the ``honest strategy'' for the provers
(i.e.
those satisfying the completeness property) in the games constructed through the reduction from Theorem~\ref{thm:main}
provably {require} the provers to share entanglement.
Moreover, it is often possible to obtain lower bounds on the dimension
of entanglement required to achieve close to optimal success
probability; this is the case for our result, as described below.

The proof of Theorem~\ref{thm:main} is based on a compression
technique that significantly simplifies and extends the approach
pioneered in~\cite{ji2017compression}.
Our generalized compression result can be recursively composed with
itself in order to obtain the statement of Theorem~\ref{thm:main} for
any integer-valued $R(n)\geq 1$. 

The starting point of the compression approach
of~\cite{ji2017compression} is to extend the notion of a \emph{history
  state}.
The concept of a history state was first introduced by Kitaev in order
to efficiently encode any polynomial-time quantum computation as the ground
state of a local Hamiltonian, in a way that is also efficiently
verifiable~\cite{kitaev2002classical}.
The compression result of~\cite{ji2017compression} as well as the one
in this paper constructs a game to verify history states that encode
the execution of a (different) multiprover game, including the actions
of the provers (which in general are not efficiently computable).
The verification is performed by executing a ``games'' version of the
traditional verification procedure for history states, that consists
in randomly sampling a local Hamiltonian term and measuring its
energy.

There are two key ideas behind our generalized compression technique.
The first is to ensure that the game $\Game$ that verifies the history
state of a multiprover game $\Game'$ can be executed using a circuit
that is logarithmic in the size of $\Game'$, provided that $\Game'$ is
specified in a sufficiently uniform and succinct manner.
The second idea is to compose the first idea with itself, i.e.
consider the history state for the computation performed by the
history state verification procedure.
At this point there are a number of delicate issues to consider,
including identifying the right model for specifying verifiers,
verifiers of verifiers, etc.; we give more details in
Section~\ref{sec:overview}.

On a more informal note, we observe that the kind of compression
achieved here may be thought of as a ``bootstrapping'' of Kitaev's
history state technique, in a similar sense to the composition
technique from the PCP literature that ``bootstraps'' an efficient PCP
into a super-efficient one.\footnote{The analogy only goes so far:
  composition in PCPs reduces the answer size; here, we reduce the
  query size.}
The fact that history states are ground states of local Hamiltonians
is a statement about the local verifiability of arbitrary quantum
computation.
Our result goes further by making the following observations.
First, not only is the verification procedure local, it is also
exceedingly efficient --- it can be executed in time logarithmic in
the size of the original computation.
Second, it is possible to consider a history state for the
verification procedure itself.
Third, and most strikingly, the latter history state can be verified
with the same complexity as the verification procedure, without
reference to the size of the original computation.
This last step crucially relies on \emph{rigidity} properties of
entanglement which acts as a ``leash'' on quantum systems.
It is sufficient to only control the leash-holder: if the leash-holder
manages to hold the dog tightly enough, then there is no longer any
reason to worry about the (hyper-exponential-size) dog itself.

It is worth noting that such ``PCP composition on steroids'' has no
classical analogue.
A classical PCP verifier runs in polynomial time and uses
polynomially many random bits to verify an exponentially long proof.
Encoding the computation performed by such a verifier in a way that
can be verified using, say, a classical multiprover interactive proof
system, again requires a polynomial-sized verifier flipping
polynomially many bits.
This is because the only way to ``verify the verification procedure''
is to, at least with some probability, access some of the original
proof bits.
In the quantum case, it is possible to leverage entanglement between
provers to avoid the need for the ``inner'' verifier (to borrow some
terminology from the PCP literature) to make any query at all to the
original proof qubits.

\medskip \vspace{15pt}

Before proceeding we formulate another consequence of compression that
highlights the versatility of our approach.
As already mentioned, it was recently shown by Slofstra that the
problem of determining whether a given multiprover game has a perfect
entangled strategy is undecidable.
Slofstra's result proceeds by an ingenious (and intricate) reduction
to the word problem in finitely presented groups, which is known to be
undecidable.
The proof of the latter itself involves a sophisticated embedding of
the computation of an arbitrary Turing Machine (in fact, a Minsky
machine) in an instance of the word problem in a suitable finitely
presented
group~\cite{novikov1955on,boone1958the,karlampovivc1982finitely}.

We give a different proof of Slofstra's undecidability
result, by directly constructing an interactive proof system from a
Turing machine.
Arguably, our result provides an intuitive reason for \emph{why} the
problem is undecidable, showing in a precise sense how smaller and
smaller gaps can be leveraged to verify that the provers are
performing an increasingly complex computation.
More precisely, the main idea for our proof is to design a family of
games $\{\Game_n\}_{n\geq 1}$ such that for any $n\geq 1$ the verifier in the game $\Game_n$ verifies
if a Turing machine provided as input halts within $n$ steps, and if
it does not, executes a game with the provers that verifies that,
either the provers hold a quantum proof that the Turing machine halts
within $2^n$ steps, or they hold a history state for the verification of a
quantum proof that either the Turing machine halts within $2^{2^n}$ steps,
or...
Somewhat more formally, we obtain the following (see
Theorem~\ref{thm:undecidable} for a more complete statement).

\begin{theorem}\label{thm:main-undecidable}
  For all deterministic Turing machines $M$, there exists a
  multiprover game $\Game_M$ (that can be computed from the
  description of $M$) such that if $M$ halts in finite time then
  $\omega^*(\Game_M)<1$, whereas if $M$ does not halt then
  $\omega^*(\Game_M) = 1$.
  Furthermore, there exists a universal constant $\eta >
  0$ such that for any non-halting $M$, any strategy for the provers that succeeds with
  probability at least $1 - \eps$ in $\Game_M$, for some $\eps\geq 0$, 
  requires the use of an entangled state of local dimension at least
  $2^{\Omega(\eps^{-\eta})}$.
\end{theorem}

The game $\Game_M$ in Theorem~\ref{thm:main-undecidable} is a game
with $15$ provers that can be efficiently computed from $M$; the
undecidability result follows immediately.
In addition, as stated in the theorem our game can be used as a form
of dimension test for the strategies of the provers.
Up to the value of the constant $\eta$ the bound
$2^{\Omega(\eps^{-\eta})}$ matches the best bound known, for a
three-prover game considered in~\cite{ji2018three}.

\subsection{Proof overview}
\label{sec:overview}

We provide a detailed overview for the proof of
Theorem~\ref{thm:main}.
In Section~\ref{sec:intro-compression} we sketch our main
``compression'' result and expand on the compression technique from~\cite{ji2017compression}.
The following sections sketch the proof of the compression theorem.
We start by describing a method to succinctly describe the actions of
a verifier in a multiprover game in Section~\ref{sec:intro-gtm}.
In Section~\ref{sec:intro-history} we describe the main steps of the
proof: (1) design a history state associated with the execution of a
multiprover game, (2) design a game that verifies the history state
with the help of an additional trusted prover, and finally (3) design
a game in which the honest prover has been merged into existing
provers.
This last step, prover merging, is described in more detail in
Section~\ref{sec:intro-merge}. In Section~\ref{sec:rec-comp} we sketch how the compression theorem can be applied recursively to show Theorem~\ref{thm:main} and
Theorem~\ref{thm:main-undecidable}.

\subsubsection{Protocol compression}
\label{sec:intro-compression}

The main workhorse of this paper is a compression theorem for quantum
multiprover interactive protocols that simplifies and strengthens the
compression result of~\cite{ji2017compression}.
To state the result, we first review the notion of $k$-prover
``extended nonlocal (ENL) game'', which is a type of quantum
multiprover game introduced in~\cite{johnston2016extended}.
A $k$-prover ENL game is a three-turn interaction between a quantum
verifier and $k$ quantum provers sharing entanglement.
The game (or ``protocol'') proceeds in three stages.
First, the provers send a quantum register $\reg{C}$ to the verifier.
Second, the verifier measures the register $\reg{C}$ to obtain an
outcome $t$.~\footnote{Our definition of ENL game is slightly more
  general than that in~\cite{johnston2016extended}, where the sampling
  of questions is classical and does not depend on $\reg{C}$.}
The verifier then computes a classical query $Q=(q_1,\ldots,q_k)$ that
it distributes to the provers.
Third, the provers respond with classical answers $a=(a_1,\ldots,a_k$)
to their respective questions.
In general, each prover's answer is determined by performing a
measurement on the prover's share of a quantum state that may be entangled with $\reg{C}$.
Finally, the verifier makes an accept/reject decision based on the
outcome $t$, its internal randomness, and the provers' answers.
The maximum acceptance probability of an ENL game $\Game$ is denoted
$\omega^*(\Game)$, and is also called the (entangled) \emph{value} of
$\Game$.

The whole interaction between verifier and provers in an ENL game can
be represented as a quantum circuit of a special form that we call a
\emph{protocol circuit}, as depicted in Figure~\ref{fig:qip}.
A protocol circuit starts with the application of a quantum circuit
$C_Q$ on registers $\reg{C}$ (which holds the provers' first message),
$\reg{V}$ (the verifier's private workspace), and $\reg{M}$ (which
holds the messages exchanged between the verifier and provers).
The circuit $C_Q$ implements the verifier's measurement on register
$\reg{C}$, and the verifier's choice of questions to the provers.
The circuit $C_Q$ is followed by an arbitrary unitary transformation
for each prover $i$, applied on the component $\reg{M}_i$ of the
message register that the prover has access to, as well as its private
workspace $\reg{P}_i$ (that contains the prover's part of shared
entangled state).
Finally, the last step in the protocol circuit is the application of a
circuit $C_A$ that acts on $\reg{C}$, $\reg{V}$ and $\reg{M}$ and
computes the verifier's decision in the game, that is written on a
specially designated ``output qubit''.

\begin{figure}[H]
  \begin{mdframed}[style=figstyle]
    \begin{center}
      \begin{tikzpicture}[scale=1, biggate/.style={draw, minimum
          height = 2cm, minimum width = 1.3cm, fill=ChannelColor}]

        \node (C) at (-2,1) {$\sC$}; \node (V) at (-2,.5)
        {$\phantom{\sC}$}; \node (V2) at (-2,0) {$\phantom{\sC}$};
        \node (V3) at (-2,.25) {$\sV$}; \node (M) at (-2,-.5) {$\sM$};
        \node (P) at (-2,-1.8) {$\phantom{\sC}$}; \node (P2) at
        (-2,-1.3) {$\phantom{\sC}$}; \node (P3) at (-2,-1.55) {$\sP$};

        \node (OutC) at (4,1) {}; \node (OutV) at (4,.5) {}; \node
        (OutV2) at (4,0) {}; \node (OutM) at (4,-.5) {}; \node (OutP)
        at (4,-1.8) {}; \node (OutP2) at (4,-1.3) {};

        \draw (C.east)--(OutC); \draw (V.east)--(OutV); \draw
        (V2.east)--(OutV2); \draw (M.east)--(OutM); \draw
        (P.east)--(OutP); \draw (P2.east)--(OutP2);

        \node[biggate] (CQ) at (-.8,.25) {$C_Q$}; \node[biggate] (PR)
        at (1,-1.25) {$P$}; \node[biggate] (CA) at (2.8,.25) {$C_A$};
      \end{tikzpicture}
    \end{center}
  \end{mdframed}
  \caption{The protocol circuit of an extended nonlocal game.}
  \label{fig:qip}
\end{figure}
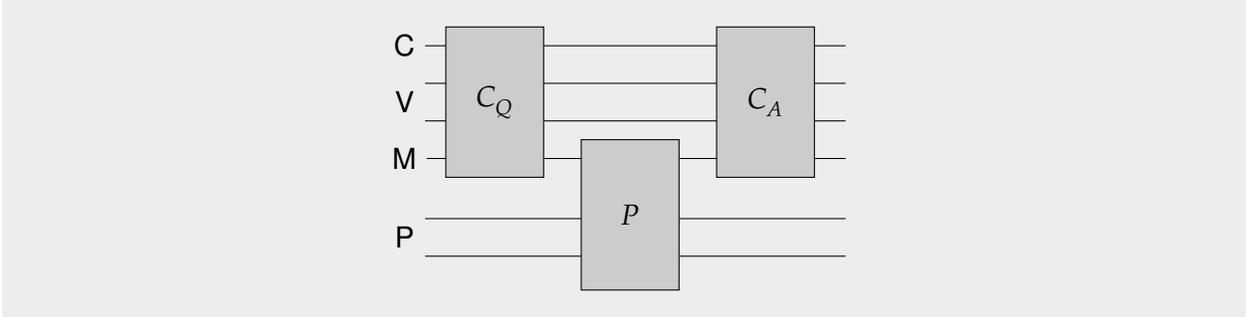

The compression theorem applies to families of ENL games $
\{\Game_N\}$ that have \emph{succinct descriptions}.
By this we mean, not only that the protocol circuit associated with
$\Game_N$ has size polynomial in $N$, but moreover there exists a
deterministic Turing machine $G$ (called a \emph{Gate Turing Machine} (GTM)) that on input $(N,t)$, where
$N$ and $t$ are two integers written in binary, runs in polynomial time and returns the description of the $t$-th gate of the protocol
circuit associated with $\Game_N$ (and a special symbol if $t$ is
larger than the circuit size).
If the $t$-th gate is an action of the prover, the GTM
returns another special symbol. 

\begin{theorem}[Compression Theorem]\label{thm:main-compression}
  Let $k\geq 7$ be an integer and let $ \{ \Game_N \}$ be a succinctly
  described family of $k$-prover ENL games with GTM $G$.
  Then there exists a family of $k$-prover ENL
  games $ \set{ \Game^{\compr}_n}$ such that for all
  integer $n\geq 1$ and $N=2^n$, it holds that
  \begin{equation}\label{eq:compression-00}
	\omega^*(\Game^{\compr}_n) \,\leq\, 1 - \frac{(1 -
      \omega^*(\Game_N))^\alpha}{\poly(N)}\;,
  \end{equation}
  where $\alpha \geq 1$ is a universal constant, and if $\omega^*(\Game_N) = 1$ then we have $\omega^*(\Game^\compr_n) = 1$. Moreover, there exists a Turing machine $A^\compr$ that on input $(1^n,G)$ returns the description of $\Game^\compr_n$ in polynomial time. 
\end{theorem}

The strength of the theorem lies in the exponential reduction in the
size of the verifiers of the ENL game, from $\poly(N)$ (the size of $\Game_N$) to $\poly(n) = \poly(\log N)$ (the size of $\Game^{\compr}_n$). 
The cost of this exponential compression of game size is that the
value of the game gets ``compressed'' towards $1$; nevertheless, games
with value $1$ (resp.
$<1$) are compressed to games with value $1$ (resp.
$<1$).
Theorem~\ref{thm:main-compression} differs from the results
of~\cite{ji2017compression} in two significant ways.
First, the compression result in~\cite{ji2017compression} does not
yield a family $\{\Game^{\compr}_n\}$ that is as efficiently described as the games
returned by our reduction.\footnote{Although the question lengths of the ``compressed'' game in~\cite{ji2017compression} are $O(\log N)$, the verifier itself has size $\poly(N)$. The verifier for the game $\Game^\compr_n$, in contrast, has size $\poly(\log N)$.}
The recourse to succinct descriptions via Gate Turing Machines is an
essential ingredient for the recursive application of
Theorem~\ref{thm:main-compression}.
Second, the compression result in~\cite{ji2017compression} increases
the number of provers, from $k$ to $k+8$.
Our result does not require the use of additional provers; this is
again essential in allowing a large (or even infinite) number of recursive applications
of the theorem.

In the following subsections we sketch the proof of
Theorem~\ref{thm:main-compression}.
The first step is to make the notion of ``succinctly described'' more
concrete.

\subsubsection{Succinct descriptions of verifiers}
\label{sec:intro-gtm}
In the study of quantum interactive proof systems, families of games $\{\Game_N\}$ are usually presented as a uniformly generated family of circuits for the verifier: there exists a 
polynomial-time deterministic Turing machine $A$ that on input $1^N$ returns a
circuit description of the verifier in $\Game_N$. However, such uniform descriptions of verifier circuits are insufficient for our compression result: from a game $\Game_N$ we aim to
design a ``compressed game'' $\Game^{\compr}_n$ that has size $\poly(n)$,
exponentially smaller than the size of $\Game_N$.
In particular, $\Game^{\compr}_n$ does not have nearly enough time to run $A$
to get a circuit description of the verifier of $\Game_N$.
What we need is that the verifier of $\Game^{\compr}_n$ be granted some form of
\emph{implicit} description of the verifier of $\Game_N$. %

We achieve this via the notion of a \emph{Gate Turing
  Machine} (GTM) for a family of ENL games $\Set{\Game_N}$.
As mentioned before, it is a Turing machine $G$ that on input $(N,t)$ outputs in $\poly\log(N)$ time the
description of the $t$-th gate of the protocol circuit of $\Game_N$ (which has size
$\poly(N)$).

Thus, our notion of ``succinct description'' for a family of ENL games
$\set{\Game_N}$ is that there is a GTM $G$ for the family.
With this notion in place,
it remains to show the compression theorem: any succinctly described
family of games $\set{\Game_N}$ can be ``compressed'' to another family of ENL games $\set{\Game_n^{\compr}}$ with the
properties described in Theorem~\ref{thm:main-compression}.
We sketch how this is done in the next sections.

\subsubsection{Testing history states of protocol circuits}
\label{sec:intro-history}

With the appropriate notion of succinct description in place, we
describe the three main steps that go into the proof of
Theorem~\ref{thm:main-compression}.

The first step consists in considering the history state
$\ket{\Psi_\Game(N)}$ of the protocol circuit (Figure~\ref{fig:qip})
associated with an execution of $\Game=\Game_N$, where $N=2^n$.
This state is defined on the registers $\reg{CVMP}$, and may be
extremely large, depending on the size of the provers' registers.
In addition, the state has a component on a clock register
$\reg{C}_{outer}$ of the same dimension as the total number of gates
$\tau_N$ in the protocol circuit, which is polynomial in $N$; thus the
register $\reg{C}_{outer}$ is over $O(n)$ qubits.
Concretely, the state $\ket{\Psi_\Game(N)}$ has the form
\begin{equation}\label{eq:def-psig}
  \ket{\Psi_\Game(N)} \,=\, \frac{1}{\sqrt{\tau_N + 1}} \sum_{t=0}^{\tau_N}
  \ket{t}_{\reg{C}_{outer}}\otimes U_t\cdots U_1
  \ket{\psi_\Game(0)}_{\reg{CV}\reg{M}\reg{P}}\;.
\end{equation}
Here $\ket{\psi_\Game(0)}$ is the initial state of the verifier and the
provers' registers in $G$, with $\reg{C}$ denoting the initial
register received from the provers, $\reg{V}$ the private workspace
for the verifier, $\reg{M}=\reg{M}_1,\ldots,\reg{M}_k$ the message
registers, and $\reg{P}=\reg{P}_1,\ldots,\reg{P}_k$ the private spaces for
the provers.

Note that in~\eqref{eq:def-psig}, almost all unitaries are gates
applied by the verifier, except $k$ of them, one for each prover, that
can be considered ``wild cards''.
The important property is that, if $\omega^*(\Game)=1$ then there exists a
state of the form~\eqref{eq:def-psig}, for some choice of
$\ket{\psi_{\Game}(0)}$, and some choice of unitaries to apply in the
``wildcard'' locations, that is a ground state (energy $0$) of the
local Hamiltonian $H_\Game(N)$ that verifies the history state (this
is entirely analogous to Kitaev's circuit-to-Hamiltonian construction, but for
the use of the prover gates which may induce large non-local Hamiltonian terms
to verify their propagation).
Conversely, if $\omega^*(\Game_N)=1$ then no such state exists,
irrespective of the choice of the ``wildcard'' unitaries.

The next step is to design an intermediate ENL game $\Game_H$ that has
one additional prover, called the ``Pauli Prover'' $PV$.
We call the verifier in $\Game_H$ the \emph{outer verifier}.
The goal of the outer verifier is to verify that the provers share the
state $\ket{\Psi_\Game(N)}$, where registers associated with the
verifier in $\Game$ (that we call the \emph{inner verifier}), i.e.
$\reg{C}$, $\reg{V}$ and $\reg{M}$, are given to $PV$, while the clock
register $\reg{C}_{outer}$ is the prover's first message in the ENL
game $\Game$.
As already mentioned, this initial message has length $O(n)$
qubits.

Informally, to achieve this verification task the outer verifier and
the Pauli Prover collaborate to implement a family of tests that are
game-like versions of the tests implemented by the local Hamiltonian
$H_\Game(N)$.
This includes an ``input check'' (the state $\ket{\psi_\Game(0)}$ is
well-formatted), a ``gate check'' (each time step corresponds to the
application of a unitary, and unitaries associated with the inner
verifier are the right ones, as specified in the circuits $C_Q$ and
$C_A$), and an ``output check'' (the final decision made by the inner
verifier is to accept).
Each of these checks involves not only the verifier and $PV$, but also
the other provers, that are required to apply their prover gate when
the corresponding propagation check is performed.

In designing $\Game_H$, we take advantage of the fact that the Pauli
Prover is considered ``honest'': it always implements the observable
that it is asked by the outer verifier.
However, for reasons that will soon become clear the Pauli Prover can
only be asked to implement single- or two-qubit Pauli
observables.\footnote{In fact, triples of commuting two-qubit
  observables; we gloss over this for purposes of this overview.}
This means that all tests performed by the outer verifier can only
require such observables on the registers $\reg{CVM}$.

The crucial point here is that the complexity of the verifier in the
game $\Game_H$ is exponentially smaller than the complexity of the
verifier in $\Game$.
The reason this is possible is that in order for the verifier in
$\Game_H$ to check that the entangled state shared by the provers is a
valid history state for the protocol circuit associated with $\Game$
it is enough to select a random time step in that circuit, and
implement the associated check.
Both of these can be performed in time $\poly\log(N)$; the first
trivially so, and the second thanks to our assumption that $\Game$ is
specified through a ``succinct description'', provided by the verifier
$\cV$ and GTM $G$ associated with $\set{\Game_N}$, as described in
Section~\ref{sec:intro-gtm}.

In the last step we convert the Single Pauli Prover game $\Game_H$
into a new ENL game $\Game^{\compr}=\Game^{\compr}_n$, with the same number of provers
as in the original ENL $\Game$, but with drastically reduced question
length --- it is now $O(n)$, when questions in $\Game$ might have been
$\poly(N)$ bits long.
For this we need to remove the ``honest'' assumption on $PV$, and
moreover we need to ``merge'' $PV$ with existing provers.
This step of prover merging is explained in the next subsection.

\subsubsection{Prover merging}
\label{sec:intro-merge}

Prover merging is performed in two steps.
The first step uses somewhat standard techniques, similar to those
employed in~\cite{ji2017compression}, that originate in the
self-testing literature.
The main idea is to require the honest Pauli prover $PV$ in
$\Protocol$ to implement the observable it is asked to measure
transversally, on an error-encoded version of his share of the state
(this is the main motivation for restricting the prover to Pauli
observables), and then to split $PV$ into as many provers as the
error-correcting code requires.
It is then possible, using self-testing technique, to test the
``split'' $PV$ so as to ensure that any deviation from the honest
actions is detected by the verifier.

The second step is the actual merging step.
This step is somewhat delicate: we take the split provers, and merge
them into existing provers from $\Game$.
Since each prover $P$ now simultaneously receives two questions --- its
question in $\Game$, as well as the share of the question to $PV$ that
would have been sent to the split prover that got merged into $P$ ---
soundness is non-obvious.

To show that this step does not compromise soundness, we leverage the
fact that, by construction, the prover that is to be merged only has
to perform very simple operations: Pauli $\sigma_X$ and $\sigma_Z$
observables, on a constant number of qubits at a time.
These kinds of operations can be tested, indeed ``commanded'', in a
very rigid way by using self-testing results.
Therefore, we can embed these actions into any prover.
It is then straightforward to enforce that a prover performs the right
action on a Pauli observable.
However, its action on the real question may depend on the Pauli
question.
To get around this we once again leverage the structure of the Pauli
Prover game as well as the quantum error-correcting code.
More details on this part are given in Section~\ref{sec:sim-pauli}.

\subsubsection{Recursive compression}
\label{sec:rec-comp}

Ultimately, we use our compression theorem (Theorem~\ref{thm:main-compression}) in a recursive fashion to prove Theorem~\ref{thm:main}. To illustrate the essential idea behind the recursive compression approach, we give an informal overview of the proof of the statement that any language computable in deterministic time $t(n)$ has a quantum interactive proof system with completeness-soundness gap that scales as an inverse polynomial in $t(n)$. 

Let $L$ be such a language. Then there exists a deterministic Turing machine $M$ that on input $x \in \{0,1\}^n$ decides whether $x \in L$ in time $t(n)$. For every $x \in \{0,1\}^n$ and integer $N \geq n$, we construct a verifier $\cV_{x,N}$ for a $7$-prover ENL game $\Game_{x,N}$ that does the following. The verifier first runs $M$ for $N$ steps on input $x$. If $M$ accepts in this time, then $\cV_{x,N}$ accepts. If $M$ rejects in this time, then $\cV_{x,N}$ rejects. Otherwise, $M$ has not halted. In this case $\cV_{x,N}$ executes a \emph{compressed} version of the protocol corresponding to $\cV_{x,2^N}$, which is an exponentially larger version of itself. This compressed protocol is provided by Theorem~\ref{thm:main-compression}. The recursion continues until at some point, $M$ is run for a large enough ``tower of exponential'' number of steps that exceeds $t(n)$, in which case $M$ either accepts or rejects input $x$. The following can then be shown by induction on $R$ such that $t(n)\leq \Lambda_R(n)$. If $x \in L$ then the value of the game $\Game_{x,t(n)}$ is $1$, and therefore for all $N \leq t(n)$ the value of $\Game_{x,N}$ is $1$, which implies that $\Game_{x,n}$ has value $1$. Otherwise, if $x \notin L$, then using Theorem~\ref{thm:main-compression} we obtain that the value of $\Game_{x,n}$ is at most $1 - \Omega \Paren{1/\poly(t(n))}$. 

This nearly shows the desired conclusion, except that Theorem~\ref{thm:main-compression} requires that the family of games to be compressed have a succinct description in the manner described in Section~\ref{sec:intro-gtm}. We thus need to argue that the family of games $\Set{ \Game_{x,n}}$ has a GTM $G$ associated with it. 
\emph{A priori} it is unclear whether the verifiers $\Set { \cV_{x,n}}$ are structured enough so that any particular gate of the verifier circuits can be specified in polylogarithmic time. However, we show that as long as the verifiers $\Set { \cV_{x,n}}$ are \emph{uniformly generated} (meaning that there is some polynomial time Turing machine $A$ that on input $(1^n,x)$ returns the description of the verifier circuits of $\cV_{x,n}$), there is an \emph{equivalent} family of verifiers $\Set{ \cV_{x,n}'}$ that has a \emph{succinct description}. We prove this fact in Section~\ref{sec:gtm}; the proof relies on a concept from classical complexity theory known as \emph{oblivious simulation} of Turing machines.
Since the family of verifiers $\Set{ \cV_{x,N}}$ is uniformly generated, we obtain that the verifiers have a succinct description via a GTM, which in turn allows us to apply the compression theorem as outlined above. 

Adapting this sketch to handle languages that are decided by \emph{nondeterministic} Turing machines (as needed in Theorem~\ref{thm:main}), as well as reproving Slofstra's undecidability result (Theorem~\ref{thm:main-undecidable}), requires additional care. We give details in Section~\ref{sec:recursive}.

\subsection{Improving the compression theorem?}

Theorem~\ref{thm:main-compression} offers the following tradeoff
between ``compression in size'' and ``compression of the gap'': the
former is scaled by an exponential factor, from polynomial in $N=2^n$ to polynomial in $n$, while
the latter is divided by a quantity that is polynomial in $N$, or
equivalently, exponential in $n$. 

Surprisingly, we show that \emph{any} better tradeoff, i.e.
one in which the gap gets reduced by a subexponential factor in
$n$, would have far-reaching consequences in complexity
theory and mathematics. The result provides a possible explanation for the absence of meaningful upper bounds on $\MIP^*$ (provided an improved compression result does hold): not only would every computable language be decided by an $\MIP^*$ proof system, there would even be \emph{undecidable} languages in $\MIP^*$.

\begin{theorem}[Consequences of an improved compression theorem]
\label{thm:hct-informal}
	Suppose an analogue of Theorem~\ref{thm:main-compression} holds, such that the factor $\poly(N)$ in the denominator on the right-hand side of~\eqref{eq:compression-00} is replaced by a subexponential function of $n = \log N$. Then
	\begin{enumerate}
		\item $\MIP^*$ with constant gap contains all computable languages. %
		\item $\MIP^*$ with constant gap contains undecidable languages.
		\item The commuting operator model of multipartite correlations is strictly more powerful than the tensor product model. 
	\end{enumerate}
\end{theorem}

We precisely define what we mean by ``improved compression theorem'' in Section~\ref{sec:improving} (see Conjecture~\ref{conj:hct}). The idea behind the proof of Theorem~\ref{thm:hct-informal} is that the tradeoff between a subexponential compression in gap and an exponential reduction in size can be ``boosted'' to a tradeoff where the gap does not get compressed at all, but the game size still gets compressed by a nontrivial amount. This uses \emph{hardness amplification} techniques for multiprover entangled games~\cite{bavarian2017hardness}, which employs a variant of parallel repetition to achieve this boosting. 

We briefly explain what we mean by the third item in Theorem~\ref{thm:hct-informal}, and refer to the end of Section~\ref{sec:cons2} for an expanded discussion. In this paper, we define the entangled value of a nonlocal game as the supremum of the success probabilities over all ``tensor product'' strategies for the provers, which consist of a finite-dimensional Hilbert space for each prover, an entangled state in the tensor product of those Hilbert spaces, and a collection of measurement operators on each prover's space. %

There is an alternate definition of the entangled value, which considers the supremum over so-called ``commuting operator'' strategies, for which there is a single (possibly infinite-dimensional) Hilbert space shared by all players, and the only restriction is that measurement operators applied by distinct provers commute with each other.
Since tensor product strategies are also commuting operator strategies, the entangled value in the tensor product model is at most the entangled value in the commuting operator model. 
It is known that in the finite dimensional case, the two models are equivalent. Whether they coincide in general is a famous problem in quantum information known as ``Tsirelson's problem'' (see e.g.~\cite{fritz2012tsirelson}). 

As we explain in Section~\ref{sec:improving} (and is well known to experts, though we could not find an explicit reference), a positive resolution to Tsirelson's problem implies the existence of an algorithm to approximate the value of any nonlocal game. However, the second item of Theorem~\ref{thm:hct-informal} shows that an improved compression theorem would refute the existence of such an algorithm, and thus would give a negative answer to (the multipartite version of) Tsirelson's problem.

It is known that Tsirelson's problem for two-prover games is essentially equivalent to Connes' Embedding Conjecture~\cite{connes1976classification}, a longstanding open problem in functional analysis (see~\cite{junge2011connes,fritz2012tsirelson,ozawa2013connes}). In particular, a separation between the definitions of entangled value for games with \emph{two} provers would refute Connes' Embedding Conjecture. We do not know if a separation for games with more than two provers (e.g., $15$) would still refute Connes' Embedding Conjecture.

\subsection{Related work}
We were informed of a forthcoming paper~\cite{coudron2018complexity} by Coudron and Slofstra that establishes a result similar (though strictly incomparable) to Theorem~\ref{thm:main}, using completely different techniques. In particular, the authors show that distinguishing between entangled value $1$ or $1 - 1/\poly(t(n))$ for games with \emph{two} provers in the commuting operator model is hard for nondeterministic $t(n)$ time (whereas our result shows hardness for nondeterministic $2^{t(n)}$ time for games with $15$ provers in the tensor product model). This result relies on the group-theoretic framework that was pioneered in~\cite{slofstra2016tsirelson,slofstra2017set}. 

\subsection{Outlook}

The most important structural properties of classical multiprover
interactive proof systems have been established since the 90s.
It is known that any multiprover interactive proof system can be
parallelized to a single round of interaction, with two provers only;
that completeness $1$ can be achieved without loss of generality; that
soundness can be amplified in parallel; finally, and most importantly,
that the class $\MIP$ of languages that can be recognized by any
multiprover interactive proof system, for any nontrivial choice of
completeness and soundness parameters, is exactly $\NEXP$.
Here, by nontrivial we mean any $(c,s)$ such that $\exp(-\poly(n))\leq
s<c\leq 1$, where $c-s$ is at least
$\exp(-\poly(n))$.
We use $\MIP_{c,s}(k,r)$ to denote the class of languages that can be
decided by a polynomial-time verifier interacting with $k$ provers
through an $r$-round interaction, with completeness $c$ and soundness
$s$.
Thus, $\MIP_{c,s}(2,1)=\NEXP$ for all nontrivial values of $(c,s)$.
When we write $\MIP$ we mean the union of all $\MIP_{c,s}(k,r)$ for
polynomially bounded functions $k,r$, and $c,s$ such that $0<s <c\leq 1$
and $(c-s)^{-1}$ is polynomially bounded.

In contrast, complexity-theoretic aspects of entangled-prover
interactive proof systems remain, to put it mildly, an untamed
wilderness.
Prior to our work it was known that $\NEXP \subseteq
\MIP^*$~\cite{IV12,Vidick13xor,NatarajanV17twoprover} with
completeness $1$ and soundness $\frac{1}{2}$, and that if one allows
the completeness-soundness gap to close exponentially fast with $n$,
then the inclusion can be strengthened to NEEXP, or, in our notation,
$\NTIME(\Lambda_2(n))$~\cite{ji2017compression}.
Interestingly, a similar phenomenon had previously been observed for
single-prover interactive proof systems, for which it is known that
$\QIP=\PSPACE$ with constant gap~\cite{jain2010qip}, but $\QIP$
contains $\EXP$ if one allows a doubly exponentially small
gap~\cite{ito2012quantum}.
Unlike $\MIP*$, however, the power of $\QIP$ does not grow arbitrarily
when the gap goes to zero; for any positive gap the class is contained in
\class{EXPSPACE}~\cite{ito2012quantum}.

For the case of multiprover interactive proof systems with entangled
provers, there is no compelling reason that a shrinking gap would be
necessary for the verification of languages beyond $\NEXP$.
Indeed, no upper bounds are known on $\MIP^*$ with constant gap --- it
is not even known to be contained in the set of decidable languages.
In fact, recent works provide indication that the class may be larger
than $\NEXP$: it is known that $\QMA_{\EXP}$, the ``exponential-size
proof'' analogue of $\QMA$, is such that $\QMA_\EXP \subseteq
\MIP^*_{1,1-2^{-n}}(5,1)$~\cite{FV14,ji2015classical}, and inclusion with a constant gap holds under randomized reductions~\cite{NV18}.
It is therefore an interesting question to determine to what extent
the exponentially small completeness-soundness gap that our technique
requires is necessary.
As mentioned earlier, significant consequences in complexity theory
and mathematics would follow from even a small improvement in our compression theorem,
Theorem~\ref{thm:main-compression}.

Another major open question on entangled-prover interactive proof
systems is the role of the number of provers.
Currently, it is not known if e.g.
$3$ provers allow to determine more languages than $2$ (for any
setting of the completeness-soundness gap).
Our proof of the compression theorem involves a ``prover merging''
step that reduces the number of provers, albeit for a very restricted
type of interactive proof systems.
We also note that our techniques restrict us to games with at least
$7$ provers.
This could potentially be decreased to $5$, or even $3$, by replacing
the use of the $7$-qubit Steane code with, say, a qutrit
error-detecting code.
Achieving a result with two provers seems more challenging.
Yet, the undecidability results in~\cite{slofstra2017set} apply to
two-prover games; it would be interesting to investigate whether some
improvements on our techniques could take us all the way to hardness
results for two-prover games as well.

A number of problems in quantum information theory are known to be
undecidable.
One that bears superficial similarity with the problem considered in
this paper, in the statement as well as in the techniques, is the
undecidability of the spectral gap of an infinite
translation-invariant Hamiltonian, shown
in~\cite{cubitt2015undecidability}.
It would be interesting to determine whether there could be a direct
reduction from a multiprover game to that problem.

\paragraph{Acknowledgments.}
Joseph Fitszsimons acknowledges support from Singapore's Ministry of
Education and National Research Foundation, and the US Air Force
Office of Scientific Research under AOARD grant FA2386-15-1-4082.
This material is based on research funded in part by the Singapore
National Research Foundation under NRF Award NRF-NRFF2013-01.
Thomas Vidick is supported by NSF CAREER Grant CCF-1553477, AFOSR YIP
award number FA9550-16-1-0495, a CIFAR Azrieli Global Scholar award,
and the IQIM, an NSF Physics Frontiers Center (NSF Grant PHY-1125565)
with support of the Gordon and Betty Moore Foundation (GBMF-12500028).
Henry Yuen is supported by ARO Grant W911NF-12-1-0541 and NSF Grant
CCF-1410022.

\paragraph{Outline.}
The rest of the paper is organized as follows.
We cover preliminaries and definitions in Section~\ref{sec:prelim}.
In Section~\ref{sec:enl} we formally define the model of extended
nonlocal games and strategies, as well as Gate Turing Machines.
In Sections~\ref{sec:single-pauli},~\ref{sec:sim-pauli}, and \ref{sec:compression} we prove our compression theorem.
In Section~\ref{sec:recursive} we prove Theorem~\ref{thm:main} and
Theorem~\ref{thm:main-undecidable}.
In Section~\ref{sec:improving} we show that quantitative improvements
to our compression theorem would lead to interesting consequences in computational complexity theory and in foundations of quantum mechanics.
\section{Preliminaries}
\label{sec:prelim}

Let $\mathbb{Z}$ and $\mathbb{N}$ be the set of integers and the set
of natural numbers respectively.
We write $\poly(N)$ for any function $f:\N\mapsto \mathbb{R}_+$ such
that there is an $\alpha>0$ and an $N_0\in\N$ such that $f(N) \leq
N^{\alpha}$ for all $N\geq N_0$.
We write $\poly(N;\eps)$ for any function $f:\N\times\mathbb{R}_+ \to
\mathbb{R}_+$ such that there exists $\alpha,\beta>0$ and $N_0\in\N$,
$\eps_0> 0$ such that, for all $N\geq N_0$ and all $\eps\leq
\eps_0$, $f(N,\eps)\leq N^\alpha \eps^\beta$.

\subsection{Quantum information theory}

All Hilbert spaces considered in the paper are finite dimensional.
We use the terminology ``quantum register'' to name specific quantum
systems with finite dimensional Hilbert spaces.
We use sans-serif font to denote registers, such as $\sA$, $\sB$.
For example, ``register $\reg{A}$'', to which is implicitly associated
the Hilbert space $\mathcal{H}_{\reg{A}}$.

$\Density(\sA)$ denotes the set of density matrices on $\sA$, and
$\Lin(\sA)$ the set of linear operators on $\sA$.
For a density matrix $\rho$ and an operator $M$, we use $\Tr_\rho(M)$
to denote $\Tr(\rho M)$.
A unitary matrix $U$ is a reflection if it has eigenvalues in $\{\pm
1\}$.

\paragraph{Universal gate set.}
The quantum circuits we discuss in this paper use single-qubit
Hadamard and three-qubit Toffoli gates, a universal gate set for
quantum computation~\cite{shi2002both}.

\paragraph{Pauli observables.}
Let $\sigma_I,\sigma_X,\sigma_Y,\sigma_Z$ denote the four single-qubit
Pauli observables
\[
  \sigma_I = \begin{pmatrix}
    1 & 0 \\
    0 & 1 \\
  \end{pmatrix}\;, \qquad \sigma_X = \begin{pmatrix}
    0 & 1 \\
    1 & 0 \\
  \end{pmatrix}\;, \qquad \sigma_Y = \begin{pmatrix}
    0 & -i \\
    i & 0 \\
  \end{pmatrix}\;, \qquad \sigma_Z = \begin{pmatrix}
    1 & 0 \\
    0 & -1 \\
  \end{pmatrix}\;.
\]
We use two ways of specifying a Pauli observable acting on a specific
qubit.
\begin{enumerate}
\item Let $W \in \{I,X,Y,Z\}$ be a label and let $\sR$ be a
  single-qubit register.
  We write $\sigma_W(\sR)$ to denote the observable $\sigma_W$ acting
  on $\sR$.
\item Let $\sR$ be an $n$-qubit register, and let $i \in
  \{1,\ldots,n\}$.
  Let $W = X_i$ (resp.
  $W = Z_i$).
  We write $\sigma_W$ to denote the $\sigma_X$ (resp.
  $\sigma_Z$) operator acting on the $i$-th qubit in $\sR$ (the
  register $\sR$ is implicit).
\end{enumerate}
We also use $W$ to label Pauli operators that have higher ``weight''.
For example, for $W = X_i Z_j$ the operator $\sigma_W$ denotes the
tensor product $\sigma_{X_i} \otimes \sigma_{Z_j}$.
For a vector $u\in\{0,1\}^n$ and $W\in\{X,Z\}$ we write $\sigma_W(u)$
for $\bigotimes_{i:u_i=1} \sigma_{W_i}$.

\begin{lemma}
  \label{lem:closeness_to_groundspace}
  Let $\sA, \sR$ be registers.
  Let $H$ be a positive semidefinite matrix acting on $\sA$ with
  smallest eigenvalue $0$ and second smallest eigenvalue $\Delta > 0$.
  If $\ket{\psi}$ is a state on $\sA \sR$ such that $\bra{\psi}
  H_{\sA} \otimes \Id_{\sR} \ket{\psi} \leq \eps$, then there exists a
  state $\ket{\theta}$ on $\sA \sR$ such that $H \ket{\theta} = 0$ and
  \[
    \left \| \ketbra{\psi}{\psi} - \ketbra{\theta}{\theta} \right \|_1
    \leq 4\sqrt{\eps/\Delta}\;.
  \]
\end{lemma}
\begin{proof}
  Let $P$ denote the projector onto the kernel of $H$.
  Let $Q = \Id - P$.
  Then since $\Delta Q \leq H$ in the positive semidefinite ordering we have $\bra{\psi} Q \ket{\psi}
  \leq \eps/\Delta$.
  The Gentle Measurement Lemma~\cite{ogawa2002new} states that for all
  density matrices $\rho$ and for all positive semidefinite $X$
  satisfying $0 \preceq X \preceq \Id$, we have
  \begin{equation}\label{eq:gm-1}
    \left \| \rho - \sqrt{X}\rho \sqrt{X}
    \right \|_1 \leq 2 \sqrt{ \Tr(\rho (\Id - X))}\;.
  \end{equation}
  Setting $\rho = \ketbra{\psi}{\psi}$ and $X = P$ in~\eqref{eq:gm-1}
  we obtain the desired conclusion with
  \[
    \ket{\theta} = \frac{ P \ket{\psi}}{\sqrt{ \bra{\psi} P
        \ket{\psi}}}\;.
  \]
\end{proof}

\section{Nonlocal games}
\label{sec:enl}

In this paper we consider interactive protocols between a
quantum verifier $V$ and $k$ quantum provers.
We mostly work with a restricted type of three-turn interactive
protocols of the following form. First, the provers send a quantum message to the verifier;
second, the verifier sends classical questions to the provers; third,
the provers reply with classical answers.
Following the terminology introduced in~\cite{johnston2016extended} we
call such protocols ``extended nonlocal games'', or ENL. We also consider \emph{nonlocal games}, which are extended nonlocal games in which the first message is trivial (i.e. there is a single round of classical communication, from verifier to provers and back). 

This section formally introduces extended nonlocal games, as well as a
convenient representation of the verifier for such games as a special
kind of Turing machine, called a ``gate Turing machine'', or GTM.

We start by defining extended nonlocal games (and the special case of nonlocal games) in
Section~\ref{sec:enl-def}. In Section~\ref{sec:mip-star} we recall the definition of the class $\MIP^*$.
In Section~\ref{sec:normalform} we introduce the formalism for
representing strategies for the provers in an ENL.
In Section~\ref{sec:gtm} we introduce a representation of a verifier
in an ENL as a Turing machine.

\subsection{Extended nonlocal games}
\label{sec:enl-def}

Extended nonlocal games are a special kind of three-turn interactive
protocol between a quantum verifier and $k$ quantum provers.
For simplicity we first introduce notation for the case when there is
a single prover $P$.
There are four registers involved: $\sC, \sV, \sM, \sP$.
The verifier $\cV$ acts on registers $\sC$ (the register containing the
prover's initial message), $\sV$ (the verifier's private space) and
$\sM$ (the message register).
The prover $P$ acts on $\sM$ and $\sP$ (the prover's private space).
The registers $\sV$ and $\sM$ are initialized in the $\ket{0}$ state.
The registers $\sC$ and $\sP$ are initialized in an arbitrary state,
chosen by the prover.
The verifier applies a circuit $C_Q$ to the three registers $\sC \sV
\sM$ ($Q$ stands for ``questions'').
The prover then applies an arbitrary unitary transformation $P$ to the
registers $\sM \sP$.
Finally, the verifier applies a circuit $C_A$ to the three registers
$\sC \sV \sM$ ($A$ stands for ``answers'').
The first qubit of $\sV$ is designated as the ``output qubit'', and
measured in the standard basis to determine whether the verifier
accepts or rejects.
See Figure~\ref{fig:qip} for a representation.

We can (and often do) assume without loss of generality that every operation
in this protocol, including the prover's, is a reflection, i.e.
a Hermitian operator that squares to identity.
Indeed, the verifier circuits $C_Q,C_A$ consist of Hadamard gates
($H$) and Toffoli gate ($T$), which are reflections.
The prover's unitary $P$ can be embedded into a reflection by
introducing an ancilla qubit initialized to $\ket{0}$ and considering
the reflection $\tilde{P} = \ketbra{1}{0} \otimes P + \ketbra{0}{1}
\otimes P^\dagger$.

The extension to $k$ provers is straightforward. The registers $\sM$
and $\sP$ are divided into $k$ parts: $\sM_1,\ldots,\sM_k$ and
$\sP_1,\ldots,\sP_k$, such that the $i$-th prover's unitary $P_i$ acts on $\sM_i \sP_i$.

We say that a verifier $\cV = (C_Q,C_A)$ for a $k$-prover three-turn
protocol is a classical-message verifier if there are question and
answer alphabets $\mathcal{Q} = \mathcal{Q}_1 \times \cdots \times
\mathcal{Q}_k$ and $\mathcal{A} = \mathcal{A}_1 \times \cdots \times
\mathcal{A}_k$ such that
\begin{itemize}
\item The only gates of circuit $C_Q$ acting on the message registers
  $\sM$ are CNOT gates, controlled on qubits in $\sV$.
  In other words, $C_Q$ copies messages of length $\log |\cQ_i|$ from
  the register $\sV$ to the register $\sM_i$ for all $i$.
\item Similarly, the circuit $C_A$ is restricted to classically
  copying messages of length $\log |\cA_i|$ from the register $\sM_i$
  into the register $\sV$ for all $i$.
  (After this, an arbitrary quantum computation can be performed on
  $\sV$ only.)
\end{itemize}

We call such protocols with classical-message verifiers \emph{extended
  nonlocal (ENL) games}.
Note that while the verifier sends and receives classical messages in
the register $\sM$, it may receive a quantum message in the register
$\sC$ in the first turn.
A $k$-prover \emph{nonlocal game} is a restricted type of ENL game
where the verifier ignores the register $\sC$.

\subsection{The class $\MIP^*$}
\label{sec:mip-star}

Given a certain class of games, or more generally interactive
protocols, it is possible to define an associated class of languages.
The most common such class is the class $\MIP^*$ of languages that can
be decided by the verifier in a multiprover interactive proof system
in which the verifier is classical and communicates with the provers
in a polynomial number of rounds of interaction, using classical messages
only. Although we have only formally defined nonlocal games with a single round of interaction, the extension to multiple rounds is straightforward. For more background and definitions of complexity classes associated with quantum interactive proof systems, we refer to the introductory text~\cite{watrous2009quantum}.

\begin{definition}[$\MIP*$]
  Let $k,r$ be polynomially bounded functions of $n$, and $0\leq s < c
  \leq 1$ computable functions of $n$.
  We say that a language $L$ is in $\MIP_{c,s}^*(k,r)$ if there is an
  efficient classical procedure that on input $1^n$ returns a family
  of circuits for a verifier that interacts with $k$ provers in $r$
  rounds and is such that
  \begin{enumerate}
  \item \emph{(Completeness:)} If $x \in L$, then there is a strategy
    for the provers that is accepted with probability at least $c$;
  \item \emph{(Soundness:)} If $x \notin L$, no strategy for the
    provers has an acceptance probability that is larger than $s$.
  \end{enumerate}
  We write
  \[
    \MIP^*(k,r) \,=\, \bigcup_{c\in(0,1],\,g\in\poly} \MIP*_{c,c -
      1/g}(k,r)\;.
  \]
\end{definition}

The following problem is complete, under polynomial time Karp reductions, for the class $\MIP^*_{c,s}(k,1)$: given the description of a verifier $\cV$ for a $k$-prover nonlocal game $\Game$, decide whether $\omega^*(\Game) \geq c$ or $\omega^*(\Game) \leq s$.

\subsection{Strategies}
\label{sec:normalform}

The definition of an ENL in Section~\ref{sec:enl-def} models the action of each prover as a single reflection acting jointly on its message and private registers. We refer to the collection of the provers' shared state $\ket{\psi}_{\reg{CPR}}$, where $\reg{R}$ is a reference register, and each prover's reflection $P_i$, $i\in\{1,\ldots,k\}$, as a \emph{reflection strategy} $\strat = (\ket{\psi},\{P_i\})$. 

Since the message register only contains classical information, it is always possible to represent a prover's reflection as a sequence of three operations: copy the message to the prover's private register; apply an arbitrary reflection on the private register; copy the answer from the private register onto the message register. 
We call a strategy for the provers that are decomposed in this form a \emph{normal form strategy}. The structure of normal form strategies will be crucial for our compression result later on. 

We use the following notation to refer to normal form strategies.
Let $\cV = (C_Q,C_A)$ be the circuits for the verifier in a $k$-prover
ENL game $\Game$. Assume without loss of generality that all question and answer sets $\cQ_i$ and $\cA_i$ have the same cardinality
For $i \in \set{1,\ldots,k}$ and $j\in\{1,\ldots,\log|\cQ_k|\}$, let 
$\sM_{ij}$ denote the $j$-th qubit of $\sM_i$.

\begin{definition}
A \emph{normal form ENL game strategy} is a tuple $\strat =
(\rho,\{Q_{ij} \},\{ P_i \}, \{A_{ij} \})$, where $\{Q_{ij} \}$ is
a set of reflections indexed by $i \in \set{1,\ldots,k}$ and $j\in\set{1,\ldots,\log|\cQ_i|}$, $\{ P_i\}$ is a set of reflections indexed by
$i \in \set{1,\ldots,k}$, and $\{ A_{ij} \}$ is a set of reflections
indexed by $i \in \set{1,\ldots,k}$ and $j\in \set{1,\ldots,\log|\sA_i|}$.
For all $(i,j)$, the reflections $Q_{ij}$, $P_i$, $A_{ij}$ act on $\sP_i$.
\end{definition}

The execution of a normal form ENL game strategy $\strat$ in the game
$\Game$ proceeds as follows:

\begin{enumerate}
\item The circuit $C_A$ is executed on the registers $\sC, \sV, \sM$.

\item For each $i \in \set{1,\ldots,k}$,
  the $i$-th prover applies the sequence of gates $\{ \ctl Q_{ij} \}$ for
  $j \in \{1,\ldots,\log|\cQ_i|\}$, where
  \[
    \ctl Q_{ij} = \ketbra{0}{0}_{\sM_{ij}} \otimes \Id_{\sP} +
    \ketbra{1}{1}_{\sM_{ij}} \otimes Q_{ij} \;.
  \]

\item The $i$-th prover applies a reflection $P_i$ on $\sP_i$.
\item For each $i \in \set{1,\ldots,k}$, the $i$-th prover applies the
  sequence of gates $\{ \tgt A_{ij} \}$ for $j \in
  \{1,\ldots,\log|\cA_i|\}$, where
  \[
    \tgt A_{ij} = \Id_{\sM} \otimes \frac{\Id + A_{ij}}{2} +
    \sigma_{X}(\sM_{ij}) \otimes \frac{\Id - A_{ij}}{2}\; .
  \]

\item The circuit $C_A$ is executed on the registers $\sC, \sV, \sM$.
\end{enumerate}
Figure \ref{fig:normalform} gives a representation for the circuit
associated with this protocol. Gates of the form $\ctl Q_{ij}$ and $\tgt A_{ij}$ are referred to as \emph{communication gates}. Gates of the form $P_i$ are referred to as \emph{prover reflection gates}.

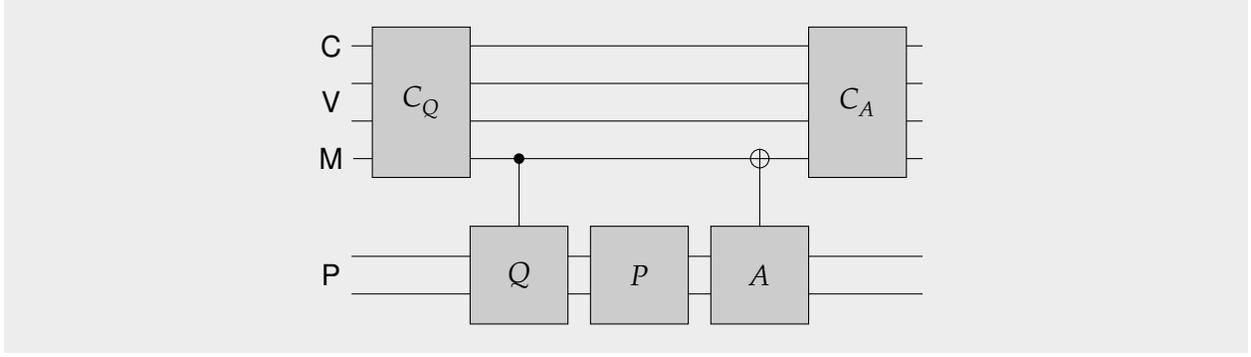
\begin{figure}[!t]
  \begin{mdframed}[style=figstyle]
    \begin{center}
      \begin{tikzpicture}[scale=1, control/.style={circle, fill,
          minimum size = 4pt, inner sep=0mm}, target/.style={circle,
          draw, minimum size = 7pt, inner sep=0mm},
        vgate/.style={draw, minimum height = 2cm, minimum width =
          1.3cm, fill=ChannelColor}, pgate/.style={draw, minimum
          height = 1.3cm, minimum width = 1.3cm, fill=ChannelColor}]

        \node (C) at (-2,1) {$\sC$}; %
        \node (V) at (-2,.5) {$\phantom{\sC}$}; %
        \node (V2) at (-2,0) {$\phantom{\sC}$}; %
        \node (V3) at (-2,.25) {$\sV$}; %
        \node (M) at (-2,-.5) {$\sM$}; %
        \node (P) at (-2,-2.3) {$\phantom{\sC}$}; %
        \node (P2) at (-2,-1.8) {$\phantom{\sC}$}; %
        \node (P3) at (-2,-2.05) {$\sP$}; %

        \node (OutC) at (6,1) {}; %
        \node (OutV) at (6,.5) {}; %
        \node (OutV2) at (6,0) {}; %
        \node (OutM) at (6,-.5) {}; %
        \node (OutP) at (6,-2.3) {}; %
        \node (OutP2) at (6,-1.8) {}; %

        \draw (C.east)--(OutC); %
        \draw (V.east)--(OutV); %
        \draw (V2.east)--(OutV2); %
        \draw (M.east)--(OutM); %
        \draw (P.east)--(OutP); %
        \draw (P2.east)--(OutP2); %

        \node[vgate] (CQ) at (-.8,.25) {$C_Q$}; %
        \node[pgate] (QG) at (.5,-2.05) {$Q$}; %
        \node[control] (C1) at (.5,-.5) {}; %
        \draw (C1.center)--(QG.north); %
        \node[pgate] (PG) at (2.1,-2.05) {$P$}; %
        \node[pgate] (AG) at (3.7,-2.05) {$A$}; %

        \node[vgate] (CA) at (5,.25) {$C_A$}; %
        \node[target] (T1) at (3.7,-.5) {}; %
        \draw (T1.north)--(AG.north);
      \end{tikzpicture}
    \end{center}
  \end{mdframed}
  \caption{An extended nonlocal game in normal form.}
  \label{fig:normalform}
\end{figure}

It is clear that any strategy for the players in an ENL game can be converted to the normal form: the provers use the gates $\ctl Q_{ij}$ to
classically read the message register $\sM$ one bit at a time, apply an arbitrary measurement, controlled on the copied message, on their private register $\sP_i$, and finally use $\tgt A_{ij}$ to classically write their answers into $\sM$ one bit at a time.

In addition we consider a second type of strategy, called \emph{measurement strategies}, which is the standard type of strategies in the study of nonlocal games. Reflection strategies and measurement strategies in ENL games are easily converted from one to another.

\begin{definition}
A \emph{measurement strategy} $\strat$ for the provers in a $k$-prover ENL game $\Game$ with question set $\cQ_1 \times \cdots \times \cQ_k$ and answer set $\cA_1 \times \cdots \times \cA_k$ consists of a
pair $(\rho,\{M_i\})$, where 
\begin{enumerate}
	\item $\rho$ is a state on $(k+1)$ registers denoted $\sC$, $\sP_1,\ldots, \sP_k$.
	\item For each $i\in\{1,\ldots,k\}$, $M_i$ is a map from $\mathcal{Q}_i \times \mathcal{A}_i$ to the set of positive semidefinite operators acting on $\sP_i$, satisfying the constraint that for all $q \in \mathcal{Q}_i$,
	\[
		\sum_{a \in \mathcal{A}_i} M_i(q,a) = \Id_{\sP_i}\;.
	\]
	For each $q \in \mathcal{Q}_i$, we write $M_i(q) = \{M_i(q,a)\}_a$ to denote the associated POVM on $\sP_i$.
\end{enumerate} 
\end{definition}

Next we define the value of a game. 

\begin{definition}
The \emph{value} of a strategy $\strat$ (either measurement or reflection) in a game $\Game$ is denoted by
$\omega^*_\strat(\Game)$ and is defined as the probability that players implementing strategy $\strat$ are accepted by the verifier in $\Game$, i.e. the probability that a measurement of the verifier's output qubit at the end of the interaction returns the outcome $1$. 
The \emph{value} of a game $\Game$ is denoted by $\omega^*(\Game)$ and
is defined as
\[
  \omega^*(\Game)\, =\, \sup_{\strat} \omega^*_\strat(\Game)\;,
\]
where the supremum is over all (finite dimensional) strategies
$\strat$ for $\Game$. 
\end{definition}

\paragraph{Distance between measurement strategies.}

We define notions of closeness of measurement strategies. (There are analogous notions of closeness of reflection strategies; however we will not need them in this paper). 

\begin{definition}[State-dependent closeness of POVMs]
	Let $\rho$ be a density matrix and let $M = \{M^a\}_a,N = \{N^a\}_a$ be two POVMs that have the same set of possible outcomes. Then define 
	\begin{align}
		d_\rho \Paren{ M, N } := \Big[ \sum_a \Tr \Paren{ (M^a - N^a)^2 \rho } \Big]^{1/2}.
	\end{align}
\end{definition}

\begin{definition}[Closeness of strategies]
\label{def:close}
	Let $\mathcal{S} = (\rho,\{M_i\}), \mathcal{S}' = (\rho',\{M_i'\})$ be strategies for an $k$-prover ENL game $\Game$. Then $\mathcal{S}$ is \textbf{$\eps$-close} to $\mathcal{S}'$ if and only if
	\begin{enumerate}
		\item $ \norm{ \rho - \rho'}_{tr} \leq \eps$
		\item For all $i\in\{1,\ldots,k\}$, $\E_q d_\rho(M_i(q),M_i'(q)) \leq \eps$, where the expectation is over $q$ drawn from the marginal distribution of the $i$th prover's questions in the game $\Game$.
	\end{enumerate}
\end{definition}

\begin{definition}[Isometric strategies]
	Let $\mathcal{S} = (\rho,\{M_i\})$ and $\mathcal{S}' = (\rho',\{M_i'\})$ be strategies for an $k$-prover ENL game $G$, where $\rho \in \Density(\sC \sP_1 \cdots \sP_k)$ and $\rho' \in \Density(\sC \sP_1' \cdots \sP_k')$. Then $\mathcal{S}$ is \textbf{$\eps$-isometric} to $\mathcal{S}'$ if and only if there exist isometries: $V_i : \sP_i \to \sP_i'$ for each $i\in\{1,\ldots,k\}$ such that the strategy $\tilde{\mathcal{S}} = (\tilde{\rho}, \set{ \tilde{M}_i })$ is $\eps$-close to $\mathcal{S}'$, where $\tilde{\mathcal{S}}$ is defined by
	\begin{enumerate}
		\item $\tilde{\rho} = (V_1 \otimes \cdots \otimes V_k) \rho (V_1 \otimes \cdots \otimes V_k)^\dagger$
		\item For all $i$, for all $(q,a) \in \mathcal{Q}_i \times \mathcal{A}_i$, $\tilde{M}_i(q,a) = V_i \, M_i(q,a)\, V_i^\dagger$.
	\end{enumerate}
\end{definition}

The following lemma shows that if strategy $\strat_1$ in a $k$-prover ENL game $G$ is $\eps$-isometric to $\strat_2$, then their success probabilities differ by at most $O(k\eps)$. 

\begin{lemma}
  \label{lem:close_strategies}
  Let $\cV = (C_Q,C_A)$ be a verifier in an ENL game $\Game$, and let
  $\strat' = (\rho_1,\{P^{(1)}_i\}),\strat = (\rho_2,\{P^{(2)}_i\})$
  be strategies for $\Game$ such that $\strat$ is $\eps$-isometric
  to $\strat'$.
  Then
  \[
    \big| \omega_{\strat}^*(\Game) - \omega_{\strat'}^*(\Game)
    \big| \,\leq\, O(k\eps)\;.
  \]
\end{lemma}
\begin{proof}
	Observe that $\omega^*_{\strat}(G) = \omega^*_{\tilde{\strat}}(G)$ where $\tilde{\strat}$ is the strategy that is $\eps$-close to $\strat'$ as given by the definition of isometric strategies. Let $\strat''$ denote the strategy that is the same as $\tilde{\strat}$ except the shared state $\rho''$ is taken to be the shared state $\rho'$ of $\strat'$. We have that $\Abs{\omega^*_{\strat''}(G) - \omega^*_{\strat}(G)} \leq \eps$. 
	
	Consider a sequence of $(k+1)$ hybrid strategies $\strat_0,\strat_1,\ldots,\strat_k$ where $\strat_0 = \strat''$ and $\strat_k = \strat'$, and  strategies $\strat_i$ and $\strat_{i+1}$ differ in that the $i$-th prover's measurement operators are switched from those of $\strat''$ to those of $\strat'$. Lemma 7 of~\cite{ji2017compression} implies that $\Abs{\omega^*_{\strat_i}(G) - \omega^*_{\strat_{i+1}}(G)} \leq \eps$. We thus obtain the statement of the lemma.
\end{proof}

\paragraph{Protocol circuits.}

A protocol circuit is a quantum circuit description of a normal form
strategy in an extended nonlocal game (see Figure \ref{fig:normalform}
for an example).
Formally, a $k$-prover protocol circuit $C$ is specified by a set of
$s$ \emph{verifier wires}, a set of $k$ \emph{prover wires}, and a finite
sequence of gates $g_1,g_2,\ldots,g_\tau$.
Every gate $g$ has a \emph{type}, denoted by $\type(g)$:
\begin{enumerate}
\item $H$, which stands for a \emph{double Hadamard} gate\footnote{A
    double Hadamard gate is simply a two-qubit gate that applies two
    Hadamard gates in parallel.
    We use this gate for technical reasons.}
\item $T$, which stands for a Toffoli gate
\item $Q$, which stands for a gate of the form $\ctl Q_{ij}$, for an arbitrary reflection $Q_{ij}$ acting on $\sP_i$.
\item $A$, which stands for a gate of the form $\tgt A_{ij}$, for an arbitrary reflection $A_{ij}$ acting on $\sP_i$.
\item $P$, which stands for a prover reflection $P_i$ acting on $\sP_i$.
\end{enumerate}
The \emph{wires} of a gate $g$, denoted by $\wire(g)$, is the set of
wires it acts on.
Each gate acts on up to $3$ wires.
The size of a $k$-prover protocol circuit with $\tau$ gates and $s +
k$ wires is defined to be $(\tau + s + k)$.

It is easy to see that, from the protocol circuit $C$ of a game, we can
extract the circuits $C_Q$ and $C_A$ defining the verifier $\cV$ of
the game. We may use protocol circuits $C$ and the corresponding
verifier $\cV$ interchangeably.

\subsection{Turing machine descriptions of verifier circuits}
\label{sec:gtm}

In this section, we discuss Turing machine descriptions of verifier circuits.

\begin{definition}\label{def:uniform-circuit}
Let $\Lambda$ denote a countable set. A family of verifier circuits $\set{\cV_{n,\lambda}}_{n\in \N,\lambda \in \Lambda}$ is \emph{uniformly generated} if there is a deterministic Turing machine $M$ that on input $(1^n,\lambda)$ runs in polynomial time and returns a description of $\cV_{n,\lambda}$. 
\end{definition}

\medskip
\noindent \emph{Remark.} In the usual definition of uniformly generated circuits, the circuits are only parameterized by an integer $n$ that denotes the size. In our definition, the verifier circuits are parameterized by both a size parameter $n$ as well as an auxiliary symbol $\lambda$; this generalization will be useful in our proof of the compression theorem. Alternatively, one can think of a family of verifier circuits $\set{\cV_{n,\lambda}}_{n,\lambda}$ as specifying, for each fixed $\lambda \in \Lambda$, a family of uniformly generated verifiers circuits $\set{ \cV_{n,\lambda}}_{n \in \N}$ (in the standard sense). Furthermore, there is a \emph{single} Turing machine $M$, that by fixing the input $\lambda$, generates each family $\set{ \cV_{n,\lambda}}_{n \in \N}$.

\medskip
For our compression result it is not enough for verifiers to have uniform Turing machine descriptions --- it is crucial that they also have highly \emph{succinct} descriptions, defined as follows. 

\begin{definition}\label{def:succinct-circuit}
A family of verifier circuits $\set{\cV_{n,\lambda}}$ has a \emph{succinct description} if there exists a deterministic Turing machine $G$, called the \emph{Gate Turing Machine (GTM) for the protocol circuits $\{C_{n,\lambda}\}$ specified by $\{\cV_{n,\lambda}\}$}  if on input $(n,t,\lambda)$ the Turing machine $G$ runs in polynomial time and returns the description of the $t$-th gate $g$ of $C_{n,\lambda}$ (and a special error symbol in case $t$ is larger than the size of $C_{n,\lambda}$). In addition, we assume that a GTM always returns the size $p_\lambda(n)$ of the protocol circuit $C_{n,\lambda}$ it specifies when provided the input $(1^n,-1,\lambda)$.  
\end{definition}

In the definition, by ``description'' of a gate  we mean the pair $(\type(g),
\wire(g))$.%

We use $\CKT(G,n)$ to denote the protocol circuit whose gates are specified by $G$ on input $(n,t)$ for $1 \leq t \leq p(n)$. We call the circuit $\CKT(G,n)$ the $n$-th protocol circuit specified by $G$, and the game $\Game_n$ corresponding to $\CKT(G,n)$ the $n$-th  game specified by $G$. We say that $G$ is a GTM for a family of ENL games $\set{\Game_n}$ if $\Game_n$ is the $n$-th game specified by $G$.

The following lemma shows that if a verifier family $\Set{\cV_n}$ is uniformly generated, then there is an \emph{equivalent} verifier family $\Set{\cV_n'}$ that has a succinct description. Here, we use a strong notion of equivalence: the question and answer alphabets of $\cV_n'$ are the same as $\cV_n$, and furthermore, the value of any strategy $\strat$ is the same in $\Game_n'$ and $\Game_n$.

\begin{lemma}
\label{lem:gtm-combine}
  Let $k \geq 0$ be an integer. Let $\{ \cV_{n,\lambda} \}= \{(C_{Q,n,\lambda},C_{A,n,\lambda})\}$ be a family of verifier circuits for a $k$-prover ENL game that is uniformly generated by a
  Turing machine $M$. Here $\lambda$ denotes an auxiliary string that is part of the input $(1^n,\lambda)$ to $M$. 
	Let $\Game_{n,\lambda}$ denote the ENL game associated with $\cV_{n,\lambda}$. 
  Then there exists a GTM $G_M$, that is computable from $M$, such that the $n$-th game specified by $G_M$ is $\Game_{n,\lambda}'$ such that:
    \begin{enumerate}
		\item The question and answer alphabets of the verifier of $\Game_{n,\lambda}'$ are the same as in $\Game_{n,\lambda}$;
		\item For all $n$ and for all ENL game strategies $\strat$, $\omega^*_\strat(\Game_{n,\lambda}') = \omega^*_\strat(\Game_{n,\lambda})$.
	\end{enumerate}
\end{lemma}

\begin{proof}
From the Turing machine $M$ it is possible to design two Turing machines $M_Q$ and $M_A$ that specify the families of circuits $\{C_{Q,n,\lambda}\}$ and $\{C_{A,n,\lambda}\}$. As shown in Lemma~\ref{lem:succinct} in Appendix~\ref{sec:gtm}, any uniformly generated family of circuits has a succinct representation of the form described in Definition~\ref{def:succinct-circuit}. Let $G_Q$ and $G_A$ be the associated GTMs. The GTM $G_M$ is a straightforward combination of $G_Q$ and $G_A$. On input $(n,t,\lambda)$, the GTM first determines if the time $t$ corresponds to a gate in $C_Q$, or is among the  $\ctl Q_{ij}$, $P_i$ or $\tgt A_{ij}$ gates, or a gate in $C_A$ (recall the notation for normal form verifiers introduced in Section~\ref{sec:normalform}). 
  This can be determined in polynomial time as each part
  has an easily computable size. If $t$ belongs to the first or last part, $G_M$ determines the appropriate gate by executing $G_Q$ or $G_A$ respectively. 
  In the remaining cases, the correct
  communication gate or prover reflection gate can easily be
  computed in polynomial time.
  \end{proof}
\section{Honest Pauli Prover games}
\label{sec:single-pauli}

As mentioned in the introduction, we prove Theorem~\ref{thm:main-compression} in two parts: first we show how to compress a family of $k$-prover ENL games $\Set{\Game_N}$ specified by a GTM $G$ to a family of $(k+1)$-prover \emph{Honest Pauli Prover} games $\Set{\Game_{H,n}^\compr}$, in which one of the provers is a specially designated ``Honest Pauli Prover'' who is ``commanded'' to measure multi-qubit Pauli observables. We describe Honest Pauli Prover games in this section. In Section~\ref{sec:sim-pauli} we show how to simulate an Honest Pauli Prover game $\Game_{H,n}^\compr$ with a $k$-prover ENL game $\Set{\Game_n^\compr}$. In Section~\ref{sec:compression} we put the two parts together to prove Theorem~\ref{thm:main-compression}.

\medskip
\vspace{10pt}

Throughout this section, we fix a GTM $G$ for a family of $k$-prover ENL games $\Set{\Game_n}$. We write $\CKT(G,n)$ for the $n$-th protocol circuit specified by $G$, and let $p(n)$ denote the size of $\CKT(G,n)$. When $n$ is fixed we let $N=2^n$ and write $\hat{\sC}, \hat{\sV}, \hat{\sM}$ for the registers that the verifier $\cV_N$ in $\Game_N$ acts on, and $\hat{\sX} = \hat{\sC} \hat{\sV} \hat{\sM}$ for the union of these registers.
We interpret $\hat{\sX}$ as an ordered sequence of single-qubit registers $\{\hat{\sC}_i\}$, $\{\hat{\sV}_i\}$, and $\{\hat{\sM}_i\}$. For any register $\sR_i$ of this form, we write 
 $\ind(\sR_i) \in \{1,\ldots,|\hat{\sX}| \}$ for the qubit of
$\hat{\sX}$ that $\sR_i$ corresponds to.

In this section we introduce a family of games $\Set{\Game_{H,n}^\compr}$ that is designed to force the provers to hold a history state of the protocol circuit $\CKT(G,N)$. (The $\compr$ superscript in $\Game_{H,n}^\compr$ indicates that the game is a compression of $\Game_N$.)  These games fall in a category of \emph{Honest Pauli Prover} games, 
defined as follows.

\begin{definition}[Honest Pauli Prover game]\label{def:honest-pauli}
Let $k,S\geq 1$ be integer. An extended nonlocal game $\Game$ is an $(k+1)$-prover $S$-qubit \emph{Honest Pauli
    Prover game} if the following holds.
  The game has $(k+1)$ provers, the first of which is labelled $PV$ and called the
  ``Pauli prover'', and the remaining $k$ are labelled $PP_1,\ldots,PP_{k}$.
  In the game, queries take the form $Q=(q_V,q_{P,1},\ldots,q_{P,k})$, where
  the question $q_V$ to the Pauli prover is a set of labels
  $\{W^{(j)}\}$ for up to three commuting $S$-qubit Pauli observables, each of
  which acts nontrivially on at most two qubits. Answers in the game are labeled $a_V$, $a_{P,1},\ldots,a_{P,k}$,
respectively.
\end{definition}

We introduce a class of strategies for Honest Pauli Prover games in
which the Pauli prover performs Pauli operations honestly.

\begin{definition}[Honest Pauli Prover strategy]
  For $k\geq 0$ we say that a $(k+1)$-prover measurement strategy
  $(\ket{\psi},\{M_i\})$ for an Honest Pauli Prover Game $\cG_H$ is an
  $S$-qubit \emph{honest Pauli Prover strategy} (or honest
  Pauli strategy for short) if the following holds. The state $\ket{\psi}$ is on $(k+3)$ registers: $\sC$ (held by the the verifier), $\sP_V$ (held by the prover $PV$),
$\sP_{P,1},\ldots,\sP_{P,k}$ (held by provers $PP_1,\ldots,PP_k$
respectively), and $\sR$ (a reference register). We use $\sP$ to denote the $(k+1)$ prover registers collectively. Furthermore, the register $\sP_V$ consists of $S$ qubits, and
  on any question $q_V$ the answer bits $a_V$ returned by the Pauli
  prover are obtained by measuring the set of commuting Pauli observables
  that is specified by its question (the prover reports one answer bit for each observable).
	\end{definition}

The verifier $\cV_{H,n}^\compr$ for the game $\Game_{H,n}^\compr$ is summarized in Figure~\ref{fig:honest}. The verifier randomly executes one of three possible routines.
We give the description of each subprotocol in
Section~\ref{sec:gate_check}, Section~\ref{sec:input_check} and
Section~\ref{sec:output_check} respectively.
We conclude with the analysis of $\cV_{H,n}^\compr$ in
Section~\ref{sec:vcomph}.

\begin{figure}[H]
  \centering
  \begin{mdframed}[style=figstyle]
    \ul{Verifier name:} $\cV_{H,n}^\compr$: 
    \begin{itemize}
    \item Execute each of the following subprotocols with
      probability $1/3$: $\textsc{Gate Check}(n)$, $\textsc{Input
        Check}(n)$, and $\textsc{Output Check}(n)$.
    \end{itemize}
  \end{mdframed}
  \caption{The verifier $\cV_{H,n}^{\compr}$.}
  \label{fig:honest}
\end{figure}

\subsection{Gate Check}
\label{sec:gate_check}

The goal of the Gate Check subprotocol is to check that the provers (already assumed to be using an honest Pauli strategy) share a state close to a history state corresponding to the execution of the protocol circuit $\CKT(G,N)$. More precisely, their strategy must be close to one of the following form.

\begin{definition}[Honest Gate Check strategy]\label{def:honest-gate}
An honest Pauli strategy $\strat=(\ket{\psi},\{M_i\})$ is an \emph{honest Gate Check strategy} for the game $\Game_{H,n}^\compr$ derived from the GTM $G$ if
the shared state $\ket{\psi}_{\sC \sP \sR}$ is a history state of the
circuit $\CKT(G,N)$,
\begin{equation}\label{eq:gate-hist}
  \ket{\psi}_{\sC \sP \sR} \,=\, \frac{1}{\sqrt{p(N) + 1}} \sum_{t =
    0}^{p(N)} \ket{t}_{\sC} \otimes \ket{\psi_t}_{\sP \sR}\;,
\end{equation}
where the state $\ket{\psi_0}_{\sP\sR}$ is arbitrary and for all $t \geq 1$, the state $\ket{\psi_t}_{\sP \sR}$ is
defined as $U_{g_t} \ket{\psi_{t-1}}_{\sP \sR}$ where $g_t =
G(N,t)$ and $U_{g_t}$ is the unitary specified in~\eqref{eq:gates},
acting on the registers specified by $\wire(g_t)$. In particular, the register $\sP_V$ is isomorphic to $\hat{\sX} = \hat{\sC} \hat{\sV} \hat{\sM}$, and $S=|\hat{\sX}|$. 
\end{definition}

We proceed to describe the Gate Check, and then state its properties.  
In the check, the verifier samples a random time $t \in
\{1,\ldots,p(N)\}$, and computes the $t$-th gate $g = G(N,t)$ (the verifier can compute this gate by simulating the Turing machine $G$ for $\poly\log(N)$ steps). Depending on the type of $g$, a double Hadamard gate, a Toffoli gate,
a communication channel gate (see Section~\ref{sec:normalform}), or a
prover reflection gate, the verifier executes a specially tailored
subprotocol to check the propagation of that particular gate.

\begin{figure}[H]
  \centering
  \begin{mdframed}[style=figstyle]
    \ul{Subprotocol name:} $\textsc{Gate Check}(n)$: 
    \begin{enumerate}
    \item %
      Select a uniformly random integer $t\in\{1,\ldots, p(N)\}$, and measure
      the clock register $\sC$ using the POVM
      \[
	\big\{ \Pi^0 = \ketbra{+_t}{+_t}, \Pi^1 = \ketbra{-_t}{-_t}, \Pi^2
        = \Id - \Pi^0 - \Pi^1 \big\}\;,
      \]
      where $\ket{\pm_{t}} \,=\, \frac{1}{\sqrt{2}}\big(\ket{t-1}\pm\ket{t}\big)$. Let $s \in \{0,1,2\}$ denote the result of the measurement. If $s = 2$, accept.
    \item Simulate the execution of the the GTM $G$ on input $(N, t)$ to
      obtain $g = G(N,t)$.
    \item If $\type(g) = T$, run $\textsc{Toffoli Check}(n,s,g)$.
    \item If $\type(g) = H$, run $\textsc{Hadamard Check}(n,s,g)$.
    \item If $\type(g) \in \{Q,A \}$, run $\textsc{Communication
        Channel Check}(n,s,g)$.
    \item If $\type(g) = P$, run $\textsc{Prover Reflection
        Check}(n,s,g)$.
    \end{enumerate}
  \end{mdframed}
  \caption{Gate Check}
  \label{fig:gate_check}
\end{figure}

Figure~\ref{fig:toffoli_hadamard_check} details the subprotocols
invoked by \textsc{Gate Check}.
The subprotocols \textsc{Toffoli Check} and \textsc{Hadamard Check} are
taken from~\cite{ji2017compression}.
A Toffoli or doubled Hadamard gate $g$ returned by the GTM $G$ always
comes together with labels for a set of qubits on which the gate acts
on.
In the subprotocols \textsc{Hadamard Check}, \textsc{Communication
  Channel Check}, and \textsc{Prover Reflection Check}, the verifier
artificially accepts with probability $1/2$ without testing anything; this is to adjust the normalization of the rejection probabilities of
these subprotocols.

The next lemma establishes an expression for the rejection probability for \textsc{Gate Check} conditioned on a choice of random $t \in \{1,\ldots,p(N)\}$.

\begin{figure}[htb!]
  \centering
  \begin{mdframed}[style=figstyle]
    \ul{Subprotocol name:} $\textsc{Toffoli Check}(n,s,g)$: \\
    \ul{Description of input:} $g$ is a Toffoli gate acting on qubits
    $u_1,u_2,u_3$, and $s \in \{0,1\}$.
    \begin{enumerate}
    \item Sample $\alpha \in \{0,1\}$ uniformly at random, and accept
      if $\alpha = 1$.
      Otherwise, continue.
    \item Set $q_V = (Z_{u_1},Z_{u_2},X_{u_3})$.
      Let $a_V = (a_1,a_2,a_3)$ be the three answer bits from $P_V$.
      Reject if $a_1 = a_2 = 1 \wedge s \oplus a_3 = 1$, or $a_1 a_2 =
      0 \wedge s = 1$.
      Accept otherwise.
    \end{enumerate}
    \vspace{10pt}
    \ul{Subprotocol name:} $\textsc{Hadamard Check}(n,s,g)$: \\
    \ul{Description of input:} $g$ is a double Hadamard gate acting on
    qubits $u_1,u_2$, and $s \in \{0,1\}$.
    \begin{enumerate}
    \item Sample $\alpha \in \{0,1\}$ uniformly at random.
    \item If $\alpha = 0$, set $q_V = ( X_{u_1} X_{u_2}, Z_{u_1}
      Z_{u_2})$.
      Let $a_1,a_2$ be the two answer bits from $P_V$.
      Reject if $s \oplus a_1 = s \oplus a_2 = 1$, accept otherwise.

    \item If $\alpha = 1$, set $q_V = (X_{u_1}Z_{u_2},Z_{u_1}
      X_{u_2})$.
      Let $a_1,a_2$ be the two answer bits from $P_V$.
      Reject if $s \oplus a_1 = s \oplus a_2 = 1$ and accept
      otherwise.
    \end{enumerate}
    \vspace{10pt}
    \ul{Subprotocol name:} $\textsc{Communication Channel Check}(n,s,g)$: \\
    \ul{Description of input:} $g$ is a communication gate $\ctl
    Q_{ij}$ or $\tgt A_{ij}$, and $s \in \{0,1\}$.
    \begin{enumerate}
    \item Sample $\alpha \in \{0,1\}$ uniformly at random, and accept
      if $\alpha = 1$.
      Otherwise, continue.
    \item Let $(i,j) = \wire(g)$.
      Let $u = \ind(\hat{\sM}_{ij})$.
    \item If $\type(g) = Q$: Set $q_V = Z_u$.
      Set $q_{P,i} = Q_{ij}$.
      Reject if $a_V = 1 \wedge s \oplus a_{P,i} = 1$, or $a_V = 0
      \wedge s = 1$.
      Accept otherwise.
    \item If $\type(g) = A$: Set $q_V = X_u$.
      Set $q_{P,i} = A_{ij}$.
      Reject if $a_{P,i} = 1 \wedge s \oplus a_V = 1$, or $a_{P,i} = 0
      \wedge s = 1$.
      Accept otherwise.
    \end{enumerate}
    \vspace{10pt}
    \ul{Subprotocol name:} $\textsc{Prover Reflection Check}(n,s,g)$: \\
    \ul{Description of input:} $g$ is a prover reflection gate, and $s
    \in \{0,1\}$.
    \begin{enumerate}
    \item Sample $\alpha \in \{0,1\}$ uniformly at random, and accept
      if $\alpha = 1$. Otherwise, continue.
    \item Let $i = \wire(g)$.
      Set $q_{P,i} = \star$.
    \item Reject if $a_{P,i} \neq s$.
      Accept otherwise.
    \end{enumerate}
  \end{mdframed}
  \caption{Toffoli, Hadamard, Communication Channel, and Prover
    Reflection Checks.}
  \label{fig:toffoli_hadamard_check}
\end{figure}

\begin{lemma}
  \label{lem:ver_gate_check}
  Let $\strat = (\ket{\psi},\{M_i\})$ be an honest Pauli strategy for the \textsc{Gate Check} subprotocol. %
  For all $i \in \{1,\ldots,k\}$ let $Q_{ij}$, $A_{ij}, P_i$ be prover $PP_i$'s
  observables on questions $Q_{ij}, A_{ij}, \star$ respectively.
  Let $\ctl Q_{ij}$ and $\tgt A_{ij}$ denote the associated controlled
  operators defined in Section~\ref{sec:normalform}.
  
  Fix $t \in \{1,\ldots,p(N)\}$. Let $g = G(N,t)$ denote the $t$-th gate of the protocol circuit $\CKT(G,N)$. Let
  \begin{equation}
    \label{eq:gates}
    U_g = \left \{
      \begin{array}{ll}
        H^{\otimes 2} & \mbox{if } \type(g) = H \\
        T 	& \mbox{if } \type(g) = T \\
        \ctl Q_{ij}	& \mbox{if } \type(g) = Q, \wire(g)=(i,j) \\
        \tgt A_{ij}	& \mbox{if } \type(g) = A, \wire(g)=(i,j) \\
        P_i & \mbox{if } \type(g) = P, \wire(g) = i.
      \end{array}
    \right.
  \end{equation}
  Then the rejection probability of \textsc{Gate Check}, conditioned on the verifier selecting time $t \in \{0,1,\ldots,p(N)\}$ in Step 1 of Figure~\ref{fig:gate_check}, is
  \[
    \frac{1}{4} \Tr_{\rho} \Paren{ K_t \paren{\Id - J_t \otimes U_g}K_t }\;,
  \]
  where $\rho = \ketbra{\psi}{\psi}$, $K_t$ denotes the projector $\ketbra{+_t}{+_t} + \ketbra{-_t}{-_t}$ acting on $\sC$ and $J_t$ denotes the unitary operator $\Id - 2\ketbra{-_t}{-_t}$ acting on $\sC$.
\end{lemma}

\begin{proof}
  The rejection probability for the double Hadamard and Toffoli gates
  was established in~\cite{ji2017compression}.
  In the case of $\type(g) = Q$, the rejection probability is
  \[
    \frac{1}{2} \Tr_{\rho} \Paren{ K_t \Brac{ \ketbra{-_t}{-_t} \otimes
      \frac{\Id + \sigma_{Z_u}}{2} + \frac{\Id - J_t \otimes Q_{ij}
      }{2} \otimes \frac{\Id - \sigma_{Z_u}}{2}} K_t }
  \]
  which can be verified to be equal to $\frac{1}{4} \Tr_{\rho} \paren{
    K_t (\Id - J_t \otimes U_g) K_t }$.
  In the case that $\type(g) = A$, the rejection probability is
  \[
    \frac{1}{2} \Tr_{\rho} \Paren{ K_t \Brac{\ketbra{-_t}{-_t} \otimes
      \frac{\Id + A_{ij}}{2} + \frac{\Id - J_t \otimes \sigma_{X_u}
      }{2} \otimes \frac{\Id - A_{ij}}{2} } K_t }
  \]
  which again can be verified to be equal to $\frac{1}{4} \Tr_{\rho}
  \paren{ K_t (\Id - J_t \otimes U_g)K_t }$.
  In the case of $\type(g) = P$, the rejection probability is by
  definition
  \[
    \frac{1}{4} \Tr_{\rho} \paren{ K_t (\Id - J_t \otimes U_g) K_t}.
  \]
\end{proof}

\begin{lemma}
  \label{lem:gate_check}
The following hold for the $\textsc{Gate check}$ subprotocol described in Figure~\ref{fig:toffoli_hadamard_check}:
  \begin{enumerate}
  \item (Completeness) An honest Gate Check strategy passes the
    $\textsc{Gate check}$ subprotocol with probability $1$.
  \item (Soundness) Any honest Pauli strategy that passes
    the $\textsc{Gate check}$ subprotocol with probability at least $1
    - \eps$ is $\delta$-close (see Definition~\ref{def:close}) to an honest Gate Check strategy, for $\delta = O(p(N)^{3/2} \sqrt{\eps})$
  \end{enumerate}
\end{lemma}

\begin{proof} 
Completeness is straightforward.
 We show soundness.
  The analysis largely follows~\cite{ji2017compression}. Let $\strat$ be an honest Pauli strategy that succeeds with probability at least $1-\eps$ in the $\textsc{Gate check}$ subprotocol. 
  Let $\ket{\psi}_{\sC \sP \sR}$ denote the provers' shared state in $\strat$, and let $\rho = \ketbra{\psi}{\psi}$.

  We calculate the rejection probability of $\textsc{Gate Check}$. At step 1. in $\textsc{Gate Check}$ the verifier selects a time $t$ uniformly at random from $\{1,\ldots,p(N)\}$. Let $g_t = G(N,t)$ denote the $t$-th gate of $\CKT(G,N)$. Let $r_t$ denote the rejection probability of $\textsc{Gate Check}$ conditioned on time $t$ having been selected. By Lemma~\ref{lem:ver_gate_check}, the
  rejection probability is $r_t = \frac{1}{4} \Tr_{\rho} \left ( K_t(\Id - J_t \otimes U_{g_t}) K_t
    \right )$.
	Thus the overall rejection probability satisfies
	\begin{align}
    \eps &\geq \E_t r_t \notag \\ 
    	 &\geq \frac{1}{4} \E_t \Tr_{\rho} \left ( K_t(\Id - J_t \otimes U_{g_t}) K_t \right ) \notag \\
	 	\label{eq:gate_check}
	 	&= \frac{1}{4} \E_t \Tr_{\rho} \Paren{ \ketbra{t-1}{t-1} \otimes \Id + \ketbra{t}{t} \otimes \Id - \ketbra{t-1}{t} \otimes U_{g_t}^\dagger - \ketbra{t}{t-1} \otimes U_{g_t} }
	\end{align}
	where in the last equality we used the fact that $U_{g_t}^\dagger = U_{g_t}$. 
	Define $Q = \sum_t \ketbra{t}{t}_{\sC}
    \otimes U_{g_t} \cdots U_{g_1}$. It is straightforward to verify that~\eqref{eq:gate_check} implies
  \[
    \Tr_{\rho} \E_t \left ( Q\ketbra{-_t}{-_t}Q^\dagger \right) \leq
    2\eps\;.
  \]
  Let $H_{prop}$ denote the operator $\sum_t \left ( Q\ketbra{-_t}{-_t}Q^\dagger
  \right)$. Notice that $H_{prop}$ is a positive semidefinite operator that is exactly the same 
  as the propagation term of the Feynman-Kitaev clock Hamiltonian~\cite{kitaev2002classical}. It has been shown that this propagation term has a spectral gap of at least $\Omega(1/p(N)^2)$~\cite{aharonov2008adiabatic}, and therefore the scaled operator $\E_t \left ( Q\ketbra{-_t}{-_t}Q^\dagger \right)$ has spectral gap of at least $\Omega(1/p(N)^3)$. Using Lemma~\ref{lem:closeness_to_groundspace}, we have that $\rho$
  is $\delta$-close to a pure state $\ketbra{\theta}{\theta}$
  satisfying $H_{prop} \ket{\theta} = 0$ for $\delta = O(p(N)^{3/2} \sqrt{\eps})$. Since the ground space of the propagation term of is spanned by history states of the form $\ket{\theta}_{\sC \sP \sR} = \frac{1}{\sqrt{p(N)+1}}
  \sum_t \ket{t}_{\sC} \otimes \ket{\theta_t}_{\sP \sR}$ where
  $\ket{\theta_t} = U_{g_t} \ket{\theta_{t-1}}$, this establishes the lemma.

\end{proof}

\subsection{Input check}
\label{sec:input_check}

Assume that the provers' strategy is an honest \textsc{Gate Check}
strategy (Definition~\ref{def:honest-gate}).
The \textsc{Input Check} subprotocol is designed to check that the
component $\ket{\psi_0}_{\sP\sR}$ of the history state~\eqref{eq:gate-hist} at time $t=0$ is a valid initial state
for the protocol circuit.

\begin{definition}[Honest Input Check strategy]\label{def:honest-input}
An honest  Gate Check strategy $\strat=(\ket{\psi},\{M_i\})$ is an \emph{honest Input Check strategy} if the initial state
$\ket{\psi_0}_{\sP \sR}$ is such that
the registers $\hat{\sV} \hat{\sM}$ of $\sP_V$ are initialized to the all zero state.
\end{definition}

\begin{figure}[H]
  \centering
  \begin{mdframed}[style=figstyle]
    \ul{Subprotocol name:} $\textsc{Input Check}(n)$:
    \begin{enumerate}
    \item Measure the clock register $\sC$ in the computational basis.
      Let $t \in \{0,\ldots,p(N)\}$ be the outcome.
      If $t \neq 0$, accept.

    \item Pick a random qubit index $j \in \supp(\hat{\sV}
      \hat{\sM})$.
    \item Set $q_V = Z_j$.
      Accept if $a_V = 0$.
      Otherwise, reject.
    \end{enumerate}
  \end{mdframed}
  \caption{Input Check.}
  \label{fig:input_check}
\end{figure}

\begin{lemma}
  \label{lem:input_check}
The following hold for the $\textsc{Input check}$ subprotocol described in Figure~\ref{fig:input_check}:
  \begin{enumerate}
  \item (Completeness) An honest Input Check strategy passes the
    $\textsc{Input Check}$ subprotocol with probability $1$.
  \item (Soundness) Any honest Gate Check strategy that passes the
    $\textsc{Input Check}$ subprotocol with probability at least $1 -
    \eps$ is $\delta$-close to an Honest Input Check strategy for $\delta = O(p(N)\sqrt{\eps})$.
  \end{enumerate}
\end{lemma}
\begin{proof}
  Completeness is straightforward.
  We show soundness. Let $\strat$ be a strategy that passes the
    $\textsc{Input Check}$ subprotocol with probability at least $1 -
    \eps$.
  Let $\ket{\psi}_{\sC \sP \sR}$ denote the shared state in $\strat$. Since the strategy $\strat$ is an honest Gate Check strategy (and
  therefore an honest Pauli Check strategy), we have that
  \[
    \ket{\psi}_{\sC\sP \sR} = \frac{1}{\sqrt{p(N)+1}} \sum_{t = 0}^{p(N)}
    \ket{t}_{\sC} \otimes \ket{\psi_t}_{\sP \sR}\;.
  \]
  Let $\Pi = \ketbra{0}{0}_{\sC}$, and let $\rho = \ketbra{\psi}{\psi}$. We have that $\Tr_\rho (\Pi) \geq (p(N)+1)^{-1}$.
  Let
  \[
    \rho_{0} = \frac{\Pi \rho \Pi}{\Tr_\rho (\Pi)} =
    \ketbra{0}{0}_{\sC} \otimes \ketbra{\psi_0}{\psi_0}_{\sP \sR}\;.
  \]
  The probability that \textsc{Input Check} rejects when the shared
  state is $\rho_{0}$ instead of $\rho$ is at most $\eps' = (p(N)+1)
  \eps$.

  Suppose now that the shared state in \textsc{Input Check} is
  $\rho_{0}$.
  The probability of rejection is then
  \begin{equation}
    \Tr_{\rho_{0}} \Paren{H_{init}} \leq \eps'\;,
    \label{eq:init}
  \end{equation}
  where
  \[
    H_{init} = \frac{1}{|\hat{\sV} \hat{\sM}|} \sum_{i \in
      \supp(\hat{\sV} \hat{\sM})} \ketbra{1}{1}_i\;,
  \]
  with $|\hat{\sV} \hat{\sM}| \leq p(N)$ the number of qubits in register
  $\hat{\sV} \hat{\sM}$.

  Observe that the operator $H_{init}$ is positive semidefinite, has smallest eigenvalue $0$, and has spectral gap of at least $1/p(N)$.
  Furthermore, the kernel of $H_{init}$ is spanned by states of the form $\ket{\theta}_{\sP \sR}$ where the register $\hat{\sV} \hat{\sM}$ is in the all zeroes state.
  Using Lemma~\ref{lem:closeness_to_groundspace}, we conclude that
  $\ket{\psi_0}$ is $\delta$-close to such a state
  $\ket{\theta}_{\sP \sR}$ for $\delta = O(p(N)\sqrt{\eps})$.
  This concludes the proof.

\end{proof}

\subsection{Output check}
\label{sec:output_check}

As for the Input check, assume that the provers share a valid history
state of the protocol circuit $\CKT(G,N)$.
The \textsc{Output Check} subprotocol is designed to check that the state
held by the provers is a history state of an accepting computation. In 
other words, the \textsc{Output Check} subprotocol enforces that the 
output qubit of the last time step of the history state is in the state $\ket{1}$.

\begin{figure}[H]
  \centering
  \begin{mdframed}[style=figstyle]
    \ul{Subprotocol name:} $\textsc{Output Check}(n)$:
    \begin{enumerate}
    \item Measure the clock register $\sC$ in the computational basis.
      Let $t \in \{0,\ldots,p(N)\}$ be the outcome.
      If $t \neq p(N)$, accept.
    \item Let $u$ denote the index of the decision bit in $\hat{\sV}$.
    \item Set $q_V = Z_u$. If $a_V = 0$, reject.
      Otherwise, accept.
    \end{enumerate}
  \end{mdframed}
  \caption{Output Check}
  \label{fig:output_check}
\end{figure}

\begin{lemma}
  \label{lem:output_check}
The following hold for the $\textsc{Output check}$ subprotocol described in Figure~\ref{fig:output_check}:  \begin{enumerate}
  \item (Completeness) For all $\gamma > 0$ there exists an honest Input Check strategy
    that passes the \textsc{Output Check} subprotocol with probability
    \[
      1 - \frac{1 - \omega^{*}(\Game_N) + \gamma}{p(N)+1}.
    \]
  \item (Soundness) Any honest Input Check strategy passes the \textsc{Output
    Check} subprotocol with probability at most
    \[
      1 - \frac{1 - \omega^{*}(\Game_N)}{p(N)+1}\;.
    \]
  \end{enumerate}
\end{lemma}
\begin{proof}
	We show the Completeness part. Consider a normal form $k$-prover strategy $\mathcal{T}$ for $\Game_N$ that achieves the value at least $\omega^*(\Game_N) - \gamma$ (there isn't necessarily a strategy that achieves the optimal value $\omega^*(\Game_N)$). The strategy $\mathcal{T}$ is comprised of a shared state $\ket{\varphi}$ on register $\hat{\sC} \hat{\sP}$ and reflections $\set{A_{ij}}$, $\set{Q_{ij}}$, and $\set{P_i}$ as described in Section~\ref{sec:normalform}.

  Consider the following Honest Input Check strategy $\strat$: the shared state $\ket{\psi}$ is the history state of the protocol circuit $\CKT(G,N)$ where the provers' reflections $\set{A_{ij}}$, $\set{Q_{ij}}$, and $\set{P_i}$ are given by the strategy $\mathcal{T}$. Since the strategy $\mathcal{T}$ succeeds in $\Game_N$ with probability at least $\omega^*(\Game_N) - \gamma$, strategy $\strat$ succeeds in \textsc{Output Check} with the claimed probability.

  We now show soundness. Let $\strat$ be an Honest Input Check strategy that passes the Output
    Check subprotocol with probability at least $1-\eps$. 
  Let $\ket{\psi}_{\sC \sP \sR}$ denote the shared state.
  Since the strategy is an Honest Input Check strategy, the shared
  state is a history state of the protocol circuit $C$
  \[
    \ket{\psi}_{\sC\sP\sR}\,=\,\frac{1}{\sqrt{p(N) + 1}} \sum_{t = 0}^{p(N)} \ket{t}_{\sC} \otimes
    \ket{\psi_t}_{\sP \sR}\;,
  \]
  with the initial snapshot state $\ket{\psi_0}$ representing the
  state of the verifier and provers at the start of an execution of
  the game $\Game_N$.
  Let $\rho = \ketbra{\psi}{\psi}$.
  Let $\Pi = \ketbra{N}{N}_{\sC}$.
  We have that $\Tr_\rho (\Pi) = 1/(p(N)+1)$.
  Let
  \[
    \rho_{f} = \frac{\Pi \sigma \Pi}{\Tr_\rho (\Pi)} = \ketbra{\psi_N}{\psi_N}.
  \]
  The probability that \textsc{Output Check} rejects when the shared
  state $\rho_f$ is at most $\eps' = (p(N)+1)\eps$.
  
  Note that $\ket{\psi_N}$ final snapshot of
	a history state of the protocol circuit $CKT(G,N)$, which
	specifies a reflection strategy $\mathcal{T}$ for the game
	$\Game_N$. Therefore the rejection probability of \textsc{Output Check} when the shared state is $\rho_f$ is $\Tr \left ( \ketbra{0}{0}_{out} \,
    \ketbra{\psi_{N}}{\psi_{N}} \right)$, which is at least $1 -
  \omega^*(\Game_N)$. 
  This concludes the proof of the lemma.
\end{proof}

\subsection{Analysis of $\cV^{\compr}_{H,n}$}
\label{sec:vcomph}

The
following lemma states the important properties of the verifier $\cV_{H,n}^\compr$ specified in Figure~\ref{fig:honest}.

\begin{lemma}
  \label{lem:single_pauli_game}
  \leavevmode Let $G$ be a GTM for a family of $k$-prover ENL games
  $\{\Game_n\}$, and let $\cV^{\compr}_{H,n}$ be the verifier described in
  Figure~\ref{fig:honest}.
  Let $n\geq 1$ be an integer, $N=2^n$, $S = p(N)$, and $\Game^{\compr}_{H,n}$ be the $S$-qubit Honest Pauli Prover game
  whose verifier is specified by $\cV_{H,n}^{\compr}$. Then the following hold:
  \begin{enumerate}
  \item (Completeness) For all $\gamma > 0$ there exists
    an honest Pauli strategy $\strat$ that has value
    \[
      \omega^{*}_\strat(\Game^{\compr}_{H,n}) = 1 - \frac{1 -
        \omega^{*}(\Game_N) + \gamma}{p(N)+1}\;.
    \]
  \item (Soundness) There exists universal constants $\alpha\geq 1,\beta>0$ such
    that for all Honest Pauli strategies $\strat$,
    \[
      \omega^{*}_\strat(\Game^{\compr}_{H,n}) \leq 1 - \Paren{\frac{1 -
          \omega^{*}(\Game_N)}{\beta \, p(N)}}^{\alpha}\;.
    \]
  \end{enumerate}
\end{lemma}

\begin{proof}
  Completeness follows from combining the completeness
  statements of the Gate Check, Input Check, and Output Check.

  We prove soundness.
  Let $\strat$ be an Honest Pauli Prover strategy that succeeds
  with probability $1 - \eps$ in the game $\Game^{\compr}_{H,n}$.
  Then it succeeds with probability at least $1 -
  3\eps$ in each of the \textsc{Gate Check}, \textsc{Input Check}, and
  \textsc{Output Check} subprotocols.

  Let $\delta = O(p(N)^{3/2} \sqrt{\eps})$. By Lemma~\ref{lem:gate_check}, there exists an honest \textsc{Gate
    Check} strategy $\mathcal{S}_1$ that is $\delta$-close to
  $\mathcal{S}$. Using Lemma~\ref{lem:close_strategies}, this implies that $\mathcal{S}_1$ succeeds in the \textsc{Input} and
  \textsc{Output Check} subprotocols with probability at least $1 -
  3\delta$.

  Let $\delta' = O(p(N) \sqrt{\delta})$. Applying Lemma~\ref{lem:input_check}, there exists an honest
  \textsc{Input Check} strategy $\mathcal{S}_2$ that is
  $\delta'$-close to $\strat_1$.
  The strategy $\strat_2$ succeeds in the \textsc{Output Check}
  subprotocol with probability at least $1 - 3\delta'$ (using
  Lemma~\ref{lem:close_strategies} again).

  Finally, applying Lemma~\ref{lem:output_check}, the success
  probability of $\strat_2$ in \textsc{Output Check} is at most
  \[
    1 - \frac{1 - \omega^{*}(\Game_N)}{p(N)+1}\;.
  \]
  This implies that there exist universal constants $\beta,\mu,\nu > 0$ we have
  \[
    \omega^*_{\strat_2}(\Game_{H,n}^\compr) = 1 - p(N)^\mu \eps^\nu \leq 1 -
    \frac{1 - \omega^{*}(\Game_N)}{\beta \, p(N)}\;,
  \]
  which implies
  \[
    \omega^*_\strat(\Game_{H,n}^\compr) = 1 - \eps \leq 1 - \Paren{\frac{1 -
        \omega^{*}(\Game_N)}{\beta \, p(N)}}^{\alpha}\;,
  \]
  for some universal constant $\alpha$. This concludes the proof.

\end{proof}

We point out some properties of the games specified by $\cV^{\compr}_{H,n}$
that will be relevant for the next stage of the argument.
In all the subprotocols above, the honest Pauli prover $PV$ gets a
question that specifies up to three commuting Pauli observables. (Furthermore, the honest Pauli prover's question can be embedded in what we call an \emph{MS-compatible triple}; see Section~\ref{sec:entanglement_tests}.) 
All other provers $PP_i$ get questions from the set $\{ \star \} \cup
\{ Q_{ij} \} \cup \{A_{ij} \} \cup \{\bot\}$, where the $\bot$ symbol is used to denote the absence of a question.
Furthermore, note that at any one time, at most one $PP_i$ prover gets
sent a message that is not $\bot$. %
\section{Simulating Honest Pauli Prover games}
\label{sec:sim-pauli}

Let $\Game_{H}$ be any $(k+1)$-prover $S$-qubit Honest Pauli Prover
game (Definition~\ref{def:honest-pauli}) such that $k\geq 7$. In this section we introduce a $k$-prover \emph{Simulated
  Pauli Prover game} $\Game_S$ that uses $7$ out of the $k$
provers to simulate the Pauli prover in $\Game_{H}$ (provided that $\Game_H$ satisfies some mild conditions) using a
technique similar to the ``code-check'' test
in~\cite{ji2017compression,NV18}. 

In Section~\ref{sec:codes} we introduce a class of error-correcting codes that will be used in the game. In Section~\ref{sec:entanglement_tests} we present a multi-qubit test for constant-weight Pauli observables. In Section~\ref{sec:sim-game-desc} we define the simulated Pauli Prover game and state its properties. 

\subsection{Stabilizer codes}
\label{sec:codes}

We consider weakly self-dual \emph{Calderbank-Shor-Steane (CSS)
  codes}~\cite{calderbank1996good,steane1996error}.
Let $C$ be a classical $[m,d]$ linear error-correcting code over
$\Fp_2$: $C$ is specified by a generator matrix $H \in \Fp_2^{m\times d}$
and a parity check matrix $K\in \Fp_2^{(m-d)\times d}$ such that $C =
\text{Im}(H)=\ker(K)$.
We say that $C$ is weakly self-dual if the dual code $C^\perp$, with
generator matrix $K^T$, is such that $C\subseteq C^\perp$; equivalently, $H^T
H=0$.
To any such code $C$ we associate a subspace $\mathcal{C}$ of
$(\complex^2)^{\otimes m}$ that is the simultaneous $+1$ eigenspace of a set
of stabilizers $\{S_{W,j}\}_{W\in\{X,Z\},j\in\{1,\ldots,k'\}}$ such that
$S_{W,j}$ is a tensor product of Pauli $\sigma_W$ observables over $\Fp_2$
in the locations indicated by the $j$-th column of the generator
matrix $H$, i.e.
\[ S_{W,j} = \sigma_W(H_{1j}) \otimes \sigma_W(H_{2j}) \otimes \cdots \otimes \sigma_W(H_{mj}), \]
where $H_{ij}$ is the $(i,j)$-th entry of $H$.
The condition that $H^TH=0$ implies that all the $S_{W,j}$ commute, so
that $\mathcal{C}$ is well-defined.

\paragraph{The $7$-qubit Steane code.} We make use of the
\emph{Steane} code, a CSS code that encodes $1$ qubit into
$7$ physical qubits~\cite{Steane96}.
In Figure \ref{fig:css}, we list the stabilizer generators of the code
as well as several logical $X$ and logical $Z$ operators (that are
equal up to multiplication by a stabilizer).
The logical generators satisfy the useful property that for every $i \in
\{1,\ldots,7\}$, there exists a logical $X$ (resp.
logical $Z$) operator that acts trivially on the $i$-th qubit.

\begin{figure}[H]
  \centering
  \begin{tabular}{ l@{\hspace{1.5em}}c c c c c c c }
    \toprule
    \multicolumn{8}{c}{Stabilizer Generators} \\
    \midrule
    	 $S_1$ & $X$ & $X$ & $X$ & $X$ & $I$ & $I$ & $I$ \\
		 $S_2$ & $X$ & $X$ & $I$ & $I$ & $X$ & $X$ & $I$ \\
		 $S_3$ & $X$ & $I$ & $X$ & $I$ & $X$ & $I$ & $X$ \\
		 $S_4$ & $Z$ & $Z$ & $Z$ & $Z$ & $I$ & $I$ & $I$ \\
		 $S_5$ & $Z$ & $Z$ & $I$ & $I$ & $Z$ & $Z$ & $I$ \\
		 $S_6$ & $Z$ & $I$ & $Z$ & $I$ & $Z$ & $I$ & $Z$ \\
    \midrule
    \multicolumn{8}{c}{Logical Operators} \\
    \midrule\\[-1em]
    $\comp{X}$ & $I$ & $I$ & $I$ & $I$ & $X$ & $X$ & $X$ \\
          & $X$ & $X$ & $I$ & $I$ & $I$ & $I$ & $X$ \\
          & $X$ & $I$ & $X$ & $I$ & $I$ & $X$ & $I$ \\
    $\comp{Z}$ & $I$ & $I$ & $I$ & $I$ & $Z$ & $Z$ & $Z$ \\
          & $Z$ & $Z$ & $I$ & $I$ & $I$ & $I$ & $Z$ \\
          & $Z$ & $I$ & $Z$ & $I$ & $I$ & $Z$ & $I$ \\
    \bottomrule
  \end{tabular}
  \caption{The $7$-qubit Steane code.}
  \label{fig:css}
\end{figure}

The next lemma establishes some basic properties of the Steane code (shared by any CSS code that can correct single-qubit errors).

\begin{lemma}\label{lem:decouple}
  Consider the $7$-qubit Steane code (Figure~\ref{fig:css}).
  Let $\sE_1,\ldots,\sE_7,\sF_1, \sF_1'$ be qubit registers.
  Let $\sE = \sE_1 \cdots \sE_7$.
  Let $\sR$ be a register of arbitrary dimension.
  \begin{enumerate}
  \item
    There exists a unitary $U$ acting on registers $\sE_2 \cdots \sE_7
    \sF_1 \sF_1' \sX$ and a state $\ket{\tau}$ such that for all states
    $\ket{\psi}_{\sE_1 \cdots \sE_7 \sR}$ such that
    $\Tr_{\sR}(\ketbra{\psi}{\psi})$ is in the code space,
    \[
      U \big(\ket{\psi}_{\sE_1 \cdots \sE_7 \sR} \otimes \ket{0}_{\sF_1 \sF_1'
        \sX}\big)\, =\, \ket{\psi}_{\sF_1 \sE_2 \cdots \sE_7 \sR} \otimes
      \ket{\tau}_{\sE_1 \sF_1' \sX}\;.
    \]
    Moreover, the reduced density matrix of $\ket{\tau}$ on $\sE_1$ is
    the maximally mixed state on one qubit.
  \item For $W \in \{X, Z\}$ let $\cL_W$ denote a logical $W$ operator
    that does not act on $\sE_1$.
    For all states $\ket{\psi}$ on $\sE_1 \cdots \sE_7$ that lie in the code
    space,
    \begin{align*}
      U \big(\cL_W \, \ket{\psi}_{\sE_1 \cdots \sE_7} \otimes \ket{0}_{\sF_1 \sF_1' \sX}\big)
      &= \cL_W U \, \big(\ket{\psi}_{\sE_1 \cdots \sE_7}
        \otimes \ket{0}_{\sF_1 \sF_1' \sX}\big)\\
      &= \big( \cL_W \ket{\psi}_{\sF_1 \sE_2 \cdots \sE_7 \sR} \big)
        \otimes \ket{\tau}_{\sE_1 \sF_1' \sX}\;.
    \end{align*}
  \end{enumerate}
\end{lemma}

\begin{proof}
  We first establish item 1.
  Since the Steane code is a quantum error-correcting code that can
  correct any one qubit error, there exists a unitary $U$ that acts on
  registers $\sE_2 \cdots \sE_7$ and ancilla registers $\sF_1 \sF_1' \sX$
  and can correct an erasure error in the register
  $\sE_1$.
  Since the $7$-qubit code can correct any single qubit erasure, the
  resulting state on registers $\sF_1 \sE_2 \cdots \sE_7$ is the original
  state $\Tr_{\sR}(\ketbra{\psi}{\psi})$.
  Formally, let $\ket{\comp{0}}$ and $\ket{\comp{1}}$ denote the
  $7$-qubit encodings of $\ket{0}$ and $\ket{1}$, respectively.
  Since the code corrects any single-qubit erasure, for any $b \in
  \{0,1\}$, applying $U$ to the state $\ket{\comp{b}}_{\sE} \otimes
  \ket{0}_{\sF_1 \sF_1' \sX}$ yields a pure state $\ket{\theta}_{\sE \sF_1
    \sF_1' \sX}$ such that
  \[
    \Tr_{\sE_1 \sF_1' \sX} \Paren{ \ketbra{\theta}{\theta}} =
    \ketbra{\comp{b}}{\comp{b}}\;.
  \]
  Since $\ket{\theta}$ is pure, after rearranging registers we obtain that
  \begin{equation}\label{eq:u-theta}
    U\ket{\comp{b}}_{\sE} \otimes \ket{0}_{\sF_1 \sF_1' \sX}\,=\,
    \ket{\theta}_{\sF_1 \sE_2 \cdots \sE_7 \sE_1 \sF_1' \sX} \,=\,
    \ket{\comp{b}}_{\sF_1 \sE_2 \cdots \sE_7} \otimes \ket{\tau_b}_{\sE_1 \sF_1' \sX}\;.
  \end{equation}
  Now we establish two claims: (1) $\Tr_{\sF_1'
    \sX}(\ketbra{\tau_b}{\tau_b})$ is the maximally mixed state on one
  qubit, and (2) $\ket{\tau_0} = \ket{\tau_1}$.
  The first claim follows from the fact that the reduced density
  matrix on one qubit of any code state of a CSS code that corrects
  single-qubit errors is maximally mixed.
  The second claim follows from the fact that if $\ket{\tau_0} \neq
  \ket{\tau_1}$, then $U$ would fail to correct an erasure error on the
  superposition $ \frac{1}{\sqrt{2}} \paren{ \ket{\comp{0}} +
    \ket{\comp{1}}}$.
  Now write
  \[
    \ket{\psi}_{\sE \sR} = \alpha_0 \ket{\comp{0}}_{\sE} \otimes \ket{\psi_0}_{\sR}
    + \alpha_1 \ket{\comp{1}}_{\sE} \otimes \ket{\psi_1}_{\sR}\;.
  \]
  Applying~\eqref{eq:u-theta},
  \[
    U \ket{\psi}_{\sE \sR} \otimes \ket{0}_{\sF_1 \sF_1' \sX} = \sum_b \alpha_b
    \ket{\comp{b}}_{\sF_1 \sE_2 \cdots \sE_7} \otimes \ket{\psi_b}_{\sR} \otimes
    \ket{\tau}_{\sE_1 \sF_1' \sX} = \ket{\psi}_{\sF_1 \sE_2 \cdots \sE_7 \sR} \otimes
    \ket{\tau}_{\sE_1 \sF_1' \sX}\;.
  \]
  This establishes item 1.
  of the lemma.

  To show item 2., we note that applying a logical operator $\cL_W$ to
  a code state $\ket{\psi}$, erasing the first qubit, and then
  performing error correction, yields the state $\cL_W \ket{\psi}$,
  except on a different set of registers.
\end{proof}

\subsection{Multi-qubit entanglement tests}
\label{sec:entanglement_tests}

In this subsection we present the $S$-qubit EPR test, which is an elementary test that aims to verify that two provers A and B share $S$ EPR pairs, on which they measure several commuting single- or two-qubit Pauli operators when asked to do so. This test uses as a primitive the Magic Square game, which is a two-prover nonlocal game that is a \emph{self-test} for two EPR pairs. We present the Magic Square game next.

\paragraph{The Magic Square game.}
The $3 \times 3$ matrix presented in Figure~\ref{fig:ms} is called the \emph{operator solution} for the Magic Square game. Each entry consists of the label for a two-qubit Pauli observable; the observables all commute within a row or a column. The product of the observables along every row and column is equal to $I$, except for the last column, which multiplies to $-I$. 
\begin{figure}[H]
\[
	\begin{bmatrix}
		XI	& IX & XX \\
		IZ  & ZI & ZZ \\
		XZ 	& ZX & YY
	\end{bmatrix}
\]
\caption{Operator solution for the Magic Square game}
\label{fig:ms}
\end{figure}

The Magic Square game is played as follows: the verifier randomly chooses one of the provers to be prover A, and the other to be prover B. The verifier then chooses a random row $r$ and column $c$ from the operator solution for the Magic Square game. Let $W$ denote the two-qubit Pauli observable in the intersection of $r$ and $c$. The verifier then chooses random Pauli observables $W_r, W_c$ from the row $r$ and column $c$, respectively. The pairs $(W,W_r)$ and $(W,W_c)$, both formatted in lexicographic order, are sent to prover A and prover B, respectively.
For example, the verifier could select the first column and second row, and send observables $(IZ,XZ)$ to prover A and $(IZ,ZZ)$ to prover B.

The provers are required to respond with two-bit answers $a,b \in \{0,1\}^2$, respectively. The verifier checks that the bits in $a$ and $b$ that correspond to the common observable $W$ sent to both provers are equal.

\begin{definition}[Honest Magic Square strategy]
The \emph{honest Magic Square strategy} $\strat$ is such that the shared state $\ket{\psi}$ is two EPR pairs (i.e. $\ket{\psi}=\frac{1}{2} \Paren{ \ket{00} + \ket{11}}^{\otimes 2}$), and when a prover receives a pair of labels for commuting two-qubit Pauli observables, they measure the observables on their half of the EPR pairs and respond with the two bit outcome.
\end{definition}

\begin{theorem}[Magic Square test, Theorem 5.9
  in~\cite{coladangelo2017robust}]
  \label{thm:ms-rigid}
  The Magic Square game satisfies the following properties:
  \begin{enumerate}
  \item (Completeness) The honest Magic Square strategy succeeds in the Magic Square (MS) game with probability $1$. 
  \item (Soundness) For any $\eps\geq 0$ there is a $\delta = O(\sqrt{\eps})$ such that
    any strategy with success probability at least $1-\eps$ in the
    game is $\delta$-isometric to the honest Magic Square strategy.
  \end{enumerate}
\end{theorem}

\paragraph{The EPR test.}
The $S$-qubit EPR test is described in Figure~\ref{fig:epr_test}. The test and its analysis are adapted from~\cite{chao2016test}. The provers in the test are denoted prover A and prover B. Furthermore, the provers each receive a triple of commuting two-qubit Pauli observables $(W^{(1)},W^{(2)},W^{(3)})$. (This is purposefully formatted as questions to the Honest Pauli Prover in Section~\ref{sec:single-pauli}.) 

The EPR test consists of two subtests, which check that the provers' measurements satisfy the Pauli commutation and anticommutation relations, respectively. The Magic Square game is used to test the anticommutation relations. In order for the EPR test --- as well as the other protocols presented in this section --- to be sound, we need to ensure that the provers cannot distinguish between the subtests. Thus we require a definition of a triple $(W^{(1)},W^{(2)},W^{(3)})$ that is \emph{compatible} with the Magic Square game. %

\begin{definition}
\label{def:ms_compat}
	A triple of commuting two-qubit Pauli observables $(W^{(1)},W^{(2)},W^{(3)})$ is \emph{MS-compatible} if at least two of the observables act on the same pair of qubits, and furthermore those two observables can occur together in a row or column in Figure~\ref{fig:ms}.
\end{definition}

In the EPR test (and the other protocols in this section) we require that any question to the provers is embedded in a uniformly random MS-compatible triple that is consistent with the question. For example, suppose the verifier samples the question $(X_1,Z_2, Z_4)$ to send to prover A where the subscripts indicate which qubits the observables are supposed to act on. This question can be embedded in, say, the MS-compatible triple $(X_1 I_2, I_1 Z_2, I_3 Z_4)$, which is then sent to prover A. Note that any commuting pair of two-qubit Pauli observables, where each single-qubit observable is taken from $\{I,X,Z\}$, can be embedded in an MS-compatible triple in several ways; it does not matter which MS-compatible triple is chosen for any particular question.

\begin{figure}[H]
  \begin{mdframed}[style=figstyle]
    The verifier performs each of the following with equal
    probability:
    \begin{enumerate}
    \item (Commutation test) 
    	\begin{enumerate}
			\item Select distinct $i,j\in\{1,\ldots,S\}$ and let $W,W'\in\{X,Z\}$
      uniformly at random. 
      		\item Send the pair of single-qubit observables $(W_i,W'_j)$, embedded in an MS-compatible triple to prover A.
      		\item Send the two-qubit observable $W_i W_j'$, embedded in an MS-compatible triple to prover B.
      		\item Receive bits $(a,a',a'')$ from prover A and $(b,b',b'')$ from
      prover $B$.  Let $a,a'$ denote the answer bits corresponding to $W_i$ and $W_j'$ respectively, and let $b$ denote the answer bit corresponding to $W_i W_j'$.  Accept if and only if $a \oplus a'=b$.
      	\end{enumerate}
    \item (Anticommutation test) 
    	\begin{enumerate}
			\item Select distinct $i,j,\in\{1,\ldots,S\}$ and a pair of
      questions $(q,q')$ in the Magic Square game. Note that $q,q'$ both consist of a pair of commuting two-qubit Pauli observables. 
      		\item Send $q$ and $q'$, embedded in MS-compatible triples, to prover A and prover B, respectively. 
			\item Accept if and only if the provers' answers associated with the
      query $(q,q')$ would be accepted in the Magic Square game.
      	\end{enumerate}
    \end{enumerate}
  \end{mdframed}
  \caption{$S$-qubit EPR test~\cite{chao2016test}.}
  \label{fig:epr_test}
\end{figure}

\begin{definition}[Honest EPR strategy]\label{def:honest-epr}
An \emph{honest $S$-qubit EPR strategy} $\strat$ is a two-prover strategy that satisfies the following conditions. In the strategy the provers share the $S$-qubit maximally entangled state $
\ket{\Phi}_{\reg{A}\reg{B}}$, where prover A
has register $\reg{A}$ and prover B has register $\reg{B}$.
When sent an MS-compatible triple $(W^{(1)},W^{(2)},W^{(3)})$ of mutually commuting two-qubit Pauli observables, 
the prover returns the three bits obtained by simultaneously measuring the three
Pauli observables $\sigma_{W^{(1)}}$, $\sigma_{W^{(2)}}$ and $\sigma_{W^{(3)}}$ on its share of $\ket{\Phi}$.
\end{definition}

\noindent The following is a consequence of the results in~\cite{chao2016test}.

\begin{theorem}\label{thm:epr-test}
  The $S$-qubit EPR test (Figure~\ref{fig:epr_test}) has the following guarantees.
  \begin{itemize}
  \item (Complexity) Questions in the test are $O(\log S)$-bit long.
    Answers are $O(1)$-bit long.
  \item (Completeness) Any honest $S$-qubit EPR strategy succeeds with
    probability $1$ in the test.
  \item (Soundness) For any $\eps\geq 0$ there is a $\delta = \poly(S;\eps)$
    such that any strategy that succeeds with probability at least
    $1-\eps$ in the test is $\delta$-isometric to a honest $S$-qubit EPR
    strategy.
  \end{itemize}
\end{theorem}

\subsection{Simulated Pauli Prover game}
\label{sec:sim-game-desc}

Let $\cV_H$ be a verifier for a $(k+1)$-prover $S$-qubit Honest Pauli Prover game
$\Game_H$ satisfying some special properties that will be specified
later (these properties are satisfied by the verifier
$\cV^{}_{H,n}$ introduced in Section~\ref{sec:single-pauli}). Assume $k\geq 7$.

We define a $k$-prover ENL game $\Game_S$ that simulates $\Game_H$.
Label the provers in $\Game_S$ as $P_1,\ldots,P_k$.
Of the $k$ provers, the first seven, $P_1,\ldots,P_7$, are
chosen to be the ``simulated Pauli provers.''
The idea is that the provers $P_1,\ldots,P_7$ are supposed to share the
state of $PV$ where each qubit is encoded using the $7$-qubit Steane
code, and prover $P_i$ holds the $i$-th share of each encoded qubit.
When in game $\Game_H$, $PV$ is asked to measure a certain Pauli
observable, in game $\Game_S$ the simulated Pauli provers are asked to
implement a logical observable on their share of the encoding.
In addition, the prover is sent its own question, as in $\Game_H$, and
asked to provide an answer.
Since, in contrast to $PV$, none of the provers $\set{P_i}$ in $\Game_S$ are trusted, the verifier in $\Game_S$
executes a sub-test (called \emph{Stabilizer Check}) to ensure that
the simulated Pauli provers do indeed share an encoding of some state
(on some sub-registers), and measure a Pauli observable when asked to
do so.

The game $\Game_S$ is described in Figure~\ref{fig:simpauli_check}.
In the game questions are of the form $(W,g)$ where $W$ is called
an ``EPR question'' (i.e. is an MS-compatible triple that could arise in the EPR test) and $g$ is a ``$\Game_H$
question'' (i.e. a
question that is asked in the game $\Game_H$).
The provers reply with answers $(A,a)$ where $A$ is the answer to the
EPR question and $a$ is the answer to $g$.
We use $q_i$ to denote the $i$-th prover's question in $\Game_S$.

Let $g_P$ be a $\Game_H$ question.
For an answer $A = (A^{(1)},A^{(2)},A^{(3)})$ to an (MS-compatible) EPR question $W$
that contains $g_P$ (which we denote by $g_P \subseteq W$), let $A|_{g_P}$
denote the projection of $A$'s three bits to those that correspond to
$g_P$.
If $g_P = \bot$, then $A|_{g_P}$ is defined to be $0$.

The description of $\Game_S$ in Figure~\ref{fig:simpauli_check} involves notions of ``composite query''
and ``composite answer'' that are defined as follows.
Let $H$ be the generator matrix corresponding to the Steane code
described in Figure~\ref{fig:css}.

\begin{definition}[Composite queries and answers]\label{def:queries}
  Let $W$ be an $S$-qubit Pauli observable.
  \begin{enumerate}
  \item The composite query associated with $W$, denoted $\comp{W}$,
    is obtained by sending each prover forming the composite prover
    the question $W$.
  \item Given answers $(A_i)_{i\in\{1,\ldots,k\} \setminus \{t\}}$ from the $6$
    provers forming the composite prover, the composite answer
    $\comp{A}$ is obtained by selecting a uniformly random vector $v$
    in the column span of $H$ such that $v_t=1$, and computing the sum
    $\comp{A} = \sum_{i \in \{1,\ldots,7\} \setminus \{t\}} v_i A_{i}$.
  \end{enumerate}
\end{definition}

\begin{figure}[H]
  \centering
  \begin{mdframed}[style=figstyle]
    Let $\Game_H$ denote a $(k+1)$-prover Honest Pauli Prover game such that $k\geq 7$. \\
    The first $7$ of the $k$ provers are designated the
    ``simulated Pauli prover''. \\
    The verifier in $\Game_S$ perform one of the following tests, each
    chosen with equal probability:
    \begin{enumerate}
    \item \emph{(Stabilizer Check)}
      \begin{enumerate}
      \item Pick $t \in \{1,\ldots,7\}$ uniformly at random.
        Prover $P_t$ is designated the ``special prover''.
        The other provers $\{P_1,\ldots,P_7\}\setminus \{ P_t \}$ are jointly
        referred to as the ``composite prover''.
        A prover is not told whether it is the special prover, or a
        composite prover.
      \item Generate a query $(W,W')$ in the $S$-qubit EPR test, and
        for $i \in \{1,\ldots,k\}$ independently sample a question $g_i$
        according to the marginal distribution of the $i$-th prover's
        question in $\Game_H$.
      \item Set $q_t = (W,g_t)$ and $q_i = (W',g_i)$ for each $i \in
        \{1,\ldots,7\} \setminus \{t\}$.
        For $i > 7$, set $q_i = (W'',g_i)$ where $W''$ is a random EPR
        question.
      \item Let $(A_i,a_i)$ denote the $i$-th prover's answer.
        Accept if and only if $(A_t,\comp{A})$ would be accepted in
        the EPR test, where $\comp{A}$ is the composite answer
        associated with $\{A_i\}_{i \neq t}$.
        (Answers to $\Game_H$ questions are ignored.)
      \end{enumerate}
    \item \emph{($\Game_H$ Simulation)}
      \begin{enumerate}
      \item Generate a query $Q = (g_P,g_1,\ldots,g_k)$ as in $\Game_H$.
        Let $i^* \in \{1,\ldots,k\}$ denote the index such that $g_{i^*} \neq
        \bot$ if it exists.
        If it doesn't, set $i^* = 1$.
      \item Let $W$ be a uniformly random MS-compatible triple that contains $g_P$.
      \item For all $i \in \{1,\ldots,k\}$, if $g_i = \bot$ set $q_i =
        (W,\tilde{g}_i)$, where $\tilde{g}_i$ is uniformly random
        question sampled from the marginal distribution of the $i$-th
        prover's question in $\Game_H$.
        If $g_i \neq \bot$ set $q_i = (W_i,g_i)$, where $W_i$ is a
        uniformly random EPR question.
      \item Let $v\in\{0,1\}^7$ be such that $\sigma_X(v)$ and $\sigma_Z(v)$
        are logical operators for the $7$-qubit code, and moreover
        $v_{i^*}=0$.
      \item Let $(A_i,a_i)$ denote the $i$-th prover's answer.
        Let $A = \sum_{i\in\{1,\ldots,7\}} v_i A_i$.
        Accept if and only if $(A|_{g_P},a_1,\ldots,a_k)$ would be
        accepted in $\Game_H$.
      \end{enumerate}
    \end{enumerate}
  \end{mdframed}
  \caption{$k$-prover ENL game $\Game_S$.}
  \label{fig:simpauli_check}
\end{figure}

For a label $W \in \{X, Z \}$, an integer $i \in [S]$, and bit $A \in
\{0,1\}$, let $\sigma^A_{W_i}$ denote the projector $\frac{1}{2} ( \Id +
(-1)^A \sigma_{W_i})$. We first analyze the Stabilizer Check of the game $\Game_S$. We show that succeeding in the Stabilizer Check with high probability enforces that the provers hold a state that is encoded using the Steane code, and furthermore they apply honest Pauli measurements. This type of rigidity statement is common to the works of~\cite{ji2015classical,ji2017compression,NV17,NV18}.

\begin{definition}[Honest Stabilizer Check strategy]\label{def:honest-stabilizer}
A strategy $\strat = (\ket{\psi},\{M_i\})$ is an \emph{honest Stabilizer Check
  strategy} (implicitly, for code $\mathcal{C}$) if the following holds.
\begin{itemize}
\item The state $\ket{\psi}$ is on registers $\sC, \sP_1, \ldots, \sP_k$, and a reference register $\sR$, where for each $i\in\{1,\ldots,k\}$,   $\sP_i = \sE_i \sA_i$ with $\sE_i$ a register of $S$ qubits
  labeled $\sE_{i1},\ldots,\sE_{iS}$.
\item For $j \in \{1,\ldots,S\}$, the reduced density matrix $\rho_{\sE_{1j}
    \cdots \sE_{7j}}$ of $\ket{\psi}$ is in the code space of $\mathcal{C}$.
  We refer to $\sE_i$ as the $S$ ``code qubits'' of prover $P_i$.
\item Let $\set{ M_i((W,g),(A,a))}$ denote the $i$-th prover's POVM
  for the question $(W,g)$, where $W = (W^{(1)},W^{(2)},W^{(3)})$ is an EPR question and $g$ is a $\Game_H$ question.
  Then
  \begin{equation}
    \label{eq:honest_stabilizer}
    \E_g \sum_a M_i((W,g),(A,a)) = \sigma_W^A\;,
  \end{equation}
  where the expectation is taken with respect to the marginal distribution of questions
  $g$ to the $i$-th prover in $\Game_H$ and $\sigma_W^A = \sigma_{W^{(1)}}^{A_1}
  \sigma_{W^{(2)}}^{A_2} \sigma_{W^{(3)}}^{A_3}$ is the product of the three
  commuting projectors corresponding to the Pauli observables $W$
  acting on $\sE_i$.
\end{itemize}
\end{definition}

\begin{lemma}[Rigidity for Stabilizer Check]\label{lem:stab-epr}
  The following properties hold for the Stabilizer Check (item 1. in Figure~\ref{fig:simpauli_check}). 
  \begin{enumerate}
  \item (Completeness) An honest Stabilizer Check strategy $\strat$
    passes the Stabilizer Check with probability $1$.
  \item (Soundness) For any $\eps\geq 0$ there is a $\delta = \poly(S;\eps)$
    such that any strategy $\strat$ that pass the Stabilizer Check
    with probability at least $1 - \eps$ is $\delta$-isometric to an
    honest Stabilizer Check strategy.
  \end{enumerate}

\end{lemma}

\begin{proof}
	We first show completeness. Let $\strat$ be an honest Stabilizer Check strategy. 
	Suppose without loss of generality that prover $1$ is selected to be the special
	prover, and provers $\{2,\ldots,7\}$ are chosen to form the composite prover. In the Stabilizer Check, the EPR test is executed between the special prover and the composite prover; thus $\strat$ can then be viewed as a two-prover strategy in the EPR test, where the special prover measures the Pauli observables corresponding to its EPR question on its share of the shared state $\ket{\psi}$, generating a triple of bits $A \in \{0,1\}^3$ as its answer. The composite prover performs the Pauli measurements of provers $P_2,\ldots,P_7$ on registers $\sE_2,\ldots,\sE_7$, generating $6$ strings $A_2,\ldots,A_7 \in \{0,1\}^3$. Assume without loss of generality that the composite answer is the sum $\comp{A} = A_2 + A_3 + A_4$ modulo $2$ (this corresponds to selecting the vector $v = 1111000$ in the column span of the generator matrix $H$ corresponding to the Steane code). 
	
	It is straightforward to verify that this two-prover strategy passes the EPR test with probability $1$. Suppose first that the commutation subtest of the EPR test is chosen by the verifier, and let $i,j,W_i,W_j'$ be as in Figure~\ref{fig:epr_test}. Then the special prover measures $\sigma_{W_i}(i)$ and $\sigma_{W_j'}(j)$ on registers $\sE_{1i}$ and $\sE_{1j}$ of $\ket{\psi}$ to obtain answer bits $a$ and $a'$, respectively. The composite prover independently measures $\sigma_{W_i}(i) \otimes \sigma_{W_j'}(j)$ on registers $\sE_{2i} \sE_{2j}$,  $\sE_{3i} \sE_{3j}$, and $\sE_{4i} \sE_{4j}$ to obtain answer bits $a_2, a_3,a_4$ which then form the composite answer $\comp{a} = a_2 + a_3 + a_4$. Since Pauli observables $\sigma_{W_i}(i)^{\otimes 4}$ acting on $\sE_{1i} \sE_{2i} \sE_{3i} \sE_{4i}$ and $\sigma_{W_j'}(j)^{\otimes 4}$ acting on $\sE_{1j} \sE_{2j} \sE_{3j} \sE_{4j}$ are stabilizers of the Steane code, this implies that $a + a' + \comp{a} = 0$, which is the condition checked in the EPR test. A similar argument holds for the anticommutation test.
	
\medskip

Next we show soundness of the Stabilizer Check. Fix a $t \in \{1,\ldots,7\}$, and condition on prover $P_t$ being
  selected as the special prover.
  The provers' strategy $\strat$ is accepted in the Stabilizer Check
  with probability at least $1 - 7\eps$.
  From $\strat$ we construct a strategy $\strat_t'$ for the EPR test
  as follows.
  Let $(W,W')$ be the query received in the EPR test.
  When prover A receives question $W$, it generates a uniformly random
  $\Game$ question $g_t$ for the $t$-th prover, and plays according to
  the special prover $P_t$'s strategy on question $(W,g_t)$.
  For prover B we combine the strategies of the six provers that make
  the composite prover (including the post-processing involved in
  computing the composite answer $\comp{A'}$).
  Prover B simulates the measurements of the six provers on
  $(W',g_i)$ where $g_i$ is a random $\Game_H$ question for the $i$-th
  prover, for $i = \{1,\ldots,7\} \setminus \{t\}$.

  The resulting two-prover strategy succeeds in the EPR test with
  success probability $1-7\eps$.
  Applying the soundness analysis of the EPR test given in
  Theorem~\ref{thm:epr-test} it follows that $\strat_t'$ is
  $\poly(S;\eps)$-isometric to an honest $S$-qubit EPR strategy.
  In particular, there is an isometry $V_t$ for the special prover,
  such that the special prover's measurement operator associated with the
  answer $A_t$ to the EPR question $W$, which is
  \[
    \E_{g_t} \sum_{a_t} M_t((W,g_t),(A_t,a_t))\;,
  \]
  is $\poly(S;\eps)$-close to the honest Pauli measurement operator $\sigma_W^A$, under
  $V_t$, on the $S$ qubits identified by the isometry.

  Applying this analysis for each $t \in \{1,\ldots,7\}$, we obtain an
  isometry $V_t$ for each prover under which their (marginalized) measurement operators are
  $\poly(S;\eps)$-close to the corresponding honest Pauli measurement operator. Let $\sE_{ij}$ denote the register that holds the $j$-th qubit of the $i$-th prover under the isometry. 
  
It remains to show that the shared state $\ket{\psi}$ (after application of the isometries $\set{V_t}$) is $\poly(S; \eps)$-close to the codespace of the Steane code. Let $\Pi$ denote the projector onto the $7$ qubit codespace of the Steane code. Observe that
	\begin{equation}
	\label{eq:stabilizer}
  		\Pi = \E_h h\;,
	\end{equation}
  where the expectation is over a uniformly random stabilizer element $h$ of the Steane code. %
  Using that the stabilizer elements of the Steane code (or any CSS code) are Hermitian and form a group, it is immediate to verify that the expectation in~Equation~\eqref{eq:stabilizer} define a projection; by definition the codespace is the eigenvalue-1 eigenspace of the projection. For $j\in\{1,\ldots,S\}$ let $\Pi_j$ (resp. $h_j$) denote projector onto the codespace (resp. the stabilizer $h$) of the Steane code that acts on registers $\sE_{1j} \cdots \sE_{7j}$.
  
  Let $\ket{\psi'} = \bigotimes_t V_t \ket{\psi}$. Succeeding with probability at least $1 - \eps$ in the Stabilizer Check test implies that for all $j \in \{1,\ldots,S\}$, we have that $\ket{\psi'}$ is \emph{approximately stabilized} by the stabilizers of the Steane code:
  \[
  	\E_{h_j} \Norm{ h_j \ket{\psi'} - \ket{\psi'}} \leq \poly(S;\eps)\;,
  \]
from which it follows that 
  \[
 	\Norm{ \Pi_j \ket{\psi'} - \ket{\psi'}}  = \Big\| \E_{h_j} h_j \ket{\psi'} - \ket{\psi'}\Big\| \leq \poly(S;\eps)\;.
  \]
  By a hybrid argument, this implies that
  \[
  	\Big\| \bigotimes_j \Pi_j \ket{\psi'} - \ket{\psi'}\Big\|\leq \poly(S; \eps)\;,
  \]
  which completes the proof.
\end{proof}

\begin{theorem}
  \label{thm:pauli_sim}
  Let $k \geq 7$ be an integer.
  Let $\Game_H$ be a $(k+1)$-prover $S$-qubit Honest Pauli Prover game
  that satisfies the following properties:
  \begin{enumerate}
  \item The distribution over queries $(g_P,g_1,\ldots,g_k)$ is such that
    for any $(g_P,g_1,\ldots,g_k)$ in the support, there is at most one
    $i^* \in \{1,\ldots,k\}$ such that $g_{i^*} \neq \bot$.
  \item For any query $(g_P,g_1,\ldots,g_k)$ the accept or reject decision
    of $\Game_H$ does not depend on the answer of prover $PP_i$, for
    all $i$ such that $g_i = \bot$.
    \item The distribution of $g_P$ is supported on sets of Pauli observables that can be embedded in MS-compatible triples (see Definition~\ref{def:ms_compat}).
  \end{enumerate}
  Let $\Game_S$ be the Simulated Pauli Prover game described in
  Figure~\ref{fig:simpauli_check}.
  Then the following hold.
  \begin{itemize}
  \item (Completeness) For all Honest Pauli Prover strategies $\strat_H$ in $\Game_H$ there exists a $k$-prover strategy $\strat$ in
    $\Game_S$ that succeeds with probability $\omega^*_{\strat_H}(\Game_H)$.
  \item (Soundness) For any $k$-prover honest Stabilizer Check
    strategy that succeeds in $\Game_S$ with probability at least
    $1-\eps$, there is a $(k+1)$-prover Honest Pauli prover strategy
    that is accepted with probability at least $1-2\eps$ in $\Game_H$.
  \end{itemize}
\end{theorem}

\begin{proof}
  The completeness part of the theorem is straightforward.

  We show soundness.
  Fix an honest Stabilizer Check strategy $\strat = (\ket{\psi},\{M_i\})$ for
  the $k$ provers in $\Game_S$ that has success probability at least
  $1-\eps$, for some $\eps\geq 0$.
  In the game $\Game_H$, the provers are labeled $PV, PP_1,\ldots,PP_k$.
  The honest Pauli prover is $PV$.
  Using the strategy $\strat$, we define an Honest Pauli strategy
  $\strat^H = (\ket{\psi}^H,\{M_i^H\})$ for the provers in $\Game_H$ as
  follows:
  \begin{itemize}
  \item $\rho^H$ is on registers $\sC, \sP^H_V, \sP^H_1, \ldots, \sP^H_k$, and $\sR$,
    where the honest Pauli prover $PV$ gets $\sP^H_V$, and prover
    $PP_i$ gets $\sP^H_i$ for $i \in \{1,\ldots,k\}$.
    The register $\sP^H_V$ is isomorphic to the union of
    $\sE_1,\ldots,\sE_7$ (i.e.
    it is $7S$ qubits).
    The register $\sP^H_i$ is isomorphic to $\sF_i \sA_i$.
    The reduced density $\rho^H$ of the state $\ket{\psi}^H$ on all registers except $\sR$ is equal to the state $\rho \otimes \sigma$, where $\rho$ is the reduced density of $\ket{\psi}$ on all registers but $\sR$, and $\sigma$ is the
    maximally mixed state on an ancilla register $\sF = \sF_1 \cdots \sF_7$
    that is isomorphic to $\sE = \sE_1 \cdots \sE_7$.
    The registers have been relabeled according to the scheme
    described in Figure~\ref{fig:relabel}.

    \begin{figure}[htb!]
      \begin{center}
        \begin{tabular}{ c  c }
          \toprule
          \textbf{Register in $\rho\otimes\sigma$} & \textbf{Register in $\rho^H$} \\
          \midrule
          $\sC$ 		& $\sC$  \\
          $\sE_1 \cdots \sE_7$ 			& $\sP^H_V$ \\
          $\sA_i \sF_i$			& $\sP_i^H$\\
          \bottomrule

        \end{tabular}
      \end{center}
      \caption{Relabeling the registers of $\rho \otimes \sigma$ to get $\rho^H$.}
      \label{fig:relabel}
    \end{figure}

    In other words, the honest Pauli prover is given the $S$ code
    qubits held by each of the $7$ provers that constitute the
    simulated Pauli prover in $\Game_S$.
    The prover $PP_i$ in $\Game_H$ gets all the other qubits of prover
    $P_i$ in $\Game_S$, as well as the maximally mixed state in
    place of the $S$ qubits.

  \item On reception of a question $g_P$ in $\Game_H$ (which is a
    collection of up to three commuting Pauli observables), the honest
    Pauli prover $PV$ samples a random EPR question $W = (W^{(1)},
    W^{(2)}, W^{(3)})$ that contains $g_P$.
    The prover $PV$ measures the three logical observables $W^{(1)},
    W^{(2)}, W^{(3)}$ on the $7S$-qubit encoded state to obtain
    $(A_1^{(j)},\ldots,A_7^{(j)})$ for $j = 1,2,3$.
    Let $(A^{(1)}, A^{(2)},A^{(3)})$ be the decoded measurement
    outcomes.
    For example, $PV$ could apply the logical operator which has
    weight only on the last $3$ qubits and set $A^{(j)} = A_5^{(j)} +
    A_6^{(j)} + A_7^{(j)}$.
    The prover $PV$ returns $A|_{g_P}$.

  \item Suppose prover $PP_i$ in $\Game_H$ receives the question
    $g_i$.
    If $g_i = \bot$, then $PP_i$ returns $0$.
    The prover $PP_i$ samples a random EPR question $W_i$ that
    contains $g_P$.
    The prover $PP_i$ performs the same measurement that prover $P_i$
    would in game $\Game_S$ on question $(W_i,g_i)$.
    It obtains answer $(A_i,a_i)$ and returns $a_i$.

  \end{itemize}

  The following claim establishes that the answer distribution of the
  honest Pauli strategy $\strat^H$, when restricted to the
  ``relevant'' provers (i.e.
  the provers who receive questions that are not $\bot)$, is essentially
  the same as in the strategy $\strat$.

\begin{claim}\label{lem:same_dist}
  Fix a query $Q = (g_P,g_1,\ldots,g_k)$ in $\Game_H$.
  \begin{enumerate}
  \item If for all $i \in \{1,\ldots,k\}$ it holds that $g_i = \bot$, then
    the distribution of $A|_{g_P}$ that is produced by strategy
    $\strat$ in the ``{$\Game_H$ Simulation}'' part of $\Game_S$ when
    query $Q$ is sampled is the same as the distribution of $a_P$ that
    is produced by prover $PV$ in the strategy $\strat^H$ when it
    receives the question $g_P$.
  \item If there exists an $i^* \in \{1,\ldots,k\}$ such that $g_{i^*} \neq
    \bot$, then the distribution of $(A|_{g_P},a_{i^*})$ that is
    produced by strategy $\strat$ in the ``{$\Game_H$ Simulation}''
    part of $\Game_S$ is the same as the distribution of
    $(a_P,a_{i^*})$ that is produced by prover $PV$ and $PP_{i^*}$ in
    the strategy $\strat^H$ when they receive questions $g_P$ and
    $g_{i^*}$ respectively.
  \end{enumerate}
\end{claim}

We defer the proof of the claim to Section~\ref{sec:claimproof} and proceed with the proof of
Theorem~\ref{thm:pauli_sim}.
Since the strategy $\strat$ succeeds with probability at least $1 -
\eps$ in $\Game_S$, it succeeds with probability at least $1 - 2\eps$
in the $\Game_H$ Simulation part of $\Game_S$.

From our assumption on the game $\Game_H$, for a fixed $\Game_H$
question $Q = (g_P,g_1,\ldots,g_k)$ that is sampled in the $\Game_H$
Simulation part of $\cV_{sim}$, the accept or reject decision of
$\Game_H$ does not depend on $a_i$ if $g_i = \bot$.
Combined with the fact that at most one index $i^*$ is such that
$g_{i^*} \neq \bot$, Lemma~\ref{lem:same_dist} implies that the
distribution of ``relevant'' answers to $\Game_H$ are the same in the
following two scenarios when $Q$ is fixed: the strategy $\strat^H$ in
$\Game_H$, and the strategy $\strat$ in the $\Game_H$ Simulation part
of $\Game_S$.

Thus for a fixed $Q$, the probability that the ``relevant'' answers
are accepted by $\Game_H$ are the same in both scenarios.
Since the distribution of $Q$ is the same in both scenarios, this
implies that $\strat^H$ passes $\Game_H$ with probability at least $1
- 2\eps$.
\end{proof}

\subsection{Proof of Lemma~\ref{lem:same_dist}}
\label{sec:claimproof}

Part 1 of the claim follows directly from the fact that $\strat$ is an
honest Stabilizer Check strategy, in which the provers $P_1,\ldots,P_7$
measure the honest Pauli observables corresponding to a random EPR
question $W$ that contains $g_P$, which is identical to $PV$'s action
in the strategy $\strat^H$.

We now argue Part 2.
For an EPR question $W = (W^{(1)},W^{(2)},W^{(3)})$, we write $\sigma_W$
for the product $\sigma_{W^{(1)}} \sigma_{W^{(2)}} \sigma_{W^{(3)}}$.
For a three-bit vector $A = (A^{(1)},A^{(2)},A^{(3)})$, we write
$\sigma_W^A$ for the projector $\prod_{j=1}^3 \frac{\Id +
  (-1)^{A^{(j)}}}{2}$.
This is a projector because the Pauli observables $\sigma_{W^{(j)}}$ all
commute.

Assume without loss of generality that $i^* = 1$, and the string $v \in
\{0,1\}^7$ chosen by the verifier in $\Game_S$ is $v = 0000111$.
Let $W$ be a fixed EPR question that contains $g_P$.
For $j\in\{1,2,3\}$ let
\[
  \cL_{W^{(j)}} \,=\, \sigma_{W^{(j)}}(v)
\]
denote the logical operator corresponding to $W^{(j)}$ which is a tensor product of two
logical operators (since $W^{(j)}$ is the label for a two-qubit Pauli
observable).

For notational clarity we write $g = g_{i^*}$ and $a = a_{i^*}$.
Let $A_i = (A^{(1)}_i,A^{(2)}_i,A^{(3)}_i)$ denote the three bits
returned by prover $P_i$ for its EPR question, and let $A^{(j)} =
A^{(j)}_5 + A^{(j)}_6 + A^{(j)}_7$ denote the $j$-th bit of the answer
vector $A$, as computed by the verifier.

Let $M_{g}^a = \E_{W_1} \sum_{A} M_1((W_1,g),(A,a))$ denote $P_1$'s
measurement on question $g$, where we have marginalized the EPR
question (which was chosen independently of $W$) and the associated
answers.

We compute the probability of the answer pair $(A,a)$ in $\strat$ when
prover $P_1$ gets the question $(W_1,g)$ for a uniformly random EPR
question $W_1$, provers $P_5,P_6,P_7$ get the EPR question $W$, and
each prover gets an independently chosen random $\Game_H$ question.
Since $\strat$ is an honest Stabilizer Check strategy, the measurement
operator each prover applies (when marginalizing over the prover's
answer to its $\Game_H$ question) is given
by~\eqref{eq:honest_stabilizer}.
By our choice of $v$, the outcome $(A,a)$ occurs with probability
\begin{align}
  \label{eq:output_dist}
  &\sum_{A_5 + A_6 + A_7 = A} \Tr_\rho \Paren{ M_{g}^a \otimes \sigma_W^{A_5}
    \otimes \sigma_W^{A_6} \otimes  \sigma_W^{A_7}} \\
  &= \Tr_\rho \Paren{ M_{g}^a \otimes \prod_{j=1}^3
    \Paren{\frac{\Id + (-1)^{A^{(j)}} \cL_{W^{(j)}}}{2} }}\;.
\end{align}
Expanding the product, we obtain eight terms of the form
\[
  \pm \frac{1}{8} \Tr_\rho \Paren{ M_{g}^a \otimes \cL_{D}}\;,
\]
where $\cL_D$ is a product of up to three logical operators $\set{
  \cL_{W^{(j)}}}$.
The label $D$ indicates a collection of up to six Pauli observables
(for example, $\cL_D = \cL_{W^{(1)}}\cL_{W^{(2)}}\cL_{W^{(3)}}$ where
each $W^{(j)}$ is a label for a two-qubit Pauli observable).

Fix one of the possible labels $D$.
Let $U$ be the unitary given by Lemma~\ref{lem:decouple}.
Since $\strat$ is an honest Stabilizer Check strategy, $\rho_{\sE_{1j} \cdots
  \sE_{7j}}$ is in the code space for all $j \in \{1,\ldots,S\}$.
Let $\sF_1, \sF_1'$ be registers isomorphic to $\sE_1$, and let $\sX$
be an ancilla register that is sufficiently large.
Applying part 1.
of Lemma~\ref{lem:decouple} we get
\begin{equation}
  \label{eq:correct}
  U^{\otimes S} \rho \otimes \ketbra{0}{0}_{\sF_1 \sF_1' \sX} (U^{\otimes S})^\dagger =
  \rho_{\sF_1 \sE_2 \cdots \sE_7} \otimes \ketbra{\tau_S}{\tau_S}_{\sE_1 \sF_1' \sX}\;,
\end{equation}
where $\ket{\tau_S}$ is the $S$-fold tensor product of the state
$\ket{\tau}$ given by Lemma~\ref{lem:decouple}.
Here, the $j$-th tensor factor of $U^{\otimes S}$ acts on registers
$\sE_{2j} \cdots \sE_{7j} \sF_{1j} \sF_{1j}' \sX_j$.
Then
\begin{align*}
  \Tr_\rho \Paren{ M_g^a \otimes \cL_D}
  & = \Tr \Paren{  (M_g^a \otimes \cL_D) \paren{ \rho \otimes
    \ketbra{0}{0}_{\sF_1 \sF_1' \sX} } (U^{\otimes S})^\dagger (U^{\otimes S}) } \\
  & = \Tr \Paren{  M_g^a U^{\otimes S} \,\, \cL_D  \paren{ \rho
    \otimes \ketbra{0}{0}_{\sF_1 \sF_1' \sX} } (U^{\otimes S})^\dagger} \\
  & = \Tr \Paren{  (M_g^a \otimes \cL_D) \,\, U^{\otimes S}  \paren{ \rho
    \otimes \ketbra{0}{0}_{\sF_1 \sF_1' \sX} } (U^{\otimes S})^\dagger} \\
  & = \Tr \Paren{  (M_g^a \otimes \cL_D) \Paren{ \rho_{\sA \sF_1 \sE_2 \cdots
    \sE_7} \otimes \ketbra{\tau_S}{\tau_S}_{\sE_1 \sF_1' \sX}}} \\
  &= \Tr \Paren{ (M_g^a \otimes \cL_D)
    \Paren{\rho_{\sA \sF_1 \sE_2 \cdots \sE_7} \otimes \sigma_{\sE_1}}}\;,
\end{align*}
where $\sigma_{\sE_1}$ is the maximally mixed state on $\sE_1$.
The second equality follows from the cyclicity of the trace and the
fact that $U$ and $M_g^a$ act on different registers.
The third equality follows from part 2 of Lemma~\ref{lem:decouple}.
The fourth equality follows from~\eqref{eq:correct}.
The last equality follows from the fact that the reduced density
matrix of $\ket{\tau_S}$ on $\sE_1$ is the maximally mixed state.

Thus the probability of obtaining outcome $(A,a)$ expressed
in~\eqref{eq:output_dist} is the same as
\begin{equation}
  \Tr \Biggl( \big(\rho_{\sA \sF_1 \sE_2 \cdots \sE_7} \otimes
  \sigma_{\sE_1} \big) \biggl( M^a_{g} \otimes \prod_j
  \frac{\Id + (-1)^{A^{(j)}} \cL_{W^{(j)}}}{2} \biggr) \Biggr)\;.
  \label{eq:dist}
\end{equation}
Here the operator $M^a_{g}$ acts on $\sA_1 \sE_1$.
Observe that the state $\rho^H = \rho_{\sA \sE_1 \sE_2 \cdots \sE_7} \otimes
\sigma_{\sF_1}$ and therefore~\eqref{eq:dist} is equal to
\[
  \Tr_{\rho^H} \Paren{M^a_{g} \otimes \prod_j \frac{\Id + (-1)^{A^{(j)}}
      \cL_{W^{(j)}}}{2} } \;
\]
where now we treat the operator $M^a_g$ as acting on registers $\sA_1
\sF_1$.
This quantity is precisely the probability that $(A,a)$ is obtained by
provers $PV$ and $PP_r$ in the strategy $\strat^H$ when given input
$g_P = W$ and $g_{i^*}$, respectively: the prover $PV$ measures the
registers $\sE_5, \sE_6, \sE_7$ using the observables $\cL_{W^{(1)}},
\cL_{W^{(2)}}, \cL_{W^{(3)}}$ and the prover $PP_{i^*}$ measures the
registers $\sA_{i^*} \sF_{i^*}$ with the POVM $\{ M_g^a \}$.
This establishes Part 2 of the claim.
\section{The Compression Theorem}
\label{sec:compression}

In this section we present the proof of our compression result,
informally stated as Theorem~\ref{thm:main-compression} in the
introduction, and formally re-stated here.

\begin{theorem}[Compression Theorem]
  \label{thm:compression}
  Let $k \geq 7$ be an integer, and let $G$ be a GTM for a family of
  $k$-prover ENL games $\{\Game_n\}$. Let $p(n)$ denote the size of $\CKT(G,n)$, the 
$n$-th  protocol circuit specified by $G$.
 There exists a family of $k$-prover ENL games $\{
  \Game^{\compr}_{n} \}$ such that the following holds, %
  for all integer $n$:
  \begin{enumerate}
  \item The verifier of $\Game^{\compr}_{n}$, denoted by $\cV^\compr_{n}$, is uniformly generated from $(1^n,G)$. 
  \item Each prover's answer in $\Game^{\compr}_{n}$ is $4$ bits long.
  \item There are universal constants $\alpha\geq 1,\beta > 0$ such that for $N=2^n$,
    \begin{equation}\label{eq:compression-0}
      1 - \frac{1 - \omega^{*}(\Game_N)}{p(N)+1} \leq
      \omega^*(\Game^\compr_{n}) \leq 1 - \Paren{\frac{1 -
          \omega^{*}(\Game_N)}{\beta \, p(N)}}^\alpha\;.
					\end{equation}
    
  \item There exists universal constants $\mu \geq 1, \nu > 0, C > 0$
    such that any strategy $\strat$ for $\Game_{n}^\compr$ that satisfies
    $\omega^*_\strat(\Game_{n}^\compr) \geq 1 - \eps$ for some $\eps\geq 0$ requires an entangled
    state such that the local dimension of registers associated with at least $7$ of the provers is at least $(1 - C \, p(N)^\mu\, \eps^\nu)
    2^{p(N)}$.
  \end{enumerate}
\end{theorem}

To make the dependence of the games $\{\Game_n^\compr\}$ on the GTM $G$ more explicit, in subsequent sections we  use the notation $\Game_{G,n}^\compr$ and $\cV_{G,n}^\compr$ to denote the game and verifier associated with $G$ in Theorem~\ref{thm:compression}.

\begin{proof}%
The proof combines the results of the Section~\ref{sec:single-pauli} and Section~\ref{sec:sim-pauli}. Let $S=p(N)$ and $\Game^{\compr}_{H,n}$ the  $S$-qubit $(k+1)$-prover Honest Pauli Prover game obtained from $G$ as described in Figure~\ref{fig:honest}. Observe that $\Game^\compr_{H,n}$ satisfies the properties required by Theorem~\ref{thm:pauli_sim}. Let $\Game^{\compr}_n$ denote the $S$-qubit Simulated Pauli Prover game obtained from $\Game^{\compr}_{H,n}$ as described in Figure~\ref{fig:simpauli_check}. Let $\cV^{\compr}_{H,n}$ and $\cV^{\compr}_n$ denote the verifiers of $\Game^{\compr}_{H,n}$ and $\Game^{\compr}_n$, respectively. The verifiers $\cV^{\compr}_{H,n}$ and $\cV^{\compr}_n$ depend on the GTM $G$, but we leave the dependence implicit. 

By inspecting each of the subprotocols of the Honest Pauli Prover game presented in Section~\ref{sec:single-pauli}, it is not hard to verify that the family of verifiers $\Set{\cV^{\compr}_{H,n}}$ for the games 
$\Set{\Game^{\compr}_{H,n}}$ is uniformly generated from $(1^n,G)$. Inspecting the protocols in Section~\ref{sec:sim-pauli}, it follows that the family of verifiers $\Set{\cV^{\compr}_n}$ for the games $\Set{\Game^{\compr}_n}$ is uniformly generated from $(1^n,G)$ as well.  This establishes the first item of the theorem. 

The second item follows since answers in $\Game^{\compr}_n$ consist of $3$ bits, to answer the EPR question, and $1$ bit, to answer the $\Game^{\compr}_H$ question.

We show the third item. The completeness statements of Lemma~\ref{lem:single_pauli_game} and Theorem~\ref{thm:pauli_sim} imply that for any $\gamma > 0$ there exists a strategy $\strat$ in $\Game^\compr_{n}$ that succeeds with probability at least $1 - \frac{1 - \omega^*(\Game_N) + \gamma}{p(N) + 1}$. Using that $\omega^*(\Game^\compr_n)$ is defined as a supremum over strategies, taking the limit $\gamma \to 0$ shows the lower bound in~\eqref{eq:compression-0}. 

For the upper bound, consider a $k$-prover strategy $\strat$ for
$\Game^{\compr}_n$ that succeeds with probability $1 - \eps$, for some $\eps \geq 0$. Then $\strat$ passes the Stabilizer Check subroutine of $\Game^{\compr}_n$ (see Figure~\ref{fig:simpauli_check})
with probability at least $1 - 2\eps$. By Lemma~\ref{lem:stab-epr}, $\strat$ is $\poly(S; \eps)$-isometric to
an honest Stabilizer Check strategy $\strat'$.
Applying Lemma~\ref{lem:close_strategies}, it follows that the strategy $\strat'$ succeeds in $\Game^{\compr}_n$ with probability at least $1 - \poly(S;\eps)$.

Observe that $\Game^{\compr}_{H,n}$ is a Honest Pauli Prover game that satisfies
the properties required for the application of Theorem~\ref{thm:pauli_sim}, and that by definition $\Game_n^\compr$ is the simulated game associated with $\Game^\compr_{H,n}$.
It follows from the soundness part of the theorem that there exists a $(k+1)$-prover Honest Pauli strategy $\strat''$ such that 
\begin{equation}\label{eq:compression-1}
\omega^*_{\strat''}(\Game^{\compr}_{H,n}) \,\geq\,1 - \poly(S;
\eps)\;.
\end{equation}
Moreover, using that $\strat''$ is a Honest Pauli strategy, from Lemma~\ref{lem:single_pauli_game} we get
\begin{equation}\label{eq:compression-2}
  \omega^*_{\strat''}(\Game^{\compr}_{H,n}) \,\leq\, 1 -
  \Paren{\frac{1 - \omega^*(\Game_N)}{\beta' \, p(N)}}^{\alpha'}\;,
\end{equation}
for universal constants $\alpha' \geq 1,\beta' > 0$.
Combining~\eqref{eq:compression-1} and~\eqref{eq:compression-2}, since $\eps = 1-\omega^*(\Game_n^{\compr})$ and $S=p(N)$, it follows that
\[
  \omega^*(\Game^{\compr}_n) \leq 1 - \Paren{\frac{1 -
      \omega^*(\Game_N)}{\beta\, p(N)}}^{\alpha}\;,
\]
for some universal constants $\alpha > 1,\beta > 0$.

Finally we show the fourth item in the theorem.%
As shown in the course of the proof of the third item, any strategy $\strat$ for $\Game_n^{\compr}$ that is accepted with probability at least $1 - \eps$, for some $\eps\geq 0$, is $\delta$-isometric to an honest Stabilizer Check strategy $\strat'$, for some $\delta = \poly(S;\eps)$. By definition the provers in an honest Stabilizer Check strategy share a state $\ket{\psi}$ such that for any $i\in\{1,\ldots,S\}$ the reduced density of $\ket{\psi}$ on registers $\sE_{i1},\ldots,\sE_{i7}$, held by provers $P_1,\ldots,P_7$ respectively, is a $7$-qubit state supported on the codespace. 
Applying item 1. from Lemma~\ref{lem:decouple} independently to each of the $S$ reduced densities, it follows that for any $t\in\{1,\ldots,7\}$ the reduced density of $\ket{\psi}$ on register $\sE_t=\sE_{1t}\ldots\sE_{St}$ is the totally mixed state on $S$ qubits. 
Using the definition of $\delta$-isometric strategies, it follows that for every $t \in\{ 1,\ldots,7\}$ there exists an isometry $V_t$ mapping register $\sE_t$ to registers $\sA \sA'$, and an isometry $V_t'$ mapping registers $\set{\sE_j}_{j \neq t}$ to registers $\sB \sB'$, such that
\[
	V_t \otimes V_t' \ket{\psi}_{\sE_1 \cdots \sE_7 \sR} \approx_{\delta} \ket{\Phi}_{\sA \sB} \otimes \ket{\psi'}_{\sA' \sB' \sR}\;,
\]
where $\ket{\Phi}_{\sA \sB}$ is an $S$-qubit maximally entangled state between $\sA$ and $\sB$, and the state $\ket{\psi'}$ is arbitrary. Here, the notation $\approx_\delta$ indicates closeness in trace distance. Using that for any two pure states $\ket{\phi},\ket{\theta}$ it holds that $1 - \Norm{ \ketbra{\phi}{\phi} - \ketbra{\theta}{\theta} }_1 \leq | \ip{\phi}{\theta}|^2$, we obtain 
\begin{equation}
\label{eq:epr-ip}
	\Big| \Paren{ \bra{\Phi}_{\sA \sB} \otimes \bra{\psi'}_{\sA' \sB' \sR}} \Paren{V_t \otimes V_t' \ket{\psi}_{\sP_1 \cdots \sP_7 \sR}}\Big|^2 \geq 1 - \delta\;.
\end{equation}
If  $\ket{\theta}_{\sA \sA' \sB \sB' \sR}$ is an arbitrary pure state with Schmidt rank at most $r$ along the cut that separates the registers $\sA \sA'$ and $\sB \sB' \sR$, then 
using that all Schmidt coefficients of $\ket{\Phi}_{\sA \sB} \otimes \ket{\psi'}_{\sA' \sB' \sR} $ along the same cut are at most $2^{-S/2}$ it follows that
\begin{equation}
\label{eq:epr-r}
	\Abs{ \Paren{ \bra{\Phi}_{\sA \sB} \otimes \bra{\psi'}_{\sA' \sB' \sR}} \Paren{\ket{\theta}_{\sA \sA' \sB \sB' \sR}}}^2 \leq r2^{-S}\;.
\end{equation}
Inequalities~\eqref{eq:epr-ip} and~\eqref{eq:epr-r} imply that the Schmidt rank of $V_t \otimes V_t' \ket{\psi}_{\sP_1 \cdots \sP_7 \sR}$ between prover $t$ and the other provers is at least $(1 - \delta) 2^{p(N)}$. Since the isometries $V_t$ and $V_t'$ cannot increase the Schmidt rank between prover $t$ and the other provers as well as the reference system $\sR$, the same lower bound holds for the Schmidt rank of $\ket{\psi}$ between register $\sP_t$ and $\set{\sP_j}_{j \neq t} \sR$. Finally, since this lower bound holds for all $t = 1,\ldots,7$, this concludes the proof of item 4.

\end{proof}

\newcommand{\HCT}{``Hypothetical Compression Theorem''\,}

\section{Recursive compression of quantum interactive proofs}
\label{sec:recursive}
In this section we show how to apply the compression theorem,
Theorem~\ref{thm:compression} in Section~\ref{sec:compression},
recursively to prove Theorem~\ref{thm:main} and
Theorem~\ref{thm:main-undecidable} stated in the introduction.
Before doing so we introduce several definitions.

A function $t: \N \to \N$ is \emph{time-constructible} if there exists
an integer $m \geq 0$ and a deterministic Turing machine $T$ such that
for all $n \geq m$, the Turing machine halts on input $1^n$ after
exactly $t(n)$ steps.
Examples of time-constructible functions include $n, n^2, 2^n,
2^{2^n}$, and so on. 
Recall the iterated exponential function $\Lambda_R(n)$, defined
inductively by $\Lambda_0(n) = n$ for all integer $n\geq 0$, and for integer
$R\geq 0$, $\Lambda_{R+1}(n) = 2^{\Lambda_{R}(n)}$ for all integer $n\geq 0$.
We call the parameter $R$ the ``height'' of $\Lambda_R(n)$.

\begin{definition}
A time-constructible function $t(n)$ is \emph{hyper-exponential} if there exists a function
$R(n)$ such that $t(n) = \Lambda_{R(n)}(n)$.
\end{definition}

Note that with this definition, any hyper-exponential function $t$ satisfies $t(n)\geq n$ for all $n\geq 0$. 

\begin{definition}
  Let $t: \N \to \N$ be a time-constructible function.
  The language $\cL[t]$ consists of all pairs $(1^n,M)$ such that $M$
  is a nondeterministic Turing machine that halts on input $0$ within
  $t(n)$ steps.
\end{definition}

\noindent For any time-constructible $t$, the language $\cL[t]$ is
complete for $\NTIME[t]$ under polynomial-time Karp reductions. The following result from~\cite{NatarajanV17twoprover} will be used as the base case for our construction. It shows that  for $t(n)=2^n$  languages in $\cL[t]$ can be decided by a polynomial-size verifier in a two-prover nonlocal game. 

\begin{theorem}[The Natarajan-Vidick verifier~\cite{NatarajanV17twoprover}]\label{thm:nv}%
  There is a universal constant $\delta>0$ and a family of verifiers  $\{\cV_{NV}(M,n)\}$ that is uniformly generated from $(1^n,M)$ such that for any integer $n$ and nondeterministic Turing machine $M$ the following hold. The game 
  $\Game_{NV}(M,n)$ associated with $\cV_{NV}(M,n)$ is a two-prover nonlocal game   such that $\omega^*(\Game_{NV}(M,n)) = 1$ if $(1^n, M) \in
  \cL[2^n]$ and $\omega^*(\Game_{NV}(M,n)) \leq 1 - \delta$ otherwise.
\end{theorem}

\subsection{The main recursive compression result}
\label{sec:vrc}

The main result we prove in this section is the following.

\begin{proposition} \label{prop:vrc} Let $t: \N \to \N$ be a
  hyper-exponential function. Let $T$ be a deterministic Turing machine that halts in exactly $t(n)$ steps on input $1^n$. Let $M$ be a nondeterministic Turing machine. 
  There exists a family of $7$-prover ENL games $\Set{\Game_{n,M,T}}$ that is uniformly generated from $(1^n,M,T)$ and such that
  \begin{enumerate}
  \item The answer length of the provers is $O(1)$ bits.
  \item There exists universal constants $c , C> 0$ such that for all integer $n$, 
    \begin{align*}
      \omega^*(\Game_{n,M,T}) &= 1  \quad & \text{if } (1^n,M) \in \cL[2^{t}] \\
      \omega^*(\Game_{n,M,T}) &\leq 1 - Ct(n)^{-c} \quad &\text{if } (1^n,M) \notin \cL[2^{t}].
    \end{align*}
  \end{enumerate}
\end{proposition}

Before proving Proposition~\ref{prop:vrc} we show that it implies Theorem~\ref{thm:main}, which we reformulate for convenience. 

\begin{theorem}\label{thm:main2}
  There exists universal constants $c',C'>0$ such that for any
  hyper-exponential function $t: \N \to \N$,
  \[
    \NTIME[2^{t(n)}] \subseteq \MIP*_{1,1-C't^{-c'}}(15,1)\;.
  \]
\end{theorem}

\begin{proof}
Let $T$ be a deterministic Turing machine that halts in exactly $t(n)$ steps on input $1^n$.   Fix an instance $(1^n,M)$ of $\cL[2^t]$.
  Applying Proposition~\ref{prop:vrc} gives a $7$-prover game
  $\Game_{n,M,T}$ of size $\poly(n)$ such that $\omega^{*}(\Game_{n,M,T})
  = 1$ if $(1^n,M) \in \cL[2^t]$, and otherwise $\omega^{*}(\Game_{n,M,T}) \leq
  1 - Ct(n)^{-c}$ for some universal constants $c, C > 0$.

  To convert the game to an $\MIP*$ protocol, i.e.
  remove the provers' initial quantum message in the ENL game, we use the compression
  result of~\cite{ji2017compression} as a black box.
  This result provides an efficient method to transform any ENL game
  $\Game$ involving $k$ provers into a nonlocal game $\Game'$ of size (as measured by the verifier circuit) $\poly(|\Game|)$,
  involving $k + 8$ provers, with the following properties. If
  $\omega^*(\Game) = 1$, then $\omega^*(\Game') = 1$.
  Otherwise,
  \[
    \omega^*(\Game') \,\leq\, 1 - \Paren{\frac{1 -
        \omega^*(\Game)}{\poly(n)}}^{d}\,\leq\, 1- C'\,t(n)^{-c'}\;,
  \]
  for some universal constants $d,c',C' > 0$.%
  Here the second inequality uses that $t(n)=\Omega(n)$ for any hyper-exponential function $t$.
 Combining the two reductions gives a
  polynomial-time reduction from $\mathcal{L}[2^t]$ to $15$-prover
  nonlocal game $\Game'_{n,M,T}$.
\end{proof}

\medskip
\vspace{10pt}

To prove Proposition~\ref{prop:vrc}, we present and analyze a family of verifiers $\Set{\cV_{RC}(n,n_0,M,T,G)}$, specified in Figure~\ref{fig:vrc}. %
The verifiers are parametrized by two integers $n \geq n_0 > 0$, a nondeterministic Turing machine $M$, a deterministic Turing machine $T$, and a GTM $G$ that takes input $(n,t,\lambda)$. Here, think of $n_0$ as the input size, and $n$ as a parameter that indicates the size of $\cV_{RC}$. For the actual verifier used to define the game, $n=n_0$, but we may also consider the case where $n$ eventually grows very large. Roughly speaking,  if $n \geq t(n_0)$, the verifier $\cV_{RC}(n,n_0,M,T,G)$ simulates the Natarajan-Vidick protocol from Theorem~\ref{thm:nv} to determine whether $(1^n,M) \in \cL[2^t]$. Otherwise, if $n$ is smaller than $t(n_0)$, then $\cV_{RC}$ is ``too small'' to perform the simulation directly. In this case, $\cV_{RC}$ instead executes the compressed protocol associated with $\cV_{RC}(2^n,n_0,M,T,G)$, i.e. an exponentially bigger version of itself. 

\begin{figure}[H]
  \centering
  \begin{mdframed}[style=figstyle]
    \ul{Verifier name:} $\cV_{RC}(n,n_0,M,T,G)$ \\
    \ul{Description of parameters:} $n \geq n_0 > 0$ are integers, $M$ is a nondeterministic Turing machine, $T$ is a deterministic Turing machine, and $G$ is a GTM that takes input $(n,t,\lambda)$.
    \begin{enumerate}
    \item Run $T$ on input $1^{n_0}$ for $n$ steps.

    \item If $T$ halts in that time, then execute the verifier $\cV_{NV}(M,n)$ from Theorem~\ref{thm:nv}.
    \item Otherwise, execute the verifier $\cV^{\compr}_{G_\lambda,n}$ from Theorem~\ref{thm:compression}, where $\lambda = (n_0,M,T,G)$ and $G_\lambda(n,t) = G(n,t,\lambda)$. 
    \end{enumerate}
  \end{mdframed}
  \caption{The recursive compression verifier}
  \label{fig:vrc}
\end{figure}

It follows from Theorem~\ref{thm:compression} and  Theorem~\ref{thm:nv}  that the family of verifiers $\Set{\cV_{RC}(n,n_0,M,T,G)}$ can be uniformly generated from $(1^n,\lambda)$, where $\lambda=(n_0,M,T,G)$, by a Turing machine $R$. By Lemma~\ref{lem:gtm-combine}, there exists a GTM $G_R$ that takes input $(n,t,\lambda)$ and returns the $t$-th gate of the protocol circuit corresponding to the verifier $\cV_{RC}(n,n_0,M,T,G)$.\footnote{Strictly speaking, the protocol circuit corresponds to an \emph{equivalent} verifier to $\cV_{RC}$, but for clarity of exposition we will not distinguish between the verifier specified by $G_{R}$ and $\cV_{RC}$ itself.} 
For the remainder of the section we consider $M$ and $T$ as implicitly fixed, and write $\cV_{RC}(n,n_0)$ for $\cV_{RC}(n,n_0,M,T,G_R)$. Let $\Game_{n,n_0}$ denote the $7$-prover game specified by
$\cV_{RC}(n,n_0)$, and let $\omega_{n,n_0}^*$ denote $\omega^*(\Game_{n,n_0})$. Let $\Game_{n} = \Game_{n,n}$.

\medskip \vspace{10pt}

Due to its recursive nature the verifier $\cV_{RC}$ may be hard to
comprehend at first.
For concreteness, we go through an execution of the protocol specified
by the verifier for the choice of the time-constructible function $t(n) =
2^n$.
Thus, $T$ is a Turing machine that on input $1^n$ iterates for $2^n$
steps exactly, and then halts. $M$ is an arbitrary nondeterministic Turing machine, and $n_0$ a positive integer. 
The verifier $\cV_{RC}(n_0,n_0)$ specifies the actions of a verifier in a $7$-prover ENL
game $\Game_0$ that has size $\poly(n_0)$.
Following the description in Figure~\ref{fig:vrc}, the verifier in
$\Game_0$ performs the following actions. 
It first executes $T$ on input $1^{n_0}$ for $n_0$ steps.
By definition of $T$, since $n_0 < t(n_0) = 2^{n_0}$, the Turing
machine has not yet halted.
Thus the verifier proceeds to the second step in Figure~\ref{fig:vrc}: it executes another verifier, $\cV^{\compr}_{G_\lambda,n_0}$ from Theorem~\ref{thm:compression}. The verifier can compute the description of $\cV^\compr_{G_\lambda,n_0}$ in polynomial time given $1^{n_0}$ and the description of $G_\lambda$.

By construction (see the proof of Theorem~\ref{thm:compression}) the verifier $\cV^{\compr}_{G_\lambda,n_0}$
specifies a $7$-prover ENL game $\Game^{\compr}_{G_\lambda,n_0}$,
which checks that the provers hold (an encoding of) the
history state of the protocol circuit $\CKT(G_\lambda,2^{n_0})$. Let $n_1 = 2^{n_0}$. The protocol circuit $\CKT(G_\lambda,n_1)$ defines a verifier $\cV_{RC}(n_1,n_0)$ and a game $\Game_1 = \Game_{n_1,n_0}$. 
 Notice that $\Game_1$ is just as $\Game_0$, except that the first input is exponentially larger, from $n_0$ to $n_1$. 

Theorem~\ref{thm:compression} relates the value of $\Game_0$ to the value of $\Game_1$.
 So it suffices to analyze the value of $\Game_1$, which means analyzing $\cV_{RC}(n_1,n_0,M,T,G_R)$. Since $n_1 \geq 2^{n_0}$, $\Game_1$ reduces to the game
 $\Game_{NV}$ specified by the Natarajan-Vidick verifier $\cV_{NV}(M,n_1)$. By Theorem~\ref{thm:nv}, if $(1^{n_1},M) \in \cL[2^{n_1}]$, then the value of $\Game_{NV}(M,n_1)$ is $1$, which implies that
 $\omega^*(\Game_1) = 1$, which in turns implies that $\omega^*(\Game_0) = 1$.
 Otherwise if $(1^{n_1},M) \notin \cL[2^{n_1}]$, $\omega^*(\Game_1) = \omega^*(\Game_{NV}(M,n_1)) \leq 1 - \delta$, which implies
 that $\omega^*(\Game_0) \leq 1 - \frac{\delta^{\alpha}}{\poly(n_1)} \leq 1 -
 C2^{-cn_0}$ for some constants $c,C > 0$. %

 Observe now that $(1^{n_1},M) \in \cL[2^{n_1}]$ if and only if
 $(1^{n_0},M) \in \cL[2^{2^{n_0}}]$.
 This establishes Proposition~\ref{prop:vrc} for the special case $t(n)
 = 2^n$. We now give the proof for the general case.

\begin{proof}[Proof of Proposition~\ref{prop:vrc}]
	Since the answer sizes are constant in both the Natarajan-Vidick protocol, as well as the games produced by Theorem~\ref{thm:compression}, this establishes item 1. of the proposition. We now show item 2.
	
  Fix $n, M, T$.
  Since $t(n)$ is a hyper-exponential function, there exists a
  smallest integer $R \geq 0$ such that $\Lambda_R(n) = t(n)$ (note that $R$
  generally depends on $n$). 

 We show by downwards induction on $0\leq r \leq R$ that there exists a constant $\beta \geq 1$ (depending only on $G_\lambda$) such that the  following holds. If $(1^n,M) \in \cL[2^{t}]$, then   $\omega^*_{\Lambda_r(n),n} = 1$.
  Otherwise, %
  \begin{equation}\label{eq:comp-sound}
    \omega^*(\Game_{n}) \leq 1 - \frac{\delta^{\alpha^{R-r}}}{\Lambda_{R}(n)^{\beta \alpha^{R-r}}
      \cdots \Lambda_{r+1}(n)^{\beta \alpha}}\;.
  \end{equation}
	Note that the case $r=0$ implies item 2. of the proposition. First, the completeness statement shows that if $(1^n,M) \in \cL[2^{t}]$, then $\omega^*(\Game_{n}) = \omega^*_{\Lambda_0(n),n}   = 1$. Second, the soundness statement~\eqref{eq:comp-sound} implies that there exists universal constants $c,C > 0$ depending only on
  $\alpha,\beta,\delta$ such that $\omega_{n,n}^* \leq 1 - C\Lambda_R(n)^{-c} = 1 - Ct(n)^{-c}$. %
  
For the base case $r=R$, note that 
 on input $1^{n}$ the Turing machine $T$
  halts in $t(n) \leq \Lambda_R(n)$ steps. Thus the game
  $\Game_{\Lambda_R(n),n}$ is the game associated with the Natarajan-Vidick verifier
  $\cV_{NV}(M,\Lambda_{R(n)})$ (Theorem~\ref{thm:nv}). Suppose that $(1^n,M) \in \cL[2^{t}]$.
  This implies that $(1^{\Lambda_R(n)},M) \in \cL[2^n]$.\footnote{Note: the
    ``$2^n$'' inside $\cL[\cdot]$ is a variable that is different from the
    $n$ used to specify the instance $(1^{\Lambda_R(n)},M)$.}
  By Theorem~\ref{thm:nv}, $\omega_{\Lambda_R(n),n}^* = 1$.
 Otherwise, if $(1^n,M) \notin \cL[2^{t}]$, then we have that
  $\omega_{\Lambda_R(n),n}^* < 1 - \delta$.

  Now suppose $r < R$. Then the Turing machine $T$ does not halt on input
  $1^{n}$ in $\Lambda_{r}(n)$ steps. Therefore, $\cV_{RC}(\Lambda_{r}(n),n)$ executes the verifier $\cV^{\compr}_{G_\lambda,\Lambda_r(n)}$, where $G_{\lambda}$ is the GTM specified in Figure~\ref{fig:vrc}, with $\lambda = (n,M,T,G_R)$. 
  In turn, the protocol circuit $\CKT(G_\lambda,2^{\Lambda_r(n)}) = \CKT(G_\lambda,\Lambda_{r+1}(n))$ corresponds to the game $\Game_{\Lambda_{r+1}(n),n}$.
  Thus it follows from Theorem~\ref{thm:compression} that
  \[
    1 - \frac{1 - \omega^*_{\Lambda_{r+1}(n),n}}{\poly(\Lambda_{r+1}(n))} \leq
    \omega^*_{\Lambda_r(n),n} \leq 1 - \Paren{\frac{1 -
      \omega^*_{\Lambda_{r+1}(n),n}}{\poly(\Lambda_{r+1}(n))}}^\alpha\;,
  \]
  for some polynomial $\poly(\cdot)$ that depends only on $G_\lambda$ and not $r$ or $n$. Using the induction hypothesis~\eqref{eq:comp-sound}, this completes the induction step. 
\end{proof}

\subsection{An alternate proof of the undecidability of nonlocal
  games}
\label{sec:undecidable}

In this section we give an alternate proof that the problem of distinguishing
between the cases when a nonlocal game has value equal to $1$, or when
it has value strictly less than $1$, is
undecidable~\cite{slofstra2016tsirelson,slofstra2017set}.
This result was stated as Theorem~\ref{thm:main-undecidable} in the
introduction. Let $M$ be an arbitrary Turing machine, and $G$ a GTM. Consider the family of verifiers $\{\cV_{Halt}(n,M,G)\}$ described in Figure~\ref{fig:halt}.

\begin{figure}[H]
  \centering
  \begin{mdframed}[style=figstyle]
    \ul{Verifier name:} $\cV_{Halt}(n,M,G)$: \\
    \ul{Description of input:} $M$ is a deterministic Turing machine, and $G$ is a GTM that takes input $(n,t,M)$.
    \begin{enumerate}
    \item Run $M$ on input $0$ for $n$ steps. If it halts in this time, then reject.
    \item Otherwise, execute the verifier $\cV^{\compr}_{G_M,n}$ from Theorem~\ref{thm:compression} where $G_M(n,t) = G(n,t,M)$.
    \end{enumerate}
  \end{mdframed}
  \caption{The verifier $\cV_{Halt}$}
  \label{fig:halt}
\end{figure}

It follows from the definition and Theorem~\ref{thm:compression} that the verifiers $\Set{\cV_{Halt}(n,M,G)}$ can be uniformly generated from $(1^n,M,G)$ by a Turing machine $H$. By Lemma~\ref{lem:gtm-combine}, there exists a GTM $G_H$ that takes input $(n,t,M,G)$ and outputs the $t$-th gate of the protocol circuit corresponding to the verifier $\cV_{Halt}(n,M,G)$. Define the verifier $\cV_{Halt}(n,M) = \cV_{Halt}(n,M,G_H)$.

\begin{theorem}\label{thm:undecidable}
  There exists universal constants $c,C > 0$ such that for any
  deterministic Turing machine $M$ there exists a $15$-prover nonlocal
  game $\Game_M$, that can be computed from the description of $M$,
  such that the following hold.
  \begin{enumerate}
  \item Suppose that $M$ halts on input $0$ in time $T$, for some $T\geq 0$. Let $R$ be the largest integer such that $T > \Lambda_{R}(1)$. Then $\omega^*(\Game_M) \leq
    1 - C\Lambda_{R}(1)^{-c}$. %
  \item Suppose that  $M$ does not halt on input $0$. Then $\omega^*(\Game_M) = 1$.
    Furthermore, there is a universal constant $\eta > 0$ such that any strategy $\strat$ for $\Game_M$ such that
    $\omega^*_\strat(\Game_M) \geq 1 - \eps$ for some $\eps \geq 0$ requires local dimension at
    least $2^{\Omega(\eps^{-\eta})}$.
  \end{enumerate}
\end{theorem}

Theorem~\ref{thm:undecidable} implies that if there were a
Turing machine $A$ that when given a description of a nonlocal game
$\Game$, decides if $\omega^*(\Game) = 1$, then $A$ could be used to solve
the Halting Problem. Thus there is no such Turing machine $A$.

\begin{proof}
  Fix a deterministic Turing machine $M$.
  For any integer $n\geq 1$ let $\Game_{n}$ denote the $7$-prover game specified by
  $\cV_{Halt}(n,M)$, and let $\omega^*_{n} = \omega^*(\Game_{n})$. It follows from Theorem~\ref{thm:compression} that
  \begin{equation}\label{eq:halt-1a}
    1 - \frac{1 - \omega^*_{2^n}}{p(2^n)+1} \leq \omega^*_{n} \leq 1 - \Paren{\frac{1
      - \omega^*_{2^n}}{\beta\, p(2^n)}}^\alpha\;,
  \end{equation}
  for some universal constants $\alpha \geq 1$, $\beta >0$ and some polynomial $p$ that depends only on $G_M$. 
  
	We first show the completeness statement, item 2. in the theorem. Suppose that 
	$M$ does not halt on input $0$. By an immediate induction it follows from the first inequality in~\eqref{eq:halt-1a} that for any $r \geq 0$,  
	\[ 1- \omega^*_{1} \,\leq\, \frac{ 1 - \omega^*_{\Lambda_r(1)}}{\big(p(\Lambda_1(1))+1\big)\cdots \big(p(\Lambda_r(1))+ 1\big)}\;,\]
	from which it follows, by taking the limit $r\to\infty$, that necessarily $\omega^*_1 = 1$.

	Next we show the soundness statement, item 1. in the theorem. Suppose that $M$ halts in time $T$, and let $R$ be the largest integer such that $\Lambda_R(1) < T$.
 Then $\omega^*_{\Lambda_{R+1}(1)} = 0$. By downwards induction it
  follows from the second inequality in~\eqref{eq:halt-1a} that there exists constants $c',C'>0$ that depend on $G_M$ such that
  \[
   1-\omega^*_1 \,\geq \,  \frac{1}{p(\Lambda_R(1))^{\alpha^{R}} p(\Lambda_{R-1}(1))^{        \alpha^{R-1}} \cdots p(\Lambda_1(1))^\alpha} \,\geq\,  C' \Lambda_R(1)^{-c'}\;.
  \]

  To conclude, as in the proof of Theorem~\ref{thm:main2} we apply the
  compression result from~\cite{ji2017compression} to $\Game_1$ to obtain a
  $15$-prover nonlocal game $\Game_M$ such that $\omega^*(\Game_M) = 1$ if
  $M$ does not halt, and otherwise $\omega^*(\Game_M) \leq 1 - \Omega ( (1 -
  \omega^*_1)^\alpha) < 1 - C\Lambda_R(1)^{-c}$ for universal constants $c,C >0$.
  Note that the game $\Game_M$ is ``constant sized'' (there is no
  asymptotic parameter here).

  The ``furthermore'' part of the theorem follows from the fact that
  any strategy $\strat$ for $\Game_M$ that is accepted with
  probability at least $1 - \eps$ is $\delta$-isometric to a strategy
  $\strat'$ such that the provers' shared state is a history state of a
  strategy $\strat_1$ in $\Game_1$ that succeeds with probability $1 -
  \delta$ for $\delta = \eps^{c}$.
  (This follows from the analysis of the compression result
  of~\cite{ji2017compression}; details omitted.)
  By part 4 of Theorem~\ref{thm:compression}, the strategy $\strat_1$ must have
  local dimension at least $2^{\Omega(\delta^{-\eta'})}$ for some universal
  constant $\eta' > 0$.
  Thus $\strat$ must have local dimension
  $2^{\Omega(\eps^{-\eta})}$ for some universal constant $\eta >
  0$.
\end{proof}

\newcommand{\hct}{\diamondsuit}

\section{Improving the Compression Theorem?}
\label{sec:improving}

We explore the question of whether our compression theorem,
Theorem~\ref{thm:compression}, is optimal in terms of the trade-off
that it provides between ``compression in game size'' versus
``compression of the game value towards $1$''.
Recall that, given a GTM $G$ for a family of games $\set{\Game_N}$,
the theorem yields a family of games $\set{\Game_{n}^\compr}$ such
that for all $n$ and $N = 2^n$ we have that if $\omega^*(\Game_N) = 1$,
then $\omega^*(\Game_{n}^\compr) = 1$, but otherwise $\omega^*(\Game_{n}^\compr) \leq 1 - \Paren{
  \frac{1 - \omega^*(\Game_N)}{\poly(N)}}^\alpha$.
The compression of the game size is exponential, from $N$ to
$\poly(\log N)$, and the value of $\Game_{n}^\compr$ is closer to $1$ by a
factor $\poly(N)$.
But suppose that there was a \emph{Hypothetical Compression Theorem (HCT)}  with
a better trade-off.

\begin{conjecture}[Hypothetical Compression Theorem]
\label{conj:hct}
Given a GTM $G$ for a family of games
$\set{\Game_N}$, there exists a family of verifiers $\{\cV_n^\hct\}$ that is uniformly generated from $(1^n,G)$, and a monotonically increasing function
$g(n)=2^{o(n)}$, such that the following hold. For any integer $n\geq 0$, the game
$\set{\Game_{n}^{\hct}}$ associated with $\cV_n^\hct$ has constant answer size, and  for  $N = 2^n$ we have that if $\omega^*(\Game_N) = 1$, then $\omega^*(\Game_n^\hct) = 1$, and in all cases,
\begin{equation}\label{eq:better-compression}
	\omega^*(\Game_{n}^{\hct}) \leq 1 - \Paren{ \frac{1 - \omega^*(\Game_N)}{g(n)}}^\alpha\;.
\end{equation}
\end{conjecture}

(Note that when $g(n) = 2^{O(n)}$~\eqref{eq:better-compression} recovers the upper bound of Theorem~\ref{thm:compression}.)
We show that Conjecture~\ref{conj:hct} implies ``constant-gap analogues'' of
Theorem~\ref{thm:main2} and Theorem~\ref{thm:undecidable}:
first, $\MIP*$ would contain all computable languages.
Second, $\MIP*$ would contain undecidable languages. The undecidability of $\MIP^*$, in turn, implies a negative answer to a multipartite generalization of Tsirelson's problem, which is a open question about the relation between the commuting operator and tensor product models for quantum correlations.

\medskip
\vspace{10pt}

The main tool we need to derive these consequences is a \emph{hardness amplification procedure} for
$k$-prover ENL games.
This is a general transformation on ENL games that preserves the game
value if the original game has value $1$, but otherwise decreases it.

We call an ENL game and its associated verifier \emph{nonadaptive} if the questions to the provers
are chosen before the measurement of the provers' first message.
The ENL games and verifiers obtained from Theorem~\ref{thm:compression}
are nonadaptive.
The following hardness amplification procedure is established
in~\cite{bavarian2017hardness}.

\begin{theorem}[Hardness amplification via
  anchoring~\cite{bavarian2017hardness}]
  \label{thm:anchoring}
  Fix an integer $k \geq 2$.
  For every integer $r$ there exists a transformation $\cA_r$ on verifiers 
  such that for any $k$-prover nonadaptive verifier $\cV$ for an ENL game $\Game$
  the following holds:
  \begin{enumerate}
  \item $\cA_r(\cV)$ is a $k$-prover verifier for a nonadaptive ENL game $\Game'$ such that
   \[
      (\omega^*(\Game))^r \leq \omega^*(\Game') \leq \Paren{ 1 - (1 -
        \omega^*(\Game))^c }^{\nu_\Game r}\;,
    \]
    where $\nu_\Game$ is a positive real that depends on the number of
    provers $k$ and the length of answers in $\Game$, and $c \geq 1$ is
    a universal constant. 
  \item The size of $\cA_r(\cV)$ is $O(r)$ times the size of $\cV$.  

  \end{enumerate}
  Furthermore, if $\{\cV_{n,\lambda}\}$ is a family of verifiers uniformly generated from $(1^n,\lambda)$, the family of verifiers $\{\cA_r(\cV_{n,\lambda})\}$ can be uniformly generated from $(1^n,1^r,\lambda)$.
\end{theorem}

Strictly speaking, the hardness amplification result
of~\cite{bavarian2017hardness} is stated for nonlocal games, in which
the verifier is completely classical.
However, the results extend to nonadaptive ENL games because the
verifier's initial measurement can be modeled as the action of an
``honest'' prover.\footnote{We believe that the nonadaptive condition
  can be omitted from the statement of Theorem~\ref{thm:anchoring},
  but we leave this for future work.}

\subsection{Consequence $1$: $\MIP^*$ contains all computable languages}
\label{sec:cons1}

A language $L$ is \emph{computable} if there exists a Turing machine $M$ that, for all inputs $x \in \{0,1\}^*$, accepts if $x \in L$ and otherwise rejects. In particular, $M$ halts on all inputs.

We introduce a verifier $\hat{\cV}_{RC}$, described in Figure~\ref{fig:vrcplus}, and analyze it in a manner similar to the verifier ${\cV}_{RC}$ considered in Section~\ref{sec:vrc}. In this section, we use $c$ and $\nu$ to denote the constants from
Theorem~\ref{thm:anchoring} that correspond to games with at most $7$
provers and the answer length provided by Conjecture~\ref{conj:hct}. We also let $\alpha$ and $g(n)$ be the constant and subexponential function $g(n)$ from Conjecture~\ref{conj:hct}. We let $\cV^{\hct}_{n}$ denote the verifier of the game $\Game^{\hct}_{n}$. 

\begin{figure}[H]
  \centering
  \begin{mdframed}[style=figstyle]
    \ul{VTM name:} $\hat{\cV}_{RC}(n,M,G)$: \\
    \ul{Description of input:} $n > 0$ is an integer, $M$ is a
    deterministic Turing machine, and $G$ is a GTM that takes input $(n,t,\lambda)$.
    \begin{enumerate}
    \item Run $M$ on input $0$ for $n$ steps.
      If $M$ accepts in that time, accept.
      If $M$ rejects in that time, reject. 
    \item Otherwise, if $M$ does not halt in $n$ steps, perform the following. Let $\lambda = (M,G)$, and $G_\lambda(n,t) = G(n,t,\lambda)$. Let $n'$ be the largest integer less than $n$ such that \\ $ (2g(n'))^{\alpha c} \cdot q(n')/\nu \leq n$, where $q(n')$ is the size of $\cV^{\hct}_{G_\lambda,n'}$. Let $r = (2g(n'))^{\alpha c }/\nu$. Execute the verifier $\cA_r(\cV^{\hct}_{G_\lambda,n'})$.
	\end{enumerate}
  \end{mdframed}

  \caption{The verifier $\hat{\cV}_{RC}$}
  \label{fig:vrcplus}
\end{figure}

It follows from the definition that the family of verifiers $\{\hat{\cV}_{RC}(n,M,G)\}$ can be uniformly generated by some Turing machine $R$. (This is the reason for the choice of the parameter $n'$, which guarantees that the size of the verifier $\cA_r(\cV^{\hct}_{G_\lambda,n'})$ is at most $n$.)

Let $R$ be a Turing machine that on input $(1^n,M,G)$ generates the verifier $\hat{\cV}_{RC}(n,M,G)$ in polynomial time. By Lemma~\ref{lem:gtm-combine}, there exists a GTM $G_R$ that takes input $(n,t,M,G)$ and outputs the $t$-th gate of the protocol circuit corresponding to the verifier $\hat{\cV}_{RC}(n,M,G)$. 

\begin{proposition} \label{prop:vrcplus} Suppose Conjecture~\ref{conj:hct} is true.
  Let $M$ be a deterministic Turing machine that halts on input $0$. Then the family of  verifiers $\{\hat{\cV}_{RC}(n,M,G_R)\}$ can be uniformly generated from $(1^n,M,G_R)$. Furthermore, %
  the $7$-prover ENL game $\Game_m$ associated with $\hat{\cV}_{RC}(m,M,G_R)$, where $m$ is the smallest integer larger than $(2g(1))^{\alpha c} q(1)/\nu$,\footnote{The justification for this choice of $m$ is to ensure that for all $n \geq m$, the integer $n'$ chosen in step 2. of the definition of $\hat{\cV}_{RC}(n,M,G)$ (Figure~\ref{fig:vrcplus}) is well-defined and at least $1$.} satisfies
  \begin{align*}
    \omega^*(\Game_{m}) &= 1  \quad &\text{if $M$ accepts on input $0$,} \\
    \omega^*(\Game_{m}) &\leq 1/2 \quad &\text{if $M$ rejects on input $0$.}
  \end{align*}
\end{proposition}

\begin{proof}
  Let $M$ and $m$ be as in the theorem statement. For any integer $n\geq 1$, define the verifier $\hat{\cV}_{RC}(n) = \hat{\cV}_{RC}(n,M,G_R)$. Let $\Game_{n}$ denote the $7$-prover ENL game specified by $\hat{\cV}_{RC}(n)$, and let $\omega^*_{n} = \omega^*(\Game_{n})$. 
  Let $R$ be the smallest integer such that $\Lambda_R(m)$ is greater than the running time of $M$ (which is well-defined since $M$ halts on input $0$). 

  If $M$ accepts on input $0$, then $\omega^*(\Game_{m})
  = 1$; this follows by induction on $R$, using similar reasoning as in the proof of
  Proposition~\ref{prop:vrc}.
  The remaining case is that $M$ does not accept on input
  $0$.
  By definition, for all $N \geq \Lambda_R(m)$, we have that $\omega_{N}^* = 0$.
  We show by downwards induction that $\omega_{N}^* \leq 1/2$ for all integers $N \geq
  m$.
  Assume the inductive hypothesis holds for all $N \geq N_0 + 1$ for
  some $N_0 < \Lambda_R(m)$. 
  Since $N_0 < \Lambda_R(m)$, $M$ does not halt on input $0$ in $N_0$
  steps.
  Therefore, the verifier in the game $\Game_{N_0}$ executes
  $\cA_r(\cV^{\hct}_{G_\lambda,N_0'})$ where $\lambda$, $G_\lambda$, $N_0'$, and $r$ are defined in Figure~\ref{fig:vrcplus}. Let $N_1 = 2^{N_0'}$. Since $g$ is a monotonically increasing but subexponential function, we have $N_0' = \omega(\log N_0)$ and therefore $N_1 > N_0$. Therefore by the induction hypothesis it follows that $\omega^*_{N_1} \leq 1/2$. Using Conjecture~\ref{conj:hct} and Theorem~\ref{thm:anchoring} together, %
  \[
    \omega^*_{N_0} \,\leq\, \Paren{ 1 - \Paren{ \frac{1 -
          \omega^*_{N_1}}{g(N_0')}}^{\alpha c}}^{\nu r}\;.
  \]
 Using that $\omega^*_{N_1} \leq 1/2$ and the choice of $r$ made in Figure~\ref{fig:vrcplus}, 
we get that $\omega^*_{N_0} \leq 1/e \leq 1/2$. This completes the induction and shows that $\omega^*_{m} \leq 1/2$, as
  desired.
\end{proof}

\begin{corollary}
  \label{cor:hct_dtime}
  Suppose Conjecture~\ref{conj:hct} is true. Then $\MIP^*$ with constant completeness-soundness gap contains all computable languages. In other words, we have $\class{R} \subseteq \MIP^*$ where $\class{R}$ is the set of all recursive languages.
\end{corollary}
\begin{proof}
  Let $L$ denote a computable language. This means that there exists a deterministic Turing machine $M$ such that for all inputs $x \in \{0,1\}^*$, $M(x)$ accepts if $x \in L$, otherwise $M(x)$ rejects. Let $M_x$ denote the Turing machine $M$ with input $x$ hardwired and
  otherwise ignores its input tape. Observe that $M_x$ halts in finite time.
	
  There exists a polynomial time deterministic Turing machine $A$ that on input $x$
  performs the following. First, $A$ computes a description of the
  $7$-player ENL game $\Game_{m,M_x}$ given by
  Proposition~\ref{prop:vrcplus}, with $m$ chosen as in the proposition statement. Let $n = |x|$.
  This game has the property that if $M_x$ accepts,
  then $\omega^*(\Game_{m,M_x}) = 1$, otherwise $\omega^*(\Game_{m,M_x}) \leq
  1/2$.
  Furthermore the size of the verifier of $\Game_{m,M_x}$ is
  $\poly(n,|M|)$. Next, the ENL game $\Game_{m,M_x}$ is converted to a nonlocal game by using the compression result of~\cite{ji2017compression}; this result gives an efficient reduction from the description of the verifier of $\Game_{m,M_x}$ to the verifier of a $15$-player nonlocal game $\Game_{m,M_x}'$ whose value satisfies
  \[
    \omega^*(\Game_{m,M_x}) \leq \omega^*(\Game_{m,M_x}') \leq 1 - \Paren{ \frac{1 -
        \omega^*(\Game_{m,M_x})}{\poly(n)}}^\alpha. %
  \]
  Finally, $A$ computes a description of the game $\Game_{m,M_x}''$, in which the hardness amplification procedure $\cA_s$ of Theorem~\ref{thm:anchoring} is applied to the verifier of $\Game_{m,M_x}'$ for some $s = \poly(n)$. The verifier of $\Game_{m,M_x}''$ still has $\poly(n)$ size, but now if $\omega^*(\Game_{m,M_x}') \leq 1 - 1/\poly(n)$, then $\omega^*(\Game_{m,M_x}'') \leq 1/2$ (provided that $s$ is a large enough polynomial). 
  
  Thus on input $x$ the Turing machine $A$ returns the description of a nonlocal game with a $\poly(n)$-sized verifier, such that if $x$ is accepted by $M$, the value of the game is $1$; otherwise, the value is at most $1/2$.
  This shows that $L$ has a one-round $\MIP*$ proof system with $15$
  provers and constant completeness-soundness gap.
\end{proof}

\subsection{Consequence $2$: $\MIP^*$ contains undecidable languages}
\label{sec:cons2}

In this section we show that Conjecture~\ref{conj:hct} implies that $\MIP*$ contains undecidable languages. We show this directly: instead of reducing the halting problem to the problem of
approximating the value of a nonlocal game, we show that
there is no Turing machine that can approximate the value of a nonlocal
game to within constant additive error.
Thus $\MIP*$ contains undecidable languages: namely, the
(promise) language $L_{c,s}$ whose YES instances consist of all
nonlocal games whose value is at least $c$, and whose NO instances
consists of all nonlocal games whose value is at most $s$, for some
constants $0 \leq s < c \leq 1$.

In Figure~\ref{fig:vhaltplus} we define a VTM $\hat{\cV}_{undec}$ that
is  differs slightly from the VTM $\cV_{Halt}$ analyzed in
Section~\ref{sec:undecidable}.
Whereas the games $\set{\Game_{n,M}}$ specified by $\cV_{Halt}$ have
value $1$ or less than $1/2$ depending on whether $M$ halts or not,
the games $\set{\Game_{n,M}}$ specified by $\hat{\cV}_{undec}$ have
value $1$ or less than $1/2$ depending on whether $M$ accepts or
rejects (when given its own description as input).
There is no guarantee on the value of the
game $\Game_{n,M}$ if $M$ does not halt.

In Figure~\ref{fig:vhaltplus}, $c$, $\nu$ and $\alpha$ are the constants introduced in
Section~\ref{sec:cons1}.

\begin{figure}[H]
  \centering
  \begin{mdframed}[style=figstyle]
    \ul{VTM name:} $\hat{\cV}_{undec}(n,M,G)$: \\
    \ul{Description of input:} $M$ is a deterministic Turing machine.
    \begin{enumerate}
    \item Run $M$ on input $M$ (i.e.
      the input to $M$ is the description of $M$ itself) for $n$ steps.
      If $M$ halts and accepts, then accept.
      If $M$ halts and rejects, then reject.
    \item If $M$ does not halt within $n$ steps, then perform the following. Let $\lambda = (M,G)$ and $G_\lambda(n,t) = G(n,t,\lambda)$. Let $n'$ be the largest integer such that $(2g(n'))^{\alpha c} q(n')/\nu \leq n$, where $q(n')$ is the size of $\cV^{\hct}_{G_\lambda,n'}$. Let $r = (2g(n'))^{\alpha c}/\nu$.
    Execute $\cA_r(\cV^{\hct}_{G_\lambda,n'})$.
    \end{enumerate}
  \end{mdframed}

  \caption{The verifier $\hat{\cV}_{undec}$}
  \label{fig:vhaltplus}
\end{figure}

It follows from the definition that the family of verifiers $\{\hat{\cV}_{undec}(n,M,G)\}$ can be uniformly generated by a Turing machine $H$. By Lemma~\ref{lem:gtm-combine}, there exists a GTM $G_H$ that takes input $(n,t,M,G)$ and returns the $t$-th gate of the protocol circuit corresponding to the verifier $\hat{\cV}_{undec}(n,M,G)$. Define the verifier $\hat{\cV}_{undec}(n,M) = \hat{\cV}_{undec}(n,M,G_R)$. Let $\Game_{n,M}$ denote the $7$-prover ENL game executed by $\hat{\cV}_{undec}(n,M)$.

\begin{proposition} \label{prop:vhaltplus} Suppose Conjecture~\ref{conj:hct} is true.
  Let $M$ be a deterministic Turing machine.
  Then for all $n$,
  \begin{align*}
    \omega^*(\Game_{n,M}) &= 1  \quad & \text{if $M$ accepts on input $M$,} \\
    \omega^*(\Game_{n,M}) &\leq 1/2 \quad &\text{if $M$ rejects on input $M$}.
  \end{align*}
\end{proposition}

Note that Proposition~\ref{prop:vhaltplus} does not specify
the value of $\Game_{n,M}$ in the case that $M$ does not halt on input
$M$.
An ideal version of
Proposition~\ref{prop:vhaltplus} would state that $\omega^*(\Game_{n,M}) = 1$ if $M$
does not halt, and $\omega^*(\Game_{n,M}) \leq 1/2$ if $M$ halts, similarly
to the conclusion of Theorem~\ref{thm:undecidable}.
We are able to obtain a guarantee on the value of $\Game_{n,M}$ when $M$
does not halt in Theorem~\ref{thm:undecidable} because of special
properties of the games specified by $\cV^{\compr}$ (namely, when the
size $N$ of the verifier increases, the value of the game goes to $1$,
no matter what game is being compressed).
However, the games specified by Conjecture~\ref{conj:hct} may not satisfy this property; the only guarantee is
that $\omega^*(\Game_{G_\lambda,n}^{\hct}) = 1$ if $\omega^*(\Game_{2^n,M}) = 1$, and
otherwise $\omega^*(\Game_{G_\lambda,n}^{\hct})$ is upper-bounded by some function of
$\omega^*(\Game_{2^n,M})$.

  The proof of Proposition~\ref{prop:vhaltplus}  is essentially the same as the proof of
  Proposition~\ref{prop:vrcplus}, and we omit it. 
We state a corollary showing  that it is possible to construct a family of nonlocal games with similar properties as the ENL games from
Proposition~\ref{prop:vhaltplus}.

\begin{corollary}\label{cor:vhaltplus}
  Suppose Conjecture~\ref{conj:hct} is true.
  Let $M$ be a deterministic Turing machine.
  There exists a $15$-prover nonlocal game $\Game_M$ such that
  \begin{align*}
    \omega^*(\Game_{M}) &= 1  \quad & \text{if $M$ accepts on input $M$,} \\
    \omega^*(\Game_{M}) &\leq 1/2 \quad &\text{if $M$ rejects on input $M$}.
  \end{align*}
  Furthermore, the description of the verifier of $\Game_M$
  is computable from $M$.
\end{corollary}

\begin{proof}
  The proof is essentially the same as the proof of
  Corollary~\ref{cor:hct_dtime}. The only additional step is to apply the hardness
  amplification procedure $\cA_r$ from Theorem~\ref{thm:anchoring} to the game returned by the Turing machine $A$, for $r = \poly(|M|)$, to amplify the gap from $1$ vs.
  $1 - 1/\poly(|M|)$ to $1$ vs.
  $1/2$.
\end{proof}

\begin{theorem}\label{thm:undecidable3}
  Suppose Conjecture~\ref{conj:hct} is true.
  Then there is no deterministic Turing machine $A$ that, given as input the
  description of the verifier circuits of a nonlocal game $\Game$,
  decides whether $\Game$ has value at least $2/3$ or less than $1/3$,
  promised that one is the case.
\end{theorem}

\begin{proof}%
  Suppose for contradiction that there exists such a
  Turing machine $A$.
  Consider the following deterministic Turing machine $M$. $M$ expects as input an $X$, which is the description of a deterministic
  Turing machine.
  The Turing machine $M$ first computes the descriptions of verifier
  circuits for two nonlocal games $\Game_X$ and $\Game_X^r$.
  The first game, $\Game_X$, is the game given by
  Corollary~\ref{cor:vhaltplus}.
  The second game, $\Game_X^r$, is the nonlocal game that results from
  applying the hardness amplification procedure $\cA_r$ from Theorem~\ref{thm:anchoring} to $\Game_X$,
  where $r$ is an integer such that $(1 - (1/3)^c)^{\nu r} \leq
  1/3$. 
  Here, $c$ and $\nu$ are the constants given by
  Theorem~\ref{thm:anchoring}.
  Thus
  \begin{equation}
    \label{eq:parrep}
    \omega^*(\Game_X^r) \leq (1 - (1 - \omega^*(\Game_X))^c)^{\nu r}\;.
  \end{equation}
  Furtherore, if $\Game_X$ has value $1$, then $\Game_X^r$ has value
  $1$.

  Having computed the descriptions of the games $\Game_X$ and
  $\Game_X^r$, the Turing machine $M$ executes two instances of $A$ in
  parallel (for example, by interleaving the executions of $A$), where
  one instance is executed on the description of $\Game_X$, and the other on
  $\Game_X^r$.
  If one of the instances halts first with output bit $a$, then $M$
   rejects if $a = 1$ and accepts if $a = 0$.
  However, $M$ may not halt (if both instances of $A$ don't halt).

Observe that at most one of the games $\Game_X$, $\Game_X^r$ has value that is
  greater than $1/3$ and less than $2/3$. Indeed, suppose the value of both games were in that range.
  In particular, we have $1/3 < \omega^*(\Game_X) < 2/3$.
  However, by~\eqref{eq:parrep} and our choice of $r$, this implies
  that $\omega^*(\Game_X^r) \leq 1/3$, a contradiction.

  Thus at least one instance of $A$ halts, because by definition
  $A$ correctly decides whether a given input game $\Game$ has
  value at least $2/3$ or at most $1/3$.
  Therefore, $M$ always halts, on all inputs $X$.

  Now we analyze $M$, when given input $M$.
  By definition of $\Game_M$, if $M$ accepts input $M$, then
  $\omega^*(\Game_M) = \omega^*(\Game_M^r) = 1$.
  In this case, both instances of $A$ accept, in which case $M$
  rejects, which is a contradiction.

  On the other hand, if $M$ rejects input $M$, then both
  $\omega^*(\Game_M)$ and $\omega^*(\Game_M^r)$ have value at most $1/3$, in
  which case both instances of $A$ reject, in which case $M$
  accepts, which is a contradiction.\footnote{The reader would be
    justified in asking why we needed to consider two games in
    the first place.
    If we only considered $\Game_M$, then we wouldn't be able to
    conclude that $\Game_M$ has value either greater than $2/3$ or at
    most $1/3$, and thus $M$ could in principle run forever.
    By defining $M$ in this way we force the resulting game
    $\Game_M$ to satisfy the promise of $A$.}

  Therefore such a Turing machine $A$ does not exist.
\end{proof}

Thus Theorem~\ref{thm:undecidable3} implies that the language
$L_{c,s}$ for $c = 2/3$ and $s = 1/3$ is undecidable, which implies
that $\MIP*$ contains undecidable languages. We end by formulating the following corollary, that relates Conjecture~\ref{conj:hct} to a famous problem in quantum information, \emph{Tsirelson's problem}. 
To state the corollary, we introduce the notion of a \emph{$k$-partite, $n$-input, $m$-output correlation}, which is a $k$-tensor $C$ of complex numbers, with size $nm\times \cdots \times nm = (nm)^k$, where $k,n,m$ are arbitrary integers. We say that a correlation $C$ is \emph{achievable in the tensor product model} if there exists finite-dimensional Hilbert spaces $\cH_1,\ldots,\cH_k$, a state $\ket{\psi} \in \cH_1\otimes \cdots \otimes \cH_k$, and for every $\ell\in\{1,\ldots,k\}$ and $i\in\{1,\ldots,n\}$ a POVM $\{A_{\ell,i}^a \}_{a\in\{1,\ldots,m\}}$ acting on $\cH_\ell$, such that for all $i_1,\ldots,i_k\in\{1,\ldots,n\}$ and $a_1,\ldots,a_k\in\{1,\ldots,m\}$, we have
$$C(i_1,a_1,\ldots,i_k,a_k) = \bra{\psi} A_{1,i_1}^{a_1}\otimes \cdots \otimes A_{k,i_k}^{a_k} \ket{\psi}\;.$$
Similarly, we say that $C$ is \emph{achievable in the commuting operator model} if there exists a (possibly infinite-dimensional) Hilbert space $\cH$, a state $\ket{\psi} \in \cH$, and for every $\ell\in\{1,\ldots,k\}$ and $i\in\{1,\ldots,n\}$ a POVM $\{A_{\ell,i}^a \}_{a\in\{1,\ldots,m\}}$ acting on $\cH$ satisfying the commutativity condition $[A_{\ell,i}^a,A_{\ell',i'}^{a'}]=0$ for all $\ell\neq \ell'$ and $i,i',a,a'$, such that for all $i_1,\ldots,i_k\in\{1,\ldots,n\}$ and $a_1,\ldots,a_k\in\{1,\ldots,m\}$, we have
$$C(i_1,a_1,\ldots,i_k,a_k) = \bra{\psi} A_{1,i_1}^{a_1} \cdots  A_{k,i_k}^{a_k} \ket{\psi}\;.$$
We also measure the \emph{distance} between two correlations $C, C'$ as the sum of the absolute differences of their entries:
\[
	| C - C' | = \sum_{\substack{i_1,\ldots,i_k \\ a_1, \ldots, a_k}} \Abs{C(i_1,a_1,\ldots,i_k,a_k) - C'(i_1,a_1,\ldots,i_k,a_k)}.
\]

Tsirelson's problem (more precisely, the multipartite version of it) asks whether for every $k$, $n$, and $m$, for every $k$-partite, $n$-input, $m$-output correlation $C$ achievable in the commuting operator model, for every $\eps > 0$, there exists a $k$-partite, $n$-input, $m$-output correlation $C'$ achievable in the tensor product model such that $|C - C'| \leq \eps$. In other words, a positive answer to Tsirelson's problem would establish that correlations in the commuting operator model can be approximated arbitrarily well by correlations in the tensor product model. The next corollary shows that Conjecture~\ref{conj:hct} would yield a negative resolution of Tsirelson's problem.

\begin{corollary}\label{cor:connes}
  Suppose Conjecture~\ref{conj:hct} is true. Then there exists $\eps > 0$, integers $n,m > 0$, and a $15$-partite, $n$-input, $m$-output correlation $C$ that is achievable in the commuting operator model that has distance at least $\eps$ from any correlation $C'$ achievable in the tensor product model. 
\end{corollary}

\begin{proof}
Suppose not, i.e. any $15$-partite correlation $C$ achievable in the commuting operator model, for all $\eps>0$ there exists a correlation $C'$ achievable in the tensor product model such that $|C - C'| \leq \eps$. This implies that for any $15$-prover nonlocal game $\Game$, the entangled value of $\Game$ in the tensor product (denoted by $\omega^*_{tp}(\Game)$) and commuting operator models (denoted by $\omega^*_c(\Game)$) are equal: for every $\delta > 0$, let $\strat_c$ be a commuting operator strategy in a $15$-prover nonlocal game $\Game$ such that $\omega^*_c(\Game) \leq \omega^*_{\strat_c}(\Game) + \delta$. Then by our assumption, for all $\eps > 0$ there is a tensor product model strategy $\strat_{tp}$ such that 
\[
	\Abs{\omega^*_{\strat_{tp}}(\Game) - \omega^*_{\strat_c}(\Game)} \leq \eps.
\]
By taking $\eps = \delta$, for every $\delta > 0$ we have that there is a strategy $\strat_{tp}$ in the tensor product model  such that $\Abs{ \omega^*_{\strat_{tp}}(\Game) - \omega^*_c(\Game)} \leq 2\delta$. Since the entangled value in the tensor product model is defined as the supremum over tensor product model strategies, by taking $\delta \to 0$ we get that $\omega^*_c(\Game) = \omega^*_{tp}(\Game)$.

We provide an algorithm that decides if the value of a $15$-prover nonlocal game is larger than $2/3$, or at most $1/3$, promised that one is the case. The algorithm interleaves two procedures. The first procedure exhaustively searches for strategies in the tensor product model of increasing dimension, and with increasing accuracy. If this procedure returns a value that is larger than $1/2$, the algorithm halts and returns YES. A second procedure computes a non-increasing sequence of upper bounds by solving semidefinite programs obtained at increasing levels of the hierarchy introduced in~\cite{doherty2008quantum,navascues2008convergent}.  If this procedure returns a value that is smaller than $1/2$, the algorithm halts and returns NO. 

We show that this algorithm always halts, and always returns the correct decision. It is clear that the first procedure provides a non-decreasing sequence that converges to the value of the game in the tensor product model from below. Conversely, it is known that the second procedure provides a non-increasing sequence that converges to the value of the game in the commuting operator model from above. Since the values in both models coincide, this implies that the algorithm described in the previous paragraph always halts with the correct decision.

However, Conjecture~\ref{conj:hct} and Theorem~\ref{thm:undecidable3} implies that there is no such algorithm, a contradiction. Thus there is a correlation $C$ achievable in the commuting operator model that cannot be approximated arbitrarily well by correlations in the tensor product model.
\end{proof}

\appendix

\section{Succinct representation of uniform circuit families}
\label{app:gtm}

In this appendix we show that any uniformly generated family of circuits has a succinct description, in the sense of Section~\ref{sec:gtm}. First we introduce a generic method for constructing a circuit that implements the same computation as a Turing machine. Then, we show that any such circuit can be written in a regular form, that has a succinct description. Finally, we apply these two steps for the case of a Turing machine that specifies a family of circuits. 

\subsection{Simulation of a Turing machine with a quantum circuit}

A universal Turing machine simulator circuit is a quantum circuit $\TMSIM$ that, given as input the description of a Turing machine $M$, a positive integer time $T$, and a designated output tape for $M$, computes the
contents of the output tape after $M$ has been executed for $T$ steps.

\begin{lemma}\label{lem:tmsim}
For any integer $k\geq 1$ there exists a family of quantum circuits $\{\TMSIM_k(T)\}_{T\in\N}$ of
  size $\poly(T)$ such that the following hold for all $T\geq 1$. 
  \begin{enumerate}
  \item $\TMSIM_k(T)$ acts on registers $\sS$ (the \emph{Turing machine state register}), $\sM$ (the \emph{Turing machine specification register}), and $\sA_1, \ldots
    \sA_k$(the \emph{Turing machine tape registers}).
  \item Let $M$ be the classical description of a $k$-tape Turing machine and $a=(a_1,\ldots,a_k)$ be a $k$-tuple of strings of symbols for the $k$ tapes of $M$, such that each $a_i$ has length at most the size of $\sA_i$. Let $a'=(a'_1,\ldots,a'_k)$ be the contents of $M$'s tapes after it has been executed for $T$ steps, starting from the tape values specified by $a$. Then after the circuit $\TMSIM_k(T)$ has been executed on input
    $\ket{0}_{\sS} \otimes \ket{M}_{\sM} \otimes \ket{a}_{\sA_1 \cdots
      \sA_k}$, the registers $\sA_1,\ldots,\sA_k$ are in state $\ket{a'}_{\sA_1 \cdots \sA_k}$.
  \end{enumerate}
  Furthermore, there exists a deterministic Turing machine $\TMSIM$-$\DESC_k$ that
  on input $T$ and an integer $t$ in binary runs in polynomial time and returns a description of the $t$-th gate of $\TMSIM_k(T)$ when it exists, and a special failure symbol when it does not.
\end{lemma}

\begin{proof}
Fix an integer $k\geq 1$ and let $U$ be a universal $(k+1)$-tape Turing machine. When provided as input the description of a $k$-tape Turing machine $M$ and a number of steps $1^T$ on its first
  tape, and some values $a$ on the remaining $k$ tapes, $U$ performs the computation of $M$ on input $a$ for $T$ steps. Furthermore, $U$ runs in polynomial time, and we assume without loss of generality~\cite{pippenger1979relations} that $U$ is \emph{oblivious}: the movements of the head of $U$ are independent of its input.  Without loss of generality, each tape head of $U$ alternates between weeping left for $T$ steps and then right for $T$ steps, and the
  heads move in sequence (i.e., the first tape's head moves first,
  then the second tape's head moves, and so on).

The circuit $\TMSIM_k(T)$ is defined as follows. The
  register $\sS$ stores the state of the universal Turing machine $U$. The register $\sM$ stores the description of the $k$-tape Turing
  machine $M$. The registers $\{\sA_j\}$ store the contents of the work tapes
  of $M$.
	Each movement of the heads of $U$ is implemented by a layer in the
  circuit. The computation of the head transition function is
  computed in register $\sS$, which is connected via two-qubit gates
  to the corresponding locations in the registers $\sA_j$. (Due to the assumption that $U$ is oblivious, these locations only depend on the index of the layer in the circuit.)

The number of gates of $\TMSIM_k(T)$ is clearly polynomial, establishing item 1. in the lemma. Furthermore, item 2. holds by construction. 

For the ``Furthermore'' part of the lemma, note that the structure of each layer is identical, with the only difference being that the
  gates that cross between $\sS$ and the registers $\{\sA_j\} \cup
  \{\sM \}$ are different depending on which cells of the tapes are
  supposed to be read/written to at that layer. Using that $U$ is oblivious, the location of the $t$-th gate of $\TMSIM_k(T)$ can be computed in time polynomial in $t$.
\end{proof}

\subsection{Simulating regular circuits}

Analogously to the circuit $\TMSIM$ that simulates a Turing Machine, we
introduce the notion of a universal circuit $\CKTSIM$ that simulates an
arbitrary quantum circuit. For purposes of efficient description it is convenient to consider  \emph{regular circuits}, which are defined as follows. 

\begin{definition}
An $n$-qubit \emph{regular circuit} of size 
$s$ is specified by a sequence of gates $g_1,\ldots,g_s$ where each $g_i \in \{H,T\}$, and the set of qubits that the gate $g_i$ acts on only depends on the triple $(i,n,s)$, and can be computed in polynomial time from the triple $(i,n,s)$ specified in binary. (For consistency, the Hadamard gate is interpreted as a $3$-qubit gate $I\otimes H \otimes I$.) 
\end{definition}

We record the easy observation that every $n$-qubit circuit of size $s$ has an equivalent
regular circuit of size $\poly(n,s)$ as the following lemma. 

\begin{lemma}
  \label{lem:regularize}
  There exists a deterministic polynomial-time Turing machine that takes as input the
  description of a quantum circuit $C$ and outputs a regular quantum
  circuit $C'$ that implements the same unitary transformation as $C$ does. 
\end{lemma}

The next lemma establishes the existence of a simulation procedure for circuits analogous to the one shown for Turing machines in Lemma~\ref{lem:tmsim}. 

\begin{lemma}\label{lem:cktsim}
There is a family of quantum circuits $\{\CKTSIM_{n,s}\}_{n,s\geq 1}$ of size $\poly(n,s)$ such that the following hold. For any $n,s\geq 1$ the circuit $\CKTSIM_{n,s}$ acts on two registers $\sA$ (the \emph{circuit specification register}) and $\sB$ (the \emph{target register}), where $\sB$ has $n$ qubits. For any $C \in \{0,1\}^s$ and state $\ket{\theta}_{\reg{B}}$
  \[
    \CKTSIM_{n,s} \big(\ket{C}_{\sA} \otimes \ket{\theta}_{\sB}\big) \,=\,
    \ket{C}_{\sA} \otimes C \ket{\theta}_{\sB}\;,
  \]
  where $C$ is interpreted as the description of a regular $n$-qubit quantum circuit of size $s$.\\
  Furthermore, there exists a deterministic Turing machine $\CKTSIM$-$\DESC$ that
  on input $(n,s,t)$ runs in polynomial time and returns a description of the $t$-th gate of
  $\CKTSIM_{n,s}$ when it exists, and a special failure symbol when it does not. 
\end{lemma}

\begin{proof}
For $n,s\geq 1$ the circuit $\CKTSIM_{n,s}$ has $s$ layers, where
  the $i$-th layer applies either a Hadamard or a Toffoli gate, depending on $g_i$, on the appropriate qubits. The indices of those qubits can be
  computed in $\poly(\log n,\log s)$ time.
 \end{proof}

\subsection{Succinct representation of uniform families of circuits}
\label{sec:succinct}

\begin{lemma}\label{lem:succinct}
Let $\{C_n\}_{n\geq 1}$ be family of circuits that is uniformly generated by the Turing machine $M$. Then there exists a deterministic Turing machine $G$, that is computable from $M$, such that on input $(n,t)$, where both $n$ and $t$ are integer written in binary, $G$ runs in polynomial time and returns a description of the $t$-th gate of a regular circuit $C'_n$ that implements the same unitary transformation as $C_n$ (but uses additional ancilla registers). 

\end{lemma}

\begin{proof}
Without loss of generality assume the number of tapes used by $M$ is $k=3$, with an input tape, a work tape and an output tape. Let $p_M$ be a polynomial that bounds the running time of $M$. 
Let $n\geq 1$. We describe the circuit $C'_n$. The circuit first initializes ancilla registers for $\TMSIM$ (see Lemma~\ref{lem:tmsim}) as follows. The register $\sS$ contains the initial state of $M$. The register $\sM$ contains a description of $M$. The registers $\sA_1,\sA_2,\sA_3$ are empty, except that the register $\sA_1$ associated with the input tape contains the input $1^n$. The next step in the circuit $C'_n$ is to execute the circuit $\TMSIM_k(p_M(n))$ on these registers to obtain a description of $C_n$. Using Lemma~\ref{lem:regularize} we may without loss of generality assume that $C_n$ is regular. Finally, the last step in the circuit $C'_n$ it to execute the circuit $\CKTSIM$ on the register $\sA_3$ associated with the output tape of $M$, that contains the description of $C_n$ and plays the role of the circuit specification register, and the target register, that is identified with the register containing the input state to $C_n$. 

It is clear that $C'_n$ implements the same transformation as $C_n$. The existence of the Turing machine $G$ follows directly from the description of $C'_n$ and the existence of the Turing machines $\TMSIM$-$\DESC$ and $\CKTSIM$-$\DESC$ from Lemma~\ref{lem:tmsim} and Lemma~\ref{lem:cktsim} respectively. Specifically, from its input $(n,t)$, $G$ may efficiently determine which of its three phases (input preparation, $\TMSIM$, $\CKTSIM$) the $t$-th gate of $C'_n$ is associated with, and then compute the gate itself using the appropriate succinct description Turing machine. 
\end{proof}

\bibliography{recursive,quantum_pcp}
\bibliographystyle{alpha}

\end{document}